\documentclass[11pt]{article}
\usepackage{amsmath,amssymb,amsfonts,latexsym,graphicx,amsthm}
\usepackage{fullpage,color}
\usepackage{setspace}
\usepackage{epstopdf}
\usepackage{url}
\usepackage{hyperref}
\usepackage{cleveref}
\usepackage{comment}
\usepackage[linesnumbered,boxed,ruled,vlined]{algorithm2e}
\usepackage{framed}
\usepackage[shortlabels]{enumitem}
\usepackage{tikz}
\usetikzlibrary{backgrounds}
\usetikzlibrary{decorations.pathreplacing,calligraphy}
\usetikzlibrary{decorations.pathmorphing}

\usepackage[margin=1in]{geometry}

\newtheorem{theorem}{Theorem}[section]
\newtheorem{remark}[theorem]{Remark}
\newtheorem{lemma}[theorem]{Lemma}
\newtheorem{definition}{Definition}[section]
\newtheorem{corollary}[theorem]{Corollary}
\newtheorem{claim}[theorem]{Claim}
\newtheorem{observation}[theorem]{Observation}

        {\medskip}

\newcommand{\ball}{\mathsf{B}}
\newcommand{\balls}{\mathcal{B}}
\newcommand{\clusters}{\mathcal{B}}
\newcommand{\tree}{\mathcal{T}}
\newcommand{\paths}{\Pi}

\newcommand{\pathpart}{\textsf{PathPartition}\xspace}
\renewcommand{\epsilon}{\eps}
\newcommand{\brac}[1]{\left(#1\right)}

\begin{document}

\newcommand{\algline}{
	\rule{0.5\linewidth}{.1pt}\hspace{\fill}%
	\par\nointerlineskip \vspace{.1pt}
}
\newenvironment{tbox}{\begin{tcolorbox}[
		enlarge top by=5pt,
		enlarge bottom by=5pt,
		breakable,
		boxsep=0pt,
		left=4pt,
		right=4pt,
		top=10pt,
		boxrule=1pt,toprule=1pt,
		colback=white,
		arc=-1pt,
		]
	}
	{\end{tcolorbox}}


\newenvironment{proofof}[1]{\noindent{\bf Proof of #1.}}
{\hspace*{\fill}\stopproof}

\newenvironment{properties}[2][0]
{\renewcommand{\theenumi}{#2\arabic{enumi}}
	\begin{enumerate} \setcounter{enumi}{#1}}{\end{enumerate}\renewcommand{\theenumi}{\arabic{enumi}}}

\newif\ifnocomments


\ifnocomments

\newcommand{\znote}[1]{}

\else
\newcommand{\znote}[1]{\textcolor{red}{\sc{[ZT: #1]}}}

\fi


\newcommand{\tG}{\textbf{G}}
\newcommand{\tH}{\textbf{H}}
\newcommand{\tE}{\textbf{E}'}
\newcommand{\tC}{\textbf{C}}
\newcommand{\tphi}{\bm{\phi}}
\newcommand{\tpsi}{\bm{\psi}}
\newcommand{\tSigma}{\bm{\Sigma}}
\newcommand{\tB}{\tilde B}
\newcommand{\dout}{D_{\mbox{\tiny{out}}}}
\newcommand{\notF}{\overline{F}}

\renewcommand{\P}{\mbox{\sf P}}
\newcommand{\NP}{\mbox{\sf NP}}
\newcommand{\PCP}{\mbox{\sf PCP}}
\newcommand{\ZPP}{\mbox{\sf ZPP}}
\newcommand{\DTIME}{\mbox{\sf DTIME}}
\newcommand{\opt}{\mathsf{OPT}}
\newcommand{\optcro}{\mathsf{OPT}_{\mathsf{cr}}}
\newcommand{\optcrors}{\mathsf{OPT}_{\mathsf{cnwrs}}}
\newcommand{\set}[1]{\left\{ #1 \right\}}
\newcommand{\sse}{\subseteq}
\newcommand{\B}{{\mathcal{B}}}
\newcommand{\tset}{{\mathcal T}}
\newcommand{\vset}{{\mathcal V}}
\newcommand{\uset}{{\mathcal U}}
\newcommand{\iset}{{\mathcal{I}}}
\newcommand{\pset}{{\mathcal{P}}}
\newcommand{\nset}{{\mathcal{N}}}
\newcommand{\dset}{{\mathcal{D}}}
\newcommand{\tpset}{\tilde{\mathcal{P}}}
\newcommand{\qset}{{\mathcal{Q}}}
\newcommand{\tqset}{\tilde{\mathcal{Q}}}
\newcommand{\lset}{{\mathcal{L}}}
\newcommand{\bset}{{\mathcal{B}}}
\newcommand{\tbset}{\tilde{\mathcal{B}}}
\newcommand{\aset}{{\mathcal{A}}}
\newcommand{\cset}{{\mathcal{C}}}
\newcommand{\fset}{{\mathcal{F}}}
\newcommand{\mset}{{\mathcal M}}
\newcommand{\jset}{{\mathcal{J}}}
\newcommand{\xset}{{\mathcal{X}}}
\newcommand{\wset}{{\mathcal{W}}}
\newcommand{\gset}{{\mathcal{G}}}
\newcommand{\oset}{{\mathcal{O}}}
\newcommand{\yset}{{\mathcal{Y}}}
\newcommand{\rset}{{\mathcal{R}}}
\newcommand{\I}{{\mathcal I}}
\newcommand{\hset}{{\mathcal{H}}}
\newcommand{\sset}{{\mathcal{S}}}
\newcommand{\zset}{{\mathcal{Z}}}
\newcommand{\notu}{\overline U}
\newcommand{\vol}{\operatorname{vol}}
\newcommand{\nots}{\overline S}
\newcommand{\eint}{E^{\tiny\mbox{int}}}
\newcommand{\event}{{\cal{E}}}
\newcommand{\floor}[1]{\ensuremath{\left\lfloor#1\right\rfloor}}
\newcommand{\ceil}[1]{\ensuremath{\left\lceil#1\right\rceil}}

\newcommand{\marcon}{{\mathsf{MC}}}
\newcommand{\cov}{{\mathsf{cov}}}
\newcommand{\mst}{{\mathsf{MST}}}
\newcommand{\card}[1]{|#1|}
\newcommand{\coi}{{\mathsf{COI}}}

\newcommand{\med}{\mathcal{C}}
\newcommand{\avg}{\overline{\lambda}}

\newcommand{\cover}{\textsf{cover}}
\newcommand{\eps}{\varepsilon}
\newcommand{\bfs}{\textnormal{\textsf{BFS}}}
\newcommand{\pbfs}{\textnormal{\textsf{BFS}}}
\newcommand{\lv}{\textsf{lv}}
\newcommand{\tsp}{\mathsf{TSP}}
\newcommand{\gtsp}{\textsf{GTSP}}
\newcommand{\ebt}{\tset}
\newcommand{\eb}{\textsf{EB}}
\newcommand{\optmst}{\textsf{MST}}
\newcommand{\defi}{\textsf{def}}
\newcommand{\ord}{\textsf{ord}}
\newcommand{\rc}{\textnormal{\textsf{rc}}}
\newcommand{\dist}{\textnormal{\textsf{dist}}}
\newcommand{\cost}{\textnormal{\textsf{cost}}}
\newcommand{\bw}{\textsf{bw}}
\newcommand{\local}{\textsf{Local}}
\newcommand{\pseudo}{\textsf{Pseudo-IP}}
\newcommand{\vin}{v^{\textnormal{\textsf{in}}}}
\newcommand{\vout}{v^{\textnormal{\textsf{out}}}}
\newcommand{\diam}{\textsf{diam}}
\newcommand{\expect}{\mathbb{E}}
\newcommand{\proover}{\pi_{\textsf{Overwrite}}}
\newcommand{\promst}{\pi_{\textsf{MST}}}
\newcommand{\protsp}{\pi_{\textsf{TSP}}}
\newcommand{\mstest}{\textsf{MST}_{\textsf{apx}}}
\newcommand{\tspest}{\textsf{TSP}_{\textsf{apx}}}
\newcommand{\proind}{\pi_{\textsf{Index}}}
\newcommand{\ind}{\textsf{Index}}
\newcommand{\distIND}{\mathcal{D}_{\textsf{Index}}}
\newcommand{\distMST}{\mathcal{D}_{\textsf{MST}}}
\newcommand{\ic}{\textnormal{\textsf{IC}}}
\newcommand{\cc}{\textnormal{\textsf{CC}}}
\newcommand{\tvd}[2]{\ensuremath{\Delta_{\textnormal{\texttt{TV}}}(#1,#2)}}
\newcommand{\dkl}[2]{\ensuremath{D_{\textnormal{\textsf{KL}}}(#1 \| #2)}}

\newcommand{\hel}{h}
\newcommand{\II}{I}
\newcommand{\HH}{H}

\newcommand{\RV}[1]{\mathbf{#1}}
\newcommand{\prot}{\ensuremath{\Pi}}
\newcommand{\Prot}{\ensuremath{\Pi}}
\newcommand{\findmiss}{\sf{FindBit}}
\newcommand{\overwrite}{\sf{Overwrite}}
\newcommand{\distfind}{\mathcal{D}_{\textsf{FindBit}}}
\newcommand{\distover}{\mathcal{D}_{\textsf{Overwrite}}}
\newcommand{\temp}{\textsf{temp}}
\newcommand{\IA}{\textsf{IA}}
\newcommand{\IB}{\textsf{IB}}

\newcommand{\row}{\textsf{Row}}
\newcommand{\col}{\textsf{Col}}
\newcommand{\alg}{\ensuremath{\mathsf{Alg}}\xspace}
\newcommand{\algrc}{\ensuremath{\mathsf{AlgURC}}\xspace}
\newcommand{\minkcut}{minimum weight $k$-cut\xspace}
\newcommand{\wts}{\omega}
\newcommand{\Eout}{E^{\textnormal{out}}}

\newcommand{\sunf}{\mathsf{sf}}

\begin{titlepage}
	
\title{Almost-Optimal Sublinear Additive Spanners\footnote{Extended abstract appeared in STOC 2023.}}
\author{Zihan Tan \thanks{DIMACS, Rutgers University, \href{}{zihantan1993@gmail.com. This author is supported by a grant to DIMACS from the Simons Foundation (820931).}} \and Tianyi Zhang \thanks{Tel Aviv University, \href{}{tianyiz21@tauex.tau.ac.il. This author is supported by the European Research Council (ERC) under the European Union’s Horizon 2020 research and innovation programme (grant agreement No 803118 UncertainENV).}}}

	\maketitle
	
	\thispagestyle{empty}
	
\begin{abstract}
Given an undirected unweighted graph $G = (V, E)$ on $n$ vertices and $m$ edges, a subgraph $H\subseteq G$ is a \emph{spanner} of $G$ with stretch function $f: \mathbb{R}_+ \rightarrow \mathbb{R}_+$, if for every pair $s, t$ of vertices in $V$,  $\dist_{H}(s, t)\le f(\dist_{G}(s, t))$. When $f(d) = d + o(d)$, $H$ is called a \emph{sublinear additive spanner}; when $f(d) = d + o(n)$, $H$ is called an \emph{additive spanner}, and $f(d) - d$ is usually called the \emph{additive stretch} of $H$.

As our primary result, we show that for any constant $\delta>0$ and constant integer $k\geq 2$, every graph on $n$ vertices has a sublinear additive spanner with stretch function $f(d)=d+O(d^{1-1/k})$ and $O\big(n^{1+\frac{1+\delta}{2^{k+1}-1}}\big)$ edges. When $k = 2$, this improves upon the previous spanner construction with stretch function $f(d) = d + O(d^{1/2})$ and $\tilde{O}(n^{1+3/17})$ edges [Chechik, 2013]; for any constant integer $k\geq 3$, this improves upon the previous spanner construction with stretch function $f(d) = d + O(d^{1-1/k})$ and $O\bigg(n^{1+\frac{(3/4)^{k-2}}{7 - 2\cdot (3/4)^{k-2}}}\bigg)$ edges [Pettie, 2009]. Most importantly, the size of our spanners almost matches the 
lower bound of $\Omega\big(n^{1+\frac{1}{2^{k+1}-1}}\big)$ [Abboud, Bodwin, Pettie, 2017], which holds for all compression schemes achieving the same stretch function.

As our second result, we show a new construction of additive spanners with stretch $O(n^{0.403})$ and $\tilde{O}(n)$ edges, which slightly improves upon the previous stretch bound of $O(n^{3/7+\eps})$ achieved by linear-size spanners [Bodwin and Vassilevska Williams, 2016]. An additional advantage of our spanner is that it admits a subquadratic construction runtime of $\tilde{O}(m + n^{13/7})$, while the previous construction in [Bodwin and Vassilevska Williams, 2016] requires all-pairs shortest paths computation which takes $O(\min\{mn, n^{2.373}\})$ time.
	
\end{abstract}

\end{titlepage}

\begin{spacing}{1.3}
\tableofcontents
\end{spacing}

\thispagestyle{empty}
\clearpage
\setcounter{page}{1}

\newpage

\section{Introduction}
Graph spanners are sparse subgraphs that approximately preserve pairwise shortest-path distances. Let $G = (V, E)$ be an undirected unweighted graph on $n$ vertices and let $f:\mathbb{R}_+\rightarrow\mathbb{R}_+$ be a function. We say that a subgraph $H\subseteq G$ is a spanner with \emph{stretch function} $f$ if for every pair $s, t\in V$, $\dist_H(s, t)\leq f(\dist_G(s, t))$. The research on spanners focuses on the optimal trade-offs between the stretch function $f$ and the sparsity (the number of edges) of the spanner $H$.

One typical case is that we allow $f(d)$ to be significantly greater than $d$, and such spanners are known as \emph{multiplicative} spanners. It was shown that for every integer $k\ge 1$, there always exists a subgraph $H$ with $O(n^{1+1/k})$ edges and stretch function $f(d) = (2k-1)d$ \cite{althofer1993sparse}. Furthermore, this sparsity bound is tight under the Girth Conjecture of Erd{\"o}s \cite{erdos1963extremal}.

Another typical case is that we restrict $f$ to be very close to $d$. In particular, $f(d) = d + O(1)$. There has been a line of previous work studying the sparsity of spanners with such stretch functions. When $f(d) = d+2$, it was shown that graph $G$ always has a spanner with $O(n^{3/2})$ edges \cite{aingworth1999fast}; when $f(d) = d+4$, a construction of spanner with $\tilde O(n^{7/5})$ edges was proposed in \cite{chechik2013new}; when $f(d) = d+6$, spanners with $O(n^{4/3})$ edges were known to exist by \cite{baswana2010additive}. For the lower bound side, in a recent breakthrough \cite{abboud20174}, it was proved that, for any constant $\epsilon>0$, there are graphs such that any spanner with $O(n^{4/3-\eps})$ edges has stretch $n^{\Omega(1)}$. Hence, we already have an almost complete understanding of the spanner sparsity when $f(d) = d + O(1)$.

Besides the two typical cases mentioned above, much less is known when $f$ lies in intermediate regimes. Two notable regimes studied in previous works are the \emph{sublinear additive} regime where $f(d) = d + o(d)$, and the \emph{additive} regime where $f(d) = d + o(n)$.

As for sublinear additive spanners, Thorup and Zwick  \cite{thorup2006spanners,huang2019thorup} were the first to design a nontrivial construction of sublinear additive spanners when $f(d) = d + O(d^{1-1/k})$ where $k\geq 2$ is a constant integer, and the number of edges in the spanner is bounded by $O(n^{1 + 1/k})$. This sparsity bound was later improved to $O\bigg(n^{1+\frac{(3/4)^{k-2}}{7-2\cdot(3/4)^{k-2}}}\bigg)$ in \cite{pettie2009low}. The sparsity bound for the special case where $k=2$ was subsequently improved to $\tilde O(n^{20/17})$ by \cite{chechik2013new}. These algorithms also work for a non-constant $k$, but here we only focus on the case where $k$ is a constant, as we are mainly interested in the stretch/size dependency on $d$ and $n$. On the lower bound side, Abboud, Bodwin, and Pettie \cite{abboud2018hierarchy} proved that any distance-reporting data structure (including spanners, emulators) achieving stretch $f(d) = d + O(d^{1-1/k})$ must have size at least $\Omega\bigg(n^{1+\frac{1}{2^{k+1}-1}-o(1)}\bigg)$ for each constant $k\ge 2$, which is also the best known lower bound for spanners. Thus, there still exists a large gap between sparsity upper and lower bounds for sublinear additive spanners.

For the additive regime where $f(d) = d + o(n)$, the tail term $f(d) - d$ is usually called the \emph{additive stretch}. A natural question is to study the best additive stretch that can be achieved by spanners with $\tilde{O}(n)$ edges. The first nontrivial construction was given by \cite{pettie2009low} with an additive stretch of $O(n^{9/16})$, which was improved subsequently by \cite{bodwin2015very, bodwin2021better} to $O(n^{3/7+\eps})$ for any constant $\epsilon>0$. On the negative side, the first stretch lower bound of $\Omega(n^{1/22})$ was proved in \cite{abboud20174}, and later on raised to $\Omega(n^{1/7})$ by a sequence of works \cite{huang2021lower,lu2022better,bodwin2022new}.

\subsection{Our results}
\paragraph{Sublinear additive spanners.} Our primary result is the following almost optimal bound on sublinear additive spanners.

\begin{theorem}\label{sublinear}
For any undirected unweighted graph $G = (V, E)$ on $n$ vertices, given any constant $\delta>0$ and any constant integer $k\ge 2$, there exists a sufficiently large constant $C = C(k, \delta)$, such that
every undirected unweighted graph $G$ on $n$ vertices admits a spanner $H\subseteq G$ with stretch function $f(d)=d+C\cdot d^{1-1/k}$ and  $O\bigg(n^{1+\frac{1+\delta}{2^{k+1}-1}}\bigg)$ edges.
\end{theorem}


This sparsity upper bound almost matches the previous lower bound of $\Omega\brac{n^{1+\frac{1}{2^{k+1}-1}-o(1)}}$ proved in \cite{abboud2018hierarchy}. Formally, the authors of \cite{abboud2018hierarchy} showed that there exists a constant $c = \Theta\brac{\frac{1}{k+1}}$, such that in some hard examples of input graph $G$, any spanner with stretch $f(d)=d+c\cdot d^{1-1/(k+1)}$ must have $\Omega\bigg(n^{1+\frac{1}{2^{k+1}-1}-o(1)}\bigg)$ edges. When the distance value $d$ is larger than $\brac{C/c}^{k(k+1)}$ which is a constant only dependent on $k$, the stretch function $f(d) = d + C\cdot d^{1-1/k}$ of our spanner construction is always better than the stretch $d + c\cdot d^{1 - 1/(k+1)}$ in the lower bound construction of \cite{abboud2018hierarchy}. Since our sparsity upper bound almost matches the lower bound of $\Omega\brac{n^{1+\frac{1}{2^{k+1}-1}-o(1)}}$, our construction is almost optimal.

The same lower bound in \cite{abboud2018hierarchy} also extends to the more relaxed notion of emulators which are graphs defined on the vertex set $V$ but not required to be a subgraph of $G$. For the sublinear additive stretch function $f(d) = d + O(d^{1-1/k})$, there are known emulator constructions of size $O\brac{kn^{1+\frac{1}{2^{k+1}-1}}}$ by \cite{thorup2006spanners,huang2019thorup} which matches the lower bound in \cite{abboud2018hierarchy} and ties our new bound for spanners. In fact, the emulator lower bound in \cite{abboud2018hierarchy} holds for any form of graph metric compression schemes, in the sense that any data structure that answers approximate distance queries (regardless of query time) with stretch $f(d) = d + c\cdot d^{1-1/(k+1)}$ for some constant $c < 2k^{1-1/(k+1)}$ must occupy at least $n^{1 + \frac{1}{2^{k+1}-1}-o(1)}$ bits of space. Therefore, since our spanner construction matches the lower bound against any graph metric compression scheme, as our take-home message, it means that although spanners are the most restrictive form of compression compared to emulators and distance oracles, they are already the most powerful ones.

Our construction also directly translates to $(1+\epsilon, \beta)$-mixed spanners with stretch function $f(d) = (1+\epsilon)d + \beta$ when $\beta =\frac{(k-1)^{k-1}C^k}{k^k\epsilon^{k-1}}$ by AM-GM inequality. Previous works on mixed spanners include \cite{elkin20041eps} where the authors showed that for every integer $\kappa$ and $\epsilon>0$, every $n$-vertex graph admits a $(1+\epsilon,O(\epsilon^{-1}\log \kappa)^{\log \kappa})$-spanner of size $O(n^{1+1/\kappa}\cdot O(\epsilon^{-1}\log \kappa)^{\log \kappa})$, and \cite{thorup2006spanners} where the authors constructed $\brac{1+\epsilon,O(k/\eps)^{k-1}}$-spanner of size $O\brac{k\cdot n^{1+\frac{1}{2^{k+1}-1}}}$ for any integer $k$, and \cite{pettie2009low} where the author showed the construction of $(1+\epsilon,O(\log\log n/\eps)^{\log\log n})$-spanner of size $O(n\log\log (\eps^{-1}\log\log n))$. Among these previous results, the result by \cite{thorup2006spanners} and ours provide the same form of stretch $\brac{1+\epsilon, O(1 / \epsilon^{k-1})}$ and sparsity $n^{1 + \frac{1+\delta}{2^{k+1}-1}}$ when $k$ is a constant integer, which is proved to be almost optimal by \cite{abboud2018hierarchy}.

\paragraph{Additive spanners.} Our second result is a slightly improved bound on near-linear size additive spanners upon the bound $O(n^{3/7+\eps})$ of \cite{bodwin2021better}. In addition, we show that such a spanner can be computed in subquadratic time, while previous constructions in \cite{bodwin2021better,bodwin2015very} need the computation of all-pairs shortest paths in $G$ which takes time $O(\min\{mn, n^{2.373} \})$.

\begin{theorem}\label{subquad}
For any undirected unweighted graph $G = (V, E)$ on $n$ vertices and $m$ egdges, there exists a spanner with $O(n^{0.403})$ additive stretch and $\tilde{O}(n)$ edges. Moreover, such a spanner can be computed in sub-quadratic time $\tilde{O}\big(m + n^{13/7}\big)$.
\end{theorem}
\begin{remark}
	In the above theorem, if we could tolerate an $O(mn)$ construction time of the additive spanner, then the size can be made purely linear $O(n)$ with the same additive stretch $O(n^{0.403})$.
\end{remark}

\subsection{Technical overview}
The basic tool of our algorithms is a clustering algorithm in \cite{bodwin2021better}. Let $G = (V, E)$ be the input undirected unweighted graph. Roughly speaking, for any radius parameter $R$, we can decompose the graph into a set $\bset$ of balls that may share vertices, with the following properties.
\begin{itemize}
	\item (Radius) The radius of each ball in $\bset$ is roughly $R$.
	\item (Coverage) The union of all balls in $\bset$ covers the whole graph $G$.
	\item (Disjointness) The total size of the balls in $\bset$ is linear in $n$.
\end{itemize}

Next, we will describe how to utilize the above clustering to construct new sublinear additive spanners and additive spanners, respectively.

\subsubsection{Sublinear additive spanners}

Our starting point is a spanner with stretch function $f(d) = d + O(d^{1/2})$ and $\tilde{O}(n^{8/7})$ edges, which improves upon the previous sparsity bound of $\tilde{O}(n^{20/17})$ in \cite{chechik2013new}. Consider any pair of vertices $s, t\in V$, and let $\pi$ be a shortest path between them of length $D$. We apply the clustering algorithm from \cite{bodwin2021better} with radius parameter $R = D^{1/2}$ to $G$, and obtain a set $\bset$ of balls. Intuitively, by the coverage property, and recalling that every ball in $\balls$ has radius roughly $R=D^{1/2}$, we can choose a subset of $O(D / R) = O(D^{1/2})$ balls from $\bset$ whose union contains the entire shortest path $\pi$.
See Figure \ref{overview-path-partition} for an illustration.

\begin{figure}[h]
	\begin{center}
		\begin{tikzpicture}[thick,scale=0.7]
	\draw (0, 0) node[circle, draw, fill=black!50, inner sep=0pt, minimum width=6pt, label = $s$] {};
	\draw (19, 0) node[circle, draw, fill=black!50, inner sep=0pt, minimum width=6pt,label = $t$] {};
	
	\draw (3.5, 0) node[circle, draw, fill=black!50, inner sep=0pt, minimum width=6pt,label = $w_1$] {};
	\draw [line width = 0.5mm] (0.2, 0) -- (3.2, 0);
	\draw [red] (1.7, 0.5) ellipse (2.1 and 2);
	\draw (1.5, 2.4) node[red, label={$\ball(c_0, r_{0})$}]{};
	
	\draw (7.7, 0) node[circle, draw, fill=black!50, inner sep=0pt, minimum width=6pt,label = $w_2$] {};
	\draw [line width = 0.5mm] (3.8, 0) -- (7.5, 0);
	\draw [red] (5.6, 0.5) ellipse (2.4 and 2);
	\draw (5.5, 2.4) node[red, label={$\ball(c_1, r_{1})$}]{};
	
	\draw (16, 0) node[circle, draw, fill=black!50, inner sep=0pt, minimum width=6pt,label = $w_l$] {};
	\draw [line width = 0.5mm] (18.8, 0) -- (16.3, 0);
	\draw [red] (18, 0.5) ellipse (2.3 and 2);
	\draw (18, 2.4) node[red, label={$\ball(c_{l}, r_{l})$}]{};
	
	\draw (12, 0) node[circle, draw, fill=black!50, inner sep=0pt, minimum width=6pt,label = $w_{l-1}$] {};
	\draw [line width = 0.5mm] (15.7, 0) -- (12.2, 0);
	\draw [red] (14, 0.5) ellipse (2.3 and 2);
	\draw (13.7, 2.4) node[red, label={$\ball(c_{l-1}, r_{l-1})$}]{};
	
	\draw [dashed] (8, 0) -- (11.7, 0); 
\end{tikzpicture}
	\end{center}
	\caption{An illustration of a covering of a shortest path $\pi$ between $s, t$ with at most $l = O(D^{1/2})$ balls from $\bset$. For simplicity, we assume that, for each index $0\le i\le l-1$, the balls $\ball(c_i,r_i)$ and $\ball(c_{i+1},r_{i+1})$ share exactly one vertex of $\pi$, denoted by $w_{i+1}$. We denote $s = w_0$ and $t = w_{l+1}$.}\label{overview-path-partition}
\end{figure}
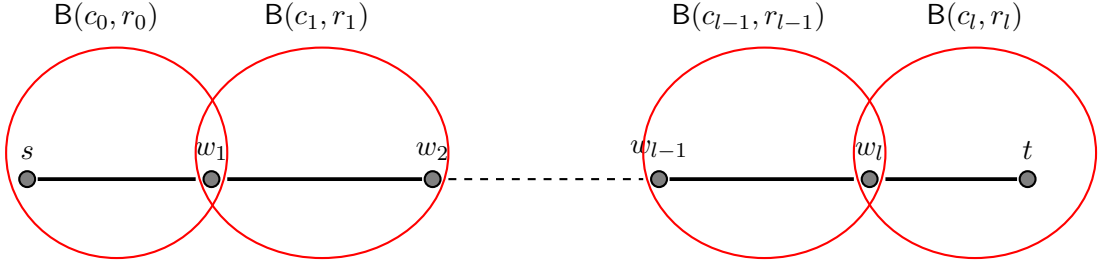

Following the notations in Figure \ref{overview-path-partition}, assume the sequence of balls divides the path into subpaths $\set{\pi[w_i, w_{i+1}]}_{0\le i\le l}$. In order to preserve the distance between $s, t$ in $G$, a simple approach is to plant, within each subgraph $G[\ball(c_i, r_i)]$, a $6$-additive spanner from \cite{baswana2010additive} with  $O(|\ball(c_i, r_i)|^{4/3})$ edges. Then for each $i$, $\dist_{H}(w_i, w_{i+1})\leq \dist_G(w_i, w_{i+1}) + 6$, so $\dist_H(s,t)-\dist_G(s,t)\le 6(l+1)=O(D^{1/2})$. If we further assume that each ball in $\balls$ contains at most $n^{3/7}$ vertices, then by the disjointness property, the union of all these $6$-additive spanners contains at most $\sum_{\ball(c, r)\in\bset}|\ball(c, r)|^{4/3} = O(n^{8/7})$ edges. So we only need to deal with large balls that contain more than $n^{3/7}$ vertices.

A natural approach to handling the large balls is taking a random subset $S\subseteq V$ of $10n^{4/7}\log n$ vertices that hits all large balls with high probability, and try preserving pairwise distances between vertices in $S$, as preserving distances among a subset is conceivably easier than the whole graph.
In fact, if this can be done, then the distance between all pairs $s, t\in V$ is also well-preserved. 
To see why this is true, assume for simplicity that $\ball(c_1, r_1)$ and $\ball(c_{l-1}, r_{l-1})$ are the first and the last large ball, respectively. Then, by construction of $S$, there exist $u\in \ball(c_1, r_1)\cap S$ and $v\in \ball(c_{l-1}, r_{l-1})\cap S$, and it is easy to verify using triangle inequality that the following path connecting $s$ to $t$ in $H$ has length at most $\dist_G(s,t)+O(R)$; the path is the concatenation of (see \Cref{overview-hitset} for an illustration):
	\begin{itemize}[leftmargin=*]
		\item a path from $s$ to $w_1$ in the $6$-additive spanner within subgraph $G[\ball(c_0, r_0)]$; and similarly a path from $t$ to $w_l$ in the $6$-additive spanner within subgraph $G[\ball(c_l, r_l)]$;
		\item a path from $w_1$ to $u$ of length at most $2R$, as we can afford to add to $H$ a breadth-first search tree of $G[\ball(c_1, r_1)]$ beforehand; similarly, a path from $w_l$ to $v$ of length at most $2R$; and
		\item the shortest path connecting $u, v$ in $H$.
	\end{itemize}

\begin{figure}[h]
	\begin{center}
		\begin{tikzpicture}[thick,scale=0.7]
	\draw (0, 0) node[circle, draw, fill=black!50, inner sep=0pt, minimum width=6pt, label = $s$] {};
	\draw (19, 0) node[circle, draw, fill=black!50, inner sep=0pt, minimum width=6pt,label = $t$] {};
	
	\draw (3.5, 0) node[circle, draw, fill=black!50, inner sep=0pt, minimum width=6pt,label = $w_1$] {};
	\draw [color=orange, style={decorate, decoration=snake}] (0.2, 0) -- (3.2, 0);
	\draw [red] (1.7, 0.5) ellipse (2.1 and 2);
	\draw (1.5, 2.4) node[red, label={$\ball(c_0, r_{0})$}]{};
	\draw (5.5, 1.5) node[circle, draw, fill=black!50, inner sep=0pt, minimum width=6pt,label = $u$] {};
	\draw [color=orange, style={decorate, decoration=snake}] (3.7, 0.2) -- (5.3, 1.3);
	
	\draw (7.7, 0) node[circle, draw, fill=black!50, inner sep=0pt, minimum width=6pt,label = $w_2$] {};
	\draw [line width = 0.5mm] (3.8, 0) -- (7.5, 0);
	\draw [red] (5.6, 0.5) ellipse (2.4 and 2);
	\draw (5.5, 2.4) node[red, label={$\ball(c_1, r_{1})$}]{};
	
	\draw (16, 0) node[circle, draw, fill=black!50, inner sep=0pt, minimum width=6pt,label = $w_l$] {};
	\draw [color=orange, style={decorate, decoration=snake}] (18.8, 0) -- (16.3, 0);
	\draw [red] (18, 0.5) ellipse (2.3 and 2);
	\draw (18, 2.4) node[red, label={$\ball(c_{l}, r_{l})$}]{};
	\draw (14, 1.5) node[circle, draw, fill=black!50, inner sep=0pt, minimum width=6pt,label = $v$] {};
	\draw [color=orange, style={decorate, decoration=snake}] (15.8, 0.2) -- (14.2, 1.3);
	
	\draw (12, 0) node[circle, draw, fill=black!50, inner sep=0pt, minimum width=6pt,label = $w_{l-1}$] {};
	\draw [line width = 0.5mm] (15.7, 0) -- (12.2, 0);
	\draw [red] (14, 0.5) ellipse (2.3 and 2);
	\draw (13.7, 2.4) node[red, label={$\ball(c_{l-1}, r_{l-1})$}]{};
	
	\draw [dashed] (8, 0) -- (11.7, 0);
	\draw [color=orange, style={decorate, decoration=snake}] (5.7, 1.5) -- (13.8, 1.5);
\end{tikzpicture}
	\end{center}
	\caption{If the balls $\ball(c_1, r_1)$ and $\ball(c_{l-1}, r_{l-1})$ are large, then we can find a short path from $s$ to $t$ drawn as the orange wavy lines.}\label{overview-hitset}
\end{figure}
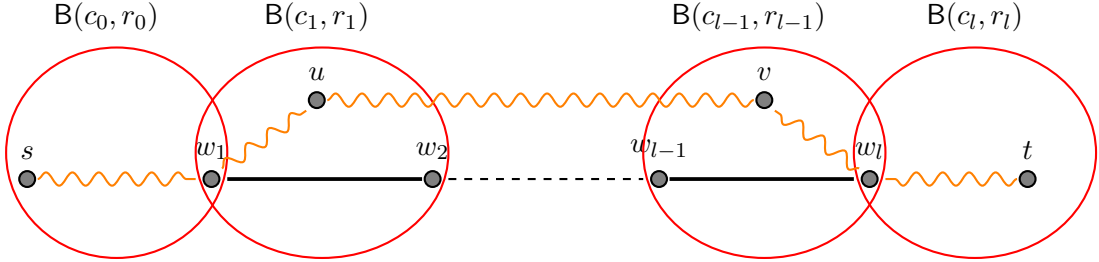

Therefore, it suffices to construct a spanner that faithfully preserves pairwise distances between vertices in $S$.
Consider now a pair $s, t$ of vertices in $S$. We proceed similarly by first finding a sequence of $O(D^{1/2})$ balls in $\bset$ that covers the shortest path between $s$ and $t$, and
then computing a sparse subgraph $H_{c_i}\subseteq G[\ball(c_i, r_i)]$ in order to preserve the distance between $w_i$ and $w_{i+1}$ to within a constant additive error.
Note that we would like that the graph $H_{c_i}$ also approximately preserves, for other pairs $s',t'\in S$, the distance between their corresponding vertices $w'_j,w'_{j+1}$.
The reason that this is easier to achieve is because there are fewer pairs $(w_i, w_{i+1})$ assigned to the ball $\ball(c_i, r_i)$ while processing all pairs in $S$, as $|S|$ itself is small. More formally, each ball $\ball(c_i, r_i)$ is associated with a set $\pset_{c_i}\subseteq V\times V$ of \emph{demand pairs}, and we want to find a spanner $H_{c_i}\subseteq G[\ball(c_i, r_i)]$ that only preserves distances between vertex pairs in $\pset_{c_i}$. The way we construct $\pset_{c_i}$ is to enumerate all pairs $s, t\in S$, find an $s$-$t$ shortest path and a set of balls that cover it, and if we have used the ball $\ball(c_i,r_i)$ in the covering, then add the corresponding pair $(w_i, w_{i+1})$ to $\pset_{c_i}$. In the end, we can show that the size of each set $\pset_{c_i}$ is at most $|S| = \tilde O(n^{4/7})$.

This special type of spanners restricted to demand pairs are called \emph{pairwise spanners}, and has been studied in \cite{cygan2013pairwise,kavitha2017new}. Applying their result in a black-box way, we can construct, for each ball $G[\ball(c_i, r_i)]$, a subgraph $H_{c_i}$ of size $O(|\ball(c_i, r_i)|\cdot |\pset_{c_i}|^{1/4})$ preserving distances between pairs in $\pset_{c_i}$ up to an additive error of $6$. Then, the total size of all $H_{c_i}$ is $\sum_{\ball(c, r)\in\bset}|\ball(c, r)|\cdot |\pset_c|^{1/4} = \tilde{O}(n^{8/7})$.

\subsubsection{Pairwise sublinear additive spanners}

The only missing component towards a sublinear additive spanner with stretch $f(d) = d + O(d^{1-1/k})$ for general $k\geq 2$ turns out to be a pairwise spanner with stretch $d + O(d^{1-1/(k-1)})$. As the previous work on pairwise spanners \cite{cygan2013pairwise,kavitha2017new} only considered stretch function $f(d) = d + \tilde O(1)$, we need to generalize their result for general $k$. Assume we are given an undirected unweighted graph $G = (V, E)$ on $n$ vertices and a set $\pset\subseteq V\times V$ of pairs, and the goal is to find a spanner $H\subseteq G$ with at most $\tilde{O}(n|\pset|^{1/2^{k+1}})$ edges, such that $\dist_H(s, t)\leq f(\dist_G(s, t))$ for all $(s, t)\in \pset$. For the purpose of this overview, we assume for simplicity that $\dist_G(s, t) = D$ for all $(s, t)\in \pset$.

The construction of $H$ is inductive on $k\geq 1$. Assume we know how to construct such spanners for $k-1$. First, apply the clustering algorithm from \cite{bodwin2021better} with radius $R = D^{1-1/k}$ to graph $G$, and obtain a set $\bset$ of balls.

\paragraph{Uniform size.} We start by considering a special case, where all balls in $\bset$ have the same size $|\ball(c, r)| = L$. 
By the disjointness property, the number of balls in $\balls$ is approximately $n/L$.
For each pair $(s, t)\in \pset$, we compute a sequence of $D^{1/k}$ balls and partition the shortest path connecting $s, t$, in a similar way as illustrated in \Cref{overview-path-partition}. 
We wish to inductively build inside each ball $G[\ball(c_i, r_i)]$ a pairwise spanner with a better stretch $d + O(d^{1-1/(k-1)})$; if this is done, then the cumulative error along the shortest path connecting $s$ and $t$ is bounded by $D^{1/k}\cdot R^{1-1/(k-1)} = D^{1-1/k}$.

In order to apply the inductive hypothesis, we will construct, for each ball $\ball(c, r)\in\bset$, a set $\pset_c$ of demand pairs. But if we add, for all $i$, the pair $(w_i, w_{i+1})$ to $\pset_{c_i}$, then the sets $|\pset_{c_i}|$ may become too large (as large as $|\pset|$) to construct a spanner for. To circumvent this, we use the approach in \cite{kavitha2017new}. Observe that, if we add each pair $(w_i, w_{i+1})$ to $\pset_{c_i}$, then not only the distance between $s, t$ is approximately preserved (in this case we call the pair $(s,t)$ \emph{settled}), but the distances between all pairs of ball centers $c_i, c_j$ are also approximately preserved. Now there are two cases.
Let $\beta$ be a parameter to be determined later.

\begin{itemize}[leftmargin=*]
	\item Case 1. There are at least $\beta\cdot (l+1)$ center pairs that were unsettled in $H$; recall that $l+1$ refers to the number of clusters as in \Cref{overview-path-partition}.
	
Then after adding each $(w_i, w_{i+1})$ as a demand pair to $\pset_{c_i}$, the total number of such center pairs would drop by $\beta\cdot (l+1)$. In other words, each new demand pair $(w_i, w_{i+1})$ contributes $\beta$ to the decrease of the number of unsettled center pairs on average.
	
	\item Case 2. There are fewer than $\beta\cdot (l+1)$ center pairs that were unsettled in $H$.
	
	In this case, we show that there is always a ``bridge'' that helps settle the pair $(s, t)$. Specifically, we can prove there are indices $x\in [1, \beta], z\in [\beta+1, l+1 - \beta], y\in [l-\beta+2, l+1]$ such that center pairs $(c_x, c_z)$ and $(c_z, c_y)$ were already settled. In this case, adding demand pairs $(w_i, w_{i+1})$ to $\pset_{c_i}$ only for $i\in [1, \beta]\cup [l-\beta+2, l+1]$ would be enough. See Figure \ref{overview-bridge} for an illustration.
\end{itemize}

To summarize, for each pair $(s, t)\in\pset$, we either add at most $2\beta$ demand pairs in total, or add a number of demand pairs such that the number of unsettled center pairs decreases by $\beta$ per pair. Therefore, the total number of demand pairs is at most $\beta |\pset| + \frac{n^2}{L^2\beta}$, which is $O\big(\frac{n|\pset|^{1/2}}{L}\big)$ if we set $\beta = \frac{n}{|\pset|^{1/2}}$, and thus the size of sets $|\pset_c|$ is bounded by $|\pset|^{1/2}$ on average. Now apply the inductive hypothesis within each ball, and the total size of all these pairwise sublinear spanners is bounded by $\tilde{O}(n|\pset|^{1/2^{k+1}})$.

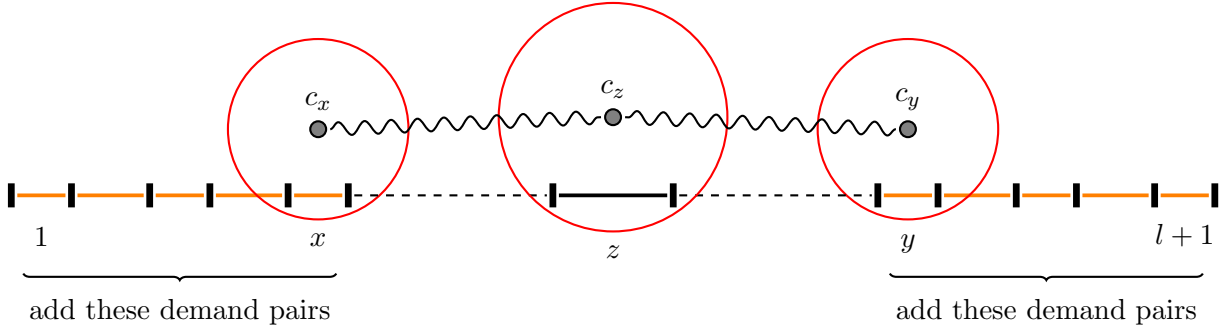
\begin{figure}[h]
	\begin{center}
		\begin{tikzpicture}[thick,scale=0.8]
	\draw [line width = .9mm] (0, 0.2) -- (0, -0.2);
	\draw [line width = .9mm] (1, 0.2) -- (1, -0.2);
	\draw [line width = 0.5mm, color=orange] (0.1, 0) -- (0.9, 0);
	\draw (0.5, -1.2) node[red, label={$1$}]{};
	
	\draw [line width = 0.5mm, color=orange] (1.1, 0) -- (2.2, 0);
	
	\draw [line width = .9mm] (2.3, 0.2) -- (2.3, -0.2);
	\draw [line width = .9mm] (3.3, 0.2) -- (3.3, -0.2);
	\draw [line width = 0.5mm, color=orange] (2.4, 0) -- (3.2, 0);
	
	\draw [line width = 0.5mm, color=orange] (3.4, 0) -- (4.5, 0);
	
	\draw [line width = .9mm] (4.6, 0.2) -- (4.6, -0.2);
	\draw [line width = .9mm] (5.6, 0.2) -- (5.6, -0.2);
	\draw [line width = 0.5mm, color=orange] (4.7, 0) -- (5.5, 0);
	\draw (5.1, -1.2) node[red, label={$x$}]{};
	\draw [red] (5.1, 1.1) ellipse (1.5 and 1.5);
	\draw (5.1, 1.1) node[circle, draw, fill=black!50, inner sep=0pt, minimum width=6pt, label = $c_{x}$] {};
	
	\draw [line width = .9mm] (20, 0.2) -- (20, -0.2);
	\draw [line width = .9mm] (19, 0.2) -- (19, -0.2);
	\draw [line width = 0.5mm, color=orange] (19.1, 0) -- (19.9, 0);
	\draw (19.5, -1.2) node[red, label={$l+1$}]{};
	
	\draw [line width = 0.5mm, color=orange] (17.8, 0) -- (18.9, 0);
	
	\draw [line width = .9mm] (17.7, 0.2) -- (17.7, -0.2);
	\draw [line width = .9mm] (16.7, 0.2) -- (16.7, -0.2);
	\draw [line width = 0.5mm, color=orange] (16.8, 0) -- (17.6, 0);
	
	\draw [line width = 0.5mm, color=orange] (15.5, 0) -- (16.6, 0);
	
	\draw [line width = .9mm] (15.4, 0.2) -- (15.4, -0.2);
	\draw [line width = .9mm] (14.4, 0.2) -- (14.4, -0.2);
	\draw [line width = 0.5mm, color=orange] (14.5, 0) -- (15.3, 0);
	\draw (14.9, -1.3) node[red, label={$y$}]{};
	\draw [red] (14.9, 1.1) ellipse (1.5 and 1.5);
	\draw (14.9, 1.1) node[circle, draw, fill=black!50, inner sep=0pt, minimum width=6pt, label = $c_{y}$] {};
	
	\draw [dashed] (5.7, 0) -- (8.9, 0);
	\draw [line width = .9mm] (9, 0.2) -- (9, -0.2);
	\draw [line width = .9mm] (11, 0.2) -- (11, -0.2);
	\draw [line width = 0.5mm] (9.1, 0) -- (10.9, 0);
	\draw (10, -1.4) node[red, label={$z$}]{};
	\draw [dashed] (11.1, 0) -- (14.3, 0);
	
	\draw [red] (10, 1.3) ellipse (1.9 and 1.9);
	\draw (10, 1.3) node[circle, draw, fill=black!50, inner sep=0pt, minimum width=6pt, label = $c_{z}$] {};
	
	\draw [style={decorate, decoration=snake}] (5.3, 1.1) -- (9.8, 1.3);
	\draw [style={decorate, decoration=snake}] (14.7, 1.1) -- (10.2, 1.3);
	
	\draw [decorate,
	decoration = {brace}] (19.8,-1.3) -- (14.6,-1.3);
	\draw (17.2, -2.5) node[label={add these demand pairs}]{};
	
	\draw [decorate,
	decoration = {brace}] (5.4,-1.3) -- (0.2,-1.3);
	\draw (2.8, -2.5) node[label={add these demand pairs}]{};
\end{tikzpicture}
	\end{center}
	\caption{In this case, we can find a ``bridge'' that allows us to travel from $c_x$ to $c_z$ and then to $c_y$ via a near-shortest path. The orange segments represent new demand pairs assigned to the balls.}\label{overview-bridge}
\end{figure}

\paragraph{Non-uniform size.} In order to provide some intuition for the general case, we consider a slightly more general case (compared with the uniform size case) where balls in $\bset$ have size either $L_1$ or $L_2$, which naturally classify balls in $\bset$ into \emph{level-$1$} balls and \emph{level-$2$} balls. Similar to the uniform size case, we process all pairs $(s, t)\in \pset$ and find the balls that partition the shortest $s$-$t$ path. Assume there are $l_1$ level-1 balls and $l_2$ level-2 balls. To extend the technique from \cite{kavitha2017new}, this time we only count the number of center pairs \emph{separately} for each level. For each $b\in \{1, 2\}$, let $\Phi_b$ be the number of unsettled level-$b$ ball center pairs, and let $\Delta\Phi_b$ be the number of unsettled level-$b$ center pairs along the $s$-$t$ path.
Similarly, we consider the following two cases.
(Below $\beta_1,\beta_2$ are to-be-determined parameters, similar to $\beta$ in the uniform case.)

\begin{itemize}[leftmargin=*]
	\item Case 1. $\Delta\Phi_b \geq \beta_b\cdot l_b$ for both $b\in \{1, 2\}$.
	
	In this case, by adding, for all $i$, the pair $(w_i, w_{i+1})$ to $\pset_{c_i}$, on average, $\Phi_1$ is decreased by $\beta_1$ per pair, and $\Phi_2$ is decreased by $\beta_2$ per pair.
	
	\item Case 2. (Without loss of generality) $\Delta\Phi_1 < \beta_1\cdot l_1$.
	
	Similar to the uniform size case, we are able to find a bridge of level-1 ball centers. Specifically, we can prove there are indices $x, y, z$ such that both pairs $(c_x, c_z)$ and $(c_z, c_y)$ are settled level-$1$ ball centers. In this case, we add demand pairs $(w_i, w_{i+1})$ to $\pset_{c_i}$ where $\ball(c_i, r_i)$ is among the first and the last $\beta_1$ level-1 balls; see Figure \ref{overview-bridge2} for an illustration.

	Now that we can travel from $c_x$ to $c_y$, we need to analyze the stretch for $s, c_x$ and $c_y, t$. The idea is to zoom into the intervals $[1, x]$ and $[y, l+1]$, and recurse. Now, since all the level-$1$ balls have already been assigned their demand pairs within these two intervals, we have now reduced to the uniform size case within those two intervals $[1, x], [y, l+1]$.
\end{itemize}

Intuitively, in the general case where all balls in $\bset$ have different sizes, we will first partition them into $O(1/\eps)$ groups according to their sizes, and then generalize the above approach for two ball sizes to multiple ball sizes. In this general setting, every time we go down by one level on the recursion tree and zoom into a smaller interval $I$, we will decrease the number of groups by one; that is, there will be one more group whose pairs in this interval $I$ are already purchased by our algorithm. The recursion tree will have size $2^{O(1/\epsilon)}$ and will eventually incur an additional factor $2^{O(k/\eps)}n^{O(k\eps)}$ in the sparsity.

\begin{figure}[h]
	\begin{center}
		\begin{tikzpicture}[thick,scale=0.8]
	\draw [line width = .9mm] (0, 0.2) -- (0, -0.2);
	\draw [line width = .9mm] (1, 0.2) -- (1, -0.2);
	\draw [line width = 0.5mm, color=orange] (0.1, 0) -- (0.9, 0);
	\draw (0.5, -1.2) node[red, label={$1$}]{};
	
	\draw [line width = 0.5mm, color=cyan] (1.1, 0) -- (2.2, 0);
	
	\draw [line width = .9mm] (2.3, 0.2) -- (2.3, -0.2);
	\draw [line width = .9mm] (3.3, 0.2) -- (3.3, -0.2);
	\draw [line width = 0.5mm, color=orange] (2.4, 0) -- (3.2, 0);
	
	\draw [line width = 0.5mm, color=cyan] (3.4, 0) -- (4.5, 0);
	
	\draw [line width = .9mm] (4.6, 0.2) -- (4.6, -0.2);
	\draw [line width = .9mm] (5.6, 0.2) -- (5.6, -0.2);
	\draw [line width = 0.5mm, color=orange] (4.7, 0) -- (5.5, 0);
	\draw (5.1, -1.2) node[red, label={$x$}]{};
	\draw [red] (5.1, 1.1) ellipse (1.5 and 1.5);
	\draw (5.1, 1.1) node[circle, draw, fill=black!50, inner sep=0pt, minimum width=6pt, label = $c_{x}$] {};
	
	\draw [line width = .9mm] (20, 0.2) -- (20, -0.2);
	\draw [line width = .9mm] (19, 0.2) -- (19, -0.2);
	\draw [line width = 0.5mm, color=orange] (19.1, 0) -- (19.9, 0);
	\draw (19.5, -1.2) node[red, label={$l+1$}]{};
	
	\draw [line width = 0.5mm, color=cyan] (17.8, 0) -- (18.9, 0);
	
	\draw [line width = .9mm] (17.7, 0.2) -- (17.7, -0.2);
	\draw [line width = .9mm] (16.7, 0.2) -- (16.7, -0.2);
	\draw [line width = 0.5mm, color=orange] (16.8, 0) -- (17.6, 0);
	
	\draw [line width = 0.5mm, color=cyan] (15.5, 0) -- (16.6, 0);
	
	\draw [line width = .9mm] (15.4, 0.2) -- (15.4, -0.2);
	\draw [line width = .9mm] (14.4, 0.2) -- (14.4, -0.2);
	\draw [line width = 0.5mm, color=orange] (14.5, 0) -- (15.3, 0);
	\draw (14.9, -1.3) node[red, label={$y$}]{};
	\draw [red] (14.9, 1.1) ellipse (1.5 and 1.5);
	\draw (14.9, 1.1) node[circle, draw, fill=black!50, inner sep=0pt, minimum width=6pt, label = $c_{y}$] {};
	
	\draw [dashed] (5.7, 0) -- (8.9, 0);
	\draw [line width = .9mm] (9, 0.2) -- (9, -0.2);
	\draw [line width = .9mm] (11, 0.2) -- (11, -0.2);
	\draw [line width = 0.5mm, color=orange] (9.1, 0) -- (10.9, 0);
	\draw (10, -1.4) node[red, label={$z$}]{};
	\draw [dashed] (11.1, 0) -- (14.3, 0);
	
	\draw [red] (10, 1.3) ellipse (1.9 and 1.9);
	\draw (10, 1.3) node[circle, draw, fill=black!50, inner sep=0pt, minimum width=6pt, label = $c_{z}$] {};
	
	\draw [style={decorate, decoration=snake}] (5.3, 1.1) -- (9.8, 1.3);
	\draw [style={decorate, decoration=snake}] (14.7, 1.1) -- (10.2, 1.3);
	
	\draw [decorate,
	decoration = {brace}] (19.8,-1.3) -- (14.6,-1.3);
	\draw (17.2, -2.5) node[label={add level-1 demand pairs}]{};
	
	\draw [decorate,
	decoration = {brace}] (5.4,-1.3) -- (0.2,-1.3);
	\draw (2.8, -2.5) node[label={add level-1 demand pairs}]{};
\end{tikzpicture}
	\end{center}
	\caption{Orange segments are covered by level-1 balls, and cyan segments are covered by level-2 balls; we only add the first and the last $\beta_1$ level-$1$ segments as new demand pairs.}\label{overview-bridge2}
\end{figure}
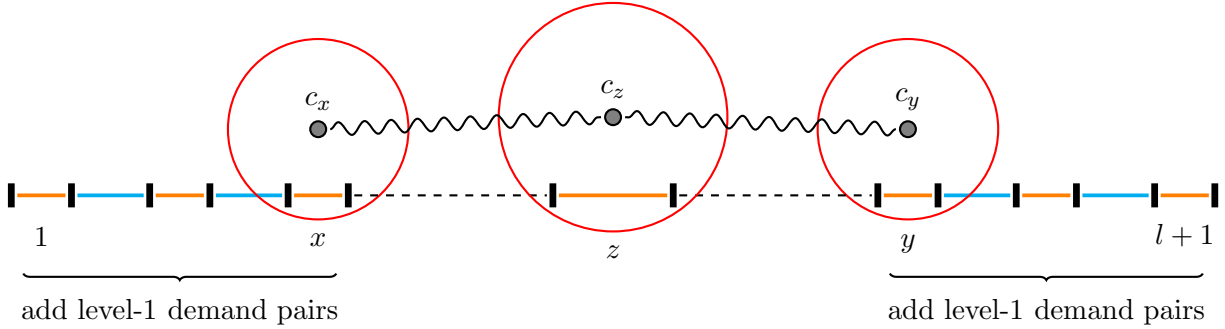

\subsubsection{Additive spanners}
We now provide an overview for our additive spanner construction. In what follows we will mainly highlight the difference between the construction by \cite{bodwin2021better} and ours.

First, we apply the clustering algorithm from \cite{bodwin2021better} with radius $R$ to $G$ and obtain a set $\bset$ of balls. 
Similar to the construction of sublinear and pairwise additive spanners, while processing unsettled pairs of vertices, we will add demand pairs to the corresponding balls that partition their shortest paths.
It has been shown that $O(\sqrt{n})$ demand pairs can be preserved exactly with $O(n)$ edges \cite{coppersmith2006sparse}, and the bound $O(\sqrt{n})$ cannot be improved, so we need to ensure that each ball $\ball$ collects at most $O(|\ball|^{1/2})$ demand pairs. On the one hand, the small balls $\ball$ with $|\ball|\le R^{4/3}$ can be handled using previous work on subset distance preservers \cite{coppersmith2006sparse}. 
On the other hand, for the large balls $\ball$ with $|\ball|\ge R^{4/3}$, 
the algorithm in \cite{bodwin2021better} placed an $O(|\ball|^{1/2})$ additive spanner (which follows from \cite{bodwin2015very}) inside them. This implies that, whenever a new path $\pi$ passes through a ball $\ball$, along with the demand pair added to $\ball$, at least $R^{4/3}\cdot (R / R^{2/3})=R^{5/3}$ 
new vertices\footnote{Say $\pi$ connects vertex $s$ to vertex $t$. Note that the $+O(\sqrt{n})$ spanners planted inside the balls settled a total of at least $R\cdot (R^{4/3}/R^{2/3})$ vertices with $s$ in the balls passed through by $\pi$, which we call the prefix of $\pi$; and similarly they also settled at least $R\cdot (R^{4/3}/R^{2/3})$ vertices with $s$ in the balls passed through by $\pi$, which we call the prefix of $\pi$. Also note that either the prefix or the suffix must be previously not settled with $\ball$ and now settled with $\ball$, as otherwise the pair $(s,t)$ was already settled and we should not take $\pi$.}
are settled with $\ball$, and so a ball $\ball$ may collect $O(n/R^{5/3})$ pairs. 
When $R=n^{3/7}$, 
a ball $\ball$ can collect at most
$O(n/R^{5/3})=O(n^{2/7})\le O\brac{(R^{4/3})^{1/2}}\le O(|\ball|^{1/2})$ demand pairs so that we can still apply the pairwise distance preserver construction from \cite{coppersmith2006sparse}.

We proceed differently. We start by showing via the above framework that any graph $G$ and any subset $S\subseteq V(G)$ has an $O(|S|^{3/2})$ additive subset spanner of linear size with respect to $S$, which is a subgraph $H\subseteq G$, such that for every pair $s,s'\in S$, $\dist_{H}(s,s')\le \dist_{G}(s,s')+O(|S|^{3/2})$. The algorithm is simple. Set $R=|S|^{3/2}$ and compute a clustering (we can ignore balls $\ball$ with $|\ball|\le (|S|^{3/2})^{4/3}=|S|^2$); iteratively go through all shortest paths connecting pairs in $S$ and buy them whenever their endpoints are unsettled. Since whenever a cluster collects a new demand pair, it is settled with a new vertex in $S$, a ball $\ball$ gets at most $|S|=O(|\ball|^{1/2})$ demand pairs. Now instead of planting a $+O(|\ball|^{1/2})$ spanner inside each ball, we plant a $+O(|S|^{3/2})$ subset spanner inside each ball with respect to its boundary vertices. The stretch between them is now $O((|\ball|/R)^{3/2})$, which is less than $|\ball|^{1/2}$ when $R^{4/3}<|\ball|<R^{3/2}$.

To see how much this improves the additive error, consider a canonical scenario where all balls have the same size $\Theta(R^{4/3})$ which is the threshold for distinguishing small/large balls. Then, a subset spanner within a ball $\ball$ incurs additive error $O(R^{1/2})$. Following the same calculation as before, whenever a new path $\pi$ passes through a ball $\ball$, along with the demand pair added to $\ball$, at least $R^{4/3}\cdot (R / R^{1/2})=R^{11/6}$ new vertices will be settled with $\ball$, and so $\ball$ may collect at most $O(n / R^{11/6})$ pairs. If we take $R = n^{0.4}$, a ball $\ball$ can collect at most $O(n/R^{11/6})=O(n^{4/15})\le O\brac{(R^{4/3})^{1/2}}\le O(|\ball|^{1/2})$ demand pairs so that we can still apply the pairwise distance preserver construction from \cite{coppersmith2006sparse}. In general, we are not lucky enough that all balls have size $\Theta(R^{4/3})$, so with some further utilization of the subset spanner, we end up with a slightly worse bound of $O(n^{0.403})$.

\subsection{Organization}
We start with preliminaries in \Cref{sec: prelim}. We provide a construction for pairwise additive spanners in \Cref{sec: pairwise}, which will be a building block in the proof of \Cref{sublinear} which appears in \Cref{sec: sublinear spanner}.
Next, we show the construction of subset spanners in \Cref{sec: subset}, which will be a crucial subroutine used in the algorithm of \Cref{subquad}. Lastly, we provide the proof of \Cref{subquad} in \Cref{sec: subq}.

\section{Preliminaries}
\label{sec: prelim}

By default, all logarithms are to the base of $2$, and all graphs are undirected and unweighted.

Let $G = (V, E)$ be a graph. 
For a subset $S\subseteq V$ of vertices, we define $\vol_G(S)=\sum_{v\in S}\deg_G(v)$, where $\deg_G(v)$ is the degree of $v$ in $G$.
For any pair $u, v\in V(G)$, let $\dist_G(u, v)$ be the shortest-path distance between $u$ and $v$ in $G$. For any $c\in V$ and $r>0$, we define the ball $\ball_G(c, r)$ as the set of vertices in $G$ that are at distance at most $r$ from $c$, namely $\ball_G(c, r)=\set{v\in V\mid \dist_G(c,v)\le r}$; $r$ is called the \emph{radius} of the ball, and we define the \emph{boundary} of the ball $\ball_G(c, r)$ as $\ball_G^=(c, r)=\set{v\in V\mid \dist_G(c,v)= r}$. We will sometimes omit the subscript $G$ in the above notations when it is clear from context. 
Let $\balls$ be a collection of balls. We say that a vertex $v$ is \emph{covered} by $\balls$ if $v$ belongs to some ball in $\balls$.
We will use the following lemma, whose proof is a simple implementation of the breadth-first search (BFS) algorithm, and is omitted.

\begin{lemma}\label{bottleneck}
Given a graph $G$, a vertex $c\in V(G)$, and two integers $0<r_1 < r_2<r$, we can find an integer $r_1< d\le r_2$, such that $|\ball^=(c, d)\cup\ball^=(c, d+1)|\leq \frac{2\cdot |\ball(c, r)|}{r_2-r_1}$, in time $O(\vol(\ball(c, r)))$.
\end{lemma}

\paragraph{Clustering in almost-linear time.}
One of the building blocks in the construction of $+O(n^{3/7+\eps})$ spanner in \cite{bodwin2021better} is a clustering procedure, which takes the input graph and computes a collections of balls with certain coverage and disjointness properties.
We will also use this clustering procedure. However, the running time of algorithm for computing such a clustering in \cite{bodwin2021better} is a large polynomial in $n$. In order to design a sub-quadratic algorithm for additive spanners, we provide an almost-linear time algorithm for computing a clustering with slightly different parameters.

\begin{lemma}[Almost-linear time algorithm for Lemma 13 in \cite{bodwin2021better}]
\label{clustering}
There is an algorithm, that, given any undirected unweighted graph $G = (V,  E)$ on $n$ vertices and $m$ edges, any parameter $\eps > 0$, and any integer $R>0$, computes in time $O(mn^{\eps}/\eps)$ a collection $\balls$ of balls in $G$, such that
\begin{itemize}
\item the radius of each ball $\ball(c,r)\in \balls$ satisfies that $R\leq r \leq 2^{10/\eps}\cdot R$; 
\item all vertices in $V$ are covered by $\balls$;
\item the following coverage properties hold:
\begin{itemize}
\item $\sum_{\ball(c, r)\in\balls}|\ball(c, r/2)| = O(n/\eps)$;
\item for each $\ball(c, r)\in\balls$, $|\ball(c, 4r)|\le n^{\eps}\cdot |\ball(c, r/2)|$, so $\sum_{\ball(c, r)\in\balls}|\ball(c, 4r)| = O(n^{1+\eps}/\eps)$;
\item $\sum_{\ball(c, r)\in\balls}\vol(\ball(c, 4r)) = O(m\cdot n^{\eps}/\eps)$.
\end{itemize}
\end{itemize}
\end{lemma}

\begin{remark}
	Compared to the original statement in \cite{bodwin2021better}, their runtime is $O(mn)$, but they have a better bound on sizes. More specifically, they have bounds $\sum_{\ball(c, r)\in\balls}|\ball(c, r/2)| = O(n)$ and $\sum_{\ball(c, r)\in\balls}|\ball(c, 4r)| = O(n^{1+\eps})$ which are smaller than our bounds by a factor of $\epsilon$.
\end{remark}

The proof of \Cref{clustering} is deferred to \Cref{apd: Proof of clustering}.

\paragraph{Consistent paths and distance preservers.} Let $\pi$ be a path and let $x, y$ be two vertices of $\pi$. We denote by $\pi[x, y]$ the subpath of $\pi$ between $x$ and $y$, and we denote by $|\pi|$ the number of edges in $\pi$; we can also define notations $\pi(x, y), \pi(x, y], \pi[x, y)$ in the natural way. Let $\Pi$ be a collection of paths. We say that $\Pi$ is \emph{consistent}, if for any pair $\pi_1, \pi_2\in \Pi$, the intersection between $\pi_1$ and $\pi_2$ is a (possibly empty) subpath of both $\pi_1$ and $\pi_2$.


Let $\pset\subseteq V\times V$ be a set of pairs of vertices in $G$. A subgraph $H\subseteq G$ is a \emph{distance preserver of $G$ with respect to $\pset$}, if $\dist_H(s, t) = \dist_{G}(s, t)$ holds for every $(s, t)\in \pset$. 
We will use the following previous results on distance preservers and $+6$ pairwise additive spanners.

\begin{lemma}[\cite{coppersmith2006sparse}]\label{consist}
	Let $G = (V, E)$ be a graph on $n$ vertices, let $\pset\subseteq V\times V$ be a set of pairs of its vertices, and let $\Pi$ be any consistent collection of paths in $G$ that contains, for each pair $(s,t)\in \pset$, a path $\pi_{s, t}$ connecting $s$ to $t$. Then $|E(\Pi)|=|E(\bigcup_{(s, t)\in \pset}\pi_{s, t})|=O(n+\sqrt{n}\cdot |\pset|)$.
\end{lemma}

It is also implicit in \cite{coppersmith2006sparse} that $|E(\Pi)| = O(n\sqrt{|\pset|})$ which will not be used in in our algorithm.

We use the following corollary of \Cref{consist}, whose proof is a straightforward implementation of the BFS algorithm and the standard edge-weight perturbation technique, and is omitted.

\begin{corollary}
\label{cor: BFS consistent}
There is an algorithm, that given a graph $G$ on $n$ vertices and $m$ edges, and a subset $S$ of its vertices, in time $O(m\cdot |S|)$, computes a consistent collection $\Pi$ of paths that contains, for each pair $s, t\in S$, a shortest path $\pi_{s,t}$ connecting $s$ to $t$ in $G$, such that $|E(\Pi)|=|E(\bigcup_{s, t\in S}\pi_{s, t})|=O(n+\sqrt{n}\cdot|S|^2)$.
\end{corollary}


\begin{lemma}[\cite{kavitha2017new}]\label{pairwise-spanner}
	There is an efficient algorithm, that, given any graph $G = (V, E)$ and any set $\pset\subseteq V\times V$ of pairs, computes a subgraph $H \subseteq G$ with $|E(H)|=O(n|\pset|^{1/4})$, such that for every pair $(s,t)\in \pset$, $\dist_H(s,t)\le \dist_G(s,t)+6$.
\end{lemma}

\section{Pairwise Sublinear Additive Spanners}
\label{sec: pairwise}

In this section, we prove the following \Cref{pairwise-sublinear}, which will serve as a building block for \Cref{sublinear}. 
Throughout this section, we use the following stretch function $f$: for any parameter $\epsilon > 0$ and any integer $k\geq 1$, $f_{k, \epsilon}(d) = d + 2^{30k/\epsilon}d^{1-1/k}$.

\begin{lemma}\label{pairwise-sublinear}
For any undirected unweighted graph $G = (V, E)$ on $n$ vertices, any collection $\pset\subseteq V\times V$ of pairs of its vertices, any integer $k\geq 1$ and parameter $\epsilon\in (0, 0.1)$, there is a subgraph $H\subseteq G$ with $O(2^{2k/\epsilon}n^{1+10k\epsilon}|\pset|^{1 / 2^{k+1}})$ edges, such that for every pair $(s, t)\in \pset$, $\dist_H(s, t)\leq f_{k, \epsilon}(\dist_G(s, t))$.
\end{lemma}

\subsection{Preparation and a subroutine}

We prove \Cref{pairwise-sublinear} by an induction on $k$. The base case (when $k=1$) immediately follows from \Cref{pairwise-spanner}. We assume now that \Cref{pairwise-sublinear} is correct for $1,\ldots, k-1$.
We will design an algorithm that computes a pairwise additive spanner required in \Cref{pairwise-sublinear}. At a high level, our algorithm can be viewed as a combination of the clustering algorithm (\Cref{clustering}) from \cite{bodwin2021better} and the path-buying schemes from \cite{kavitha2017new}. 


For each integer $D\in \set{1,2,2^2,\ldots,2^{\floor{\log n}}}$, we will construct a subgraph $H_D\subseteq G$, such that for all pairs $(s,t)\in \pset$ with $D\le \dist_G(s, t) <2D$, $\dist_{H_D}(s, t)\leq \dist_G(s, t) + 2^{30k/\epsilon}D^{1-1/k}$ holds.
We will then let $H=\bigcup_{0\le i\le \floor{\log n}}H_{2^i}$ to finish the construction.

We now describe the construction of the subgraph $H_D$.
We first apply the algorithm from Lemma \ref{clustering} to graph $G$ with parameters $R=D^{1-1/k}$ and $\epsilon$; for convenience, we will assume $D^{1-1/k}$ is an integer. Let $\clusters$ be the collection of balls we get.
We will also iteratively construct, for each ball $\ball(c,r)\in \balls$, a set $\pset_c\subseteq \ball(c, 4r)\times\ball(c, 4r)$ of pairs, which is initially empty.
Over the course of the algorithm, we maintain the graph $H_D$ as the union of the following edges over all balls $\ball(c, r)\in \clusters$:
\begin{enumerate}[(i)]
	\item a BFS tree $T_c$ that is rooted at $c$ and spans all vertices in $G[\ball(c, 4r)]$;
	\item a pairwise spanner of $G[\ball(c, 4r)]$ with respect to the set $\pset_c$ of pairs, with stretch function $f_{k-1, \epsilon}$ and size $O(2^{2k/\epsilon}|\ball(c, 4r)|^{1+10k\epsilon} |\pset_c|^{1 / 2^{k+1}})$, whose existence is guaranteed by the inductive hypothesis. 
	When the sets $\set{\pset_c}$ change during our construction algorithm, the graph $H_D$ evolves with them.
\end{enumerate}

In order to construct the sets $\set{\pset_c}$, we will take a path-buying approach. Specifically, we will iteratively find a short path that contains relatively few new edges and ``settles'' many pairs of clusters (by connecting them in a near-optimal way). As the construction of $H_D$ is recursive, we will distribute the task of ``buying this path'' to the clusters in $\bset$ as in \cite{bodwin2021better}, which requires us to first chop up a shortest path into short segments. For this, we need the following subroutine called $\pathpart$, which was initially proposed in \cite{bodwin2021better}.

\paragraph{Subroutine \pathpart.} The input to subrountine \pathpart is a shortest path $\pi$ connecting a pair $s,t$ of vertices in a graph $G$ and a collection $\balls$ of balls that covers all vertices in $G$, and the output is a partitioning of path $\pi$ as the sequential concatenation of its subpaths $\pi = \alpha_1\circ\alpha_2\circ\cdots \circ\alpha_l$, such that, for each $1\le i\le l$, if we denote by $s_i, t_i$ the endpoints of $\alpha_i$ (so $\alpha_i=\pi[s_i,t_i]$), then there exists a ball $\ball(c_i, r_i)\in \clusters$, such that $s_i\in \ball(c_i, r_i)$ and $t_i\in \ball(c_i, 2r_i)$, and we say that the ball $\ball(c_i, r_i)$ \emph{hosts} the subpath $\alpha_i$.


We start by directing $\pi$ from $s$ to $t$ and setting $s_1=s$. Since $s_1$ is covered by $\balls$, there exists a ball in $\bset$ that contains $s_1$, and we designate this ball as $\ball(c_1,r_{1})$. We then find the \textbf{last} vertex on $\pi$ that lies in $\ball(c_1,2r_{1})$, and designate it as $t_1$. The first subpath is then defined to be $\alpha_1=\pi[s_1,t_1]$. We then set $s_2=t_1$ and repeat the process to find subpaths $\alpha_2,\ldots,\alpha_l$ until all edges are included in some subpath; that is, for a general index $i\geq 2$, if $s_i\neq t$, then find a ball $\ball(c_i, r_i)\in\bset$ containing $s_i$, and define $t_i\in \pi(s_i, t]\cap \ball(c_i, 2r_i)$ to be the vertex on $\pi$ which is closest to $t$. After that, if $t_i\neq t$, assign $s_{i+1} = t_i$ and repeat.
See Figure \ref{path-partition} for an illustration.

\begin{figure}[h]
	\begin{center}
		\begin{tikzpicture}[thick,scale=0.7]
	\draw (0, 0) node[circle, draw, fill=black!50, inner sep=0pt, minimum width=6pt, label = $s$] {};
	\draw (18, 0) node[circle, draw, fill=black!50, inner sep=0pt, minimum width=6pt,label = $t$] {};
	
	\draw (2, 0) node[circle, draw, fill=black!50, inner sep=0pt, minimum width=6pt,label = $t_1$] {};
	\draw [line width = 0.5mm] (0.2, 0) -- (1.8, 0);
	\draw (0.6, -1) node[red, label={$\alpha_1$}]{};
	\draw [red] (0, 0.5) ellipse (2.3 and 2);
	\draw (0, 2.4) node[red, label={$\ball(c_1, 2r_{1})$}]{};
	
	\draw (5, 0) node[circle, draw, fill=black!50, inner sep=0pt, minimum width=6pt,label = $t_2$] {};
	\draw [line width = 0.5mm] (2.3, 0) -- (4.8, 0);
	\draw (3, -1) node[red, label={$\alpha_2$}]{};
	\draw [red] (3, 0.5) ellipse (2.3 and 2);
	\draw (3, 2.4) node[red, label={$\ball(c_2, 2r_{2})$}]{};
	
	\draw (8, 0) node[circle, draw, fill=black!50, inner sep=0pt, minimum width=6pt,label = $t_3$] {};
	\draw [line width = 0.5mm] (5.3, 0) -- (7.8, 0);
	\draw (6.5, -1) node[red, label={$\alpha_3$}]{};
	\draw [red] (6, 0.5) ellipse (2.3 and 2);
	\draw (6, 2.4) node[red, label={$\ball(c_3, 2r_{3})$}]{};
	
	\draw[dotted] (3.7, 1.2) -- (2.5, -6);
	\draw[dotted] (8.3, 1.2) -- (9.5, -6);
	\draw [fill=orange,orange] (6, -7) ellipse (2 and 2);
	\draw (6, -10) node[red, label={$\ball(c_3, r_{3})$}]{};
	\draw (4.5, -7) node[circle, draw, fill=black!50, inner sep=0pt, minimum width=6pt,label = $t_2$] {};
	\draw (9, -7) node[circle, draw, fill=black!50, inner sep=0pt, minimum width=6pt,label = $t_3$] {};
	\draw [red] (6, -7) ellipse (3.3 and 3.3);
	\draw (6, -11.5) node[red, label={$\ball(c_3, 2r_{3})$}]{};
	\draw (6, 2.4) node[red, label={$\ball(c_3, 2r_{3})$}]{};
	\draw [line width = 0.5mm] (4.7, -7) -- (8.8, -7);
	\draw (6.5, -8) node[red, label={$\alpha_3$}]{};
	\draw [dashed] (2, -7) -- (4.3, -7);
	\draw [dashed] (9.2, -7) -- (11.5, -7);

	\draw (16, 0) node[circle, draw, fill=black!50, inner sep=0pt, minimum width=6pt,label = $t_l$] {};
	\draw [line width = 0.5mm] (17.8, 0) -- (16.2, 0);
	\draw (17.4, -1) node[red, label={$\alpha_l$}]{};
	\draw [red] (17.5, 0.5) ellipse (2.3 and 2);
	\draw (18, 2.4) node[red, label={$\ball(c_l, 2r_{l})$}]{};
	
	\draw (13, 0) node[circle, draw, fill=black!50, inner sep=0pt, minimum width=6pt,label = $t_{l-1}$] {};
	\draw [line width = 0.5mm] (15.8, 0) -- (13.2, 0);
	\draw (14.5, -1) node[red, label={$\alpha_{l-1}$}]{};
	\draw [red] (15, 0.5) ellipse (2.3 and 2);
	\draw (14.5, 2.4) node[red, label={$\ball(c_{l-1}, 2r_{l-1})$}]{};
	
	\draw [dashed] (8.2, 0) -- (12.8, 0); 
\end{tikzpicture}
	\end{center}
	\caption{A partitioning of the shortest path $\pi$ from $s$ to $t$ into subpaths $\alpha_1, \alpha_2, \ldots, \alpha_l$ by balls in $\clusters$. Note that in general $G[\ball(c_i, 2r_i)]$ does not necessarily contain the entire subpath $\alpha_i$.}\label{path-partition}
\end{figure}
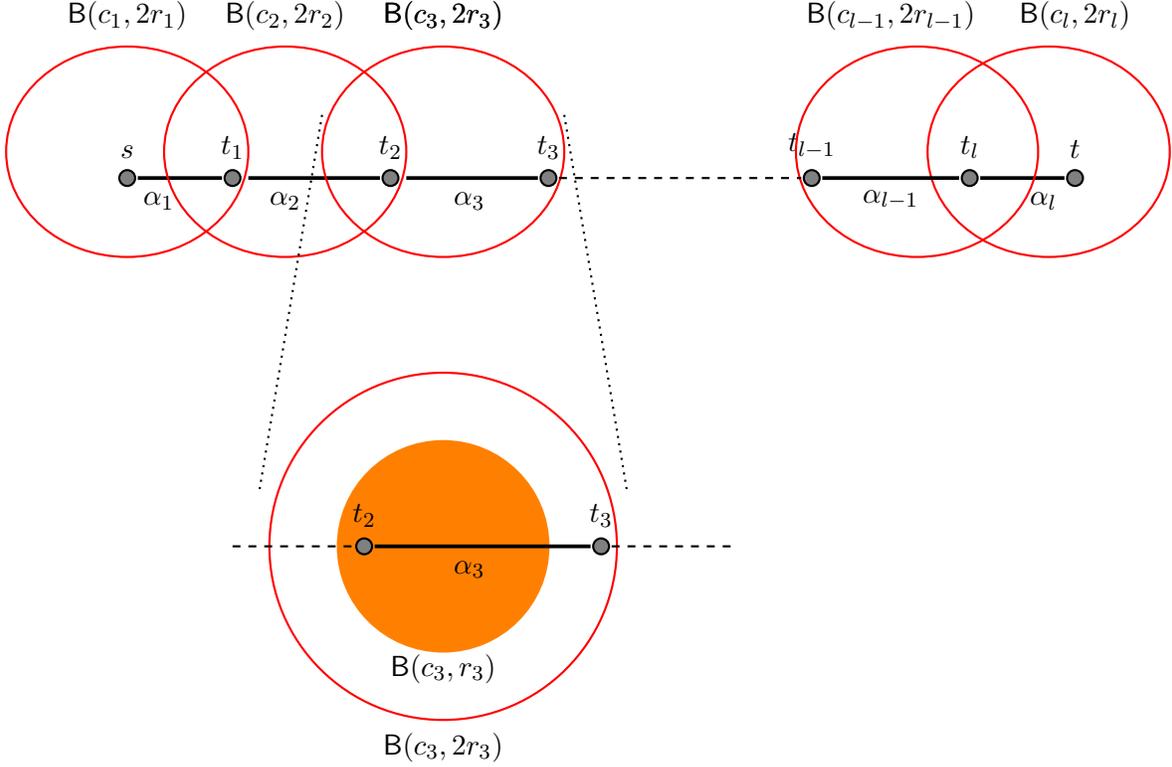

We prove the following simple properties of the subroutine.

\begin{observation}\label{different-balls}
	All balls in $\set{\ball(c_i, r_i)}_{1\leq i\leq l}$ are distinct.
\end{observation}
\begin{proof}
	Suppose otherwise that a ball is selected twice by \pathpart, namely $(c_i, r_i) = (c_j, r_j)$ for some $1\leq i<j\leq l$. Then, by the algorithm description we have $t_j\in \ball(c_j, 2r_j) = \ball(c_i, 2r_i)$, which contradicts the choice of $t_i$ which is the closest-to-$t$ vertex in $\ball(c_i, 2r_i)$ on $\pi$.
\end{proof}

\begin{observation}
For each $1\le i\le l$, the subpath $\alpha_i$ is entirely contained in the ball $\ball(c_i,4r_i)$.
\end{observation}
\begin{proof}
From the description of the subroutine, for each $1\le i\le l$, $s_i\in \ball(c_i, r_i)$ and $t_i\in \ball(c_i, 2r_i)$. Therefore, for every vertex $u\in \pi[s_i, t_i]$, by triangle inequality: $$\dist_{G}(c_i, u)\leq \dist_{G}(c_i, s_i) + \dist_{G}(s_i, t_i)\leq \dist_{G}(c_i, s_i) + \dist_G(c_i, s_i) + \dist_{G}(c_i, t_i)\leq 4r_i,$$
which implies that $u\in \ball(c_i, 4r_i)$.
\end{proof}

\begin{observation}\label{obs}
Let $R$ be the minimum radius of the balls in $\balls$. Then $l \leq |\pi|/R+1$. 
\end{observation}
\begin{proof}
It suffices to show that, for each $1\leq i<l$, $|\alpha_i|\ge R$.
From the algorithm, vertex $s_i$ belongs to the ball $\ball(c_i,r_i)$, and vertex $v_i$ is the last vertex on $\pi$ that lies in the ball $\ball(c_i,2 r_i)$, so if $t_i\neq t$ then $\dist_{G}(t_i,c_i)=2r_i$ must hold. Therefore, 
$|\pi[s_i, t_i]|\geq \dist_{G}(c_i, t_i) - \dist_G(c_i, s_i)\geq 2r_i - r_i = r_i \geq R$.
\end{proof}

\paragraph{Level of balls.} Before we describe the algorithm for constructing the sets $\set{\pset_c}$, we first classify all clusters according to their sizes as follows.
Recall that we are initially given a set $\pset\subseteq V\times V$ of pairs in $G$.
For each integer $1\leq i\leq \ceil{1/\epsilon}$, we say that a ball $\ball(c, r)$ is \emph{at level $i$} if $$n^{1-i\epsilon}/\sqrt{|\pset|}< |\ball(c, r)|\le n^{1-(i-1)\epsilon}/\sqrt{|\pset|}.$$
Denote by $\clusters_i$ the set of all level-$i$ balls in $\clusters$.
We define the set of level-$0$ balls as $\clusters_0 = \{\ball(c, r)\mid |\ball(c, r)|> n/\sqrt{|\pset|}\}$, so $\balls=\clusters_0 \cup \clusters_1\cup\cdots\cup\clusters_{\ceil{1/\eps}}$. A ball from $\clusters$ is called \textbf{large} if it is in $\clusters_0$, and \textbf{small} otherwise. Here is a simple observation.

\begin{observation}\label{obs:Bi-size}
For each $0\leq i\leq \ceil{1/\eps}$, $|\clusters_i|\leq O(n^{(i+1)\epsilon}\sqrt{|\pset|}/\epsilon)$.
\end{observation}
\begin{proof}
	By Lemma \ref{clustering}, $\sum_{\ball(c, r)\in\clusters_i}|\ball(c, r)|\leq \sum_{\ball(c, r)\in\clusters_i}|\ball(c, 4r)|\le O(n^{1+\eps}/\eps)$, and since each ball in $\bset_i$ has size at least $n^{1-i\epsilon} / \sqrt{|\pset|}$, $|\clusters_i|\leq O(n^{(i+1)\epsilon}\sqrt{|\pset|}/\epsilon)$.
\end{proof}


\subsection{Step 1. Handling large balls}

We first take a uniformly random subset $S\subseteq V$ of size $\ceil{10\sqrt{|\pset|}\log n}$. Since each ball in $\clusters_0$ contains at least $n/\sqrt{|\pset|}$ vertices, with high probability, $S$ intersects all balls in $\balls_0$. 

We now proceed to iteratively construct the sets $\set{\pset_c}$ of pairs. Throughout, we maintain, for each  ball $\ball(c,r)\in \bset$, a set $\pset_c$ of pairs of vertices in $\ball(c,2r)$, and another set $U_c\subseteq S$ of vertices in $V$, which are initially empty sets. Intuitively, $U_c$ collects all vertices in $S$ whose distances from the ball center $c$ are already preserved by the spanner $H_D$ during the algorithm.

For every pair $s,t$ of vertices in $S$, we denote by $\pi_{s,t}$ an arbitrary shortest path connecting them in $G$, and we compute the set 
$$\Pi= \{\pi_{s, t}\mid s, t\in S, \text{ } \dist_{G}(s, t) < 2D + 4\cdot 2^{10/\epsilon}\cdot D^{1-1/k}\}.$$

We then process all paths in $\Pi$ sequentially in an arbitrary order. For each path $\pi_{s, t}\in \Pi$, we first apply the subroutine \pathpart to it and the collection $\bset$ of balls, and obtain a partition $\pi_{s,t} = \alpha_1\circ\alpha_2\circ\cdots \circ\alpha_l$. 
For each $1\le i\le l$, we denote by $\ball(c_i, r_i)$ the ball that hosts the subpath $\alpha_i$.
If there exists a ball $\ball(c_i, r_i)$ such that $s, t\in U_{c_i}$, then we do nothing and move on to the next path in $\Pi$. Otherwise, for each $1\le  i\le l$, we add both vertices $s,t$ to the set $U_{c_i}$, and add the pair $(s_i, t_i)$ to the set $\pset_{c_i}$ of pairs.

\paragraph{Stretch analysis of Step 1.}
We make use of the following simple observation.
\begin{observation}\label{large-ball-size}
	At the end of Step 1, for each ball $\ball(c, r)\in \clusters$,  $|\pset_c|\leq |S| = \ceil{10|\pset|^{1/2}\log n}$.
\end{observation}
\begin{proof}
	By the algorithm description along with \Cref{different-balls}, every time a new pair is added to $\pset_c$, a new vertex from $S$ is also added to $U_c$. Since no more pair will be added to $\pset_{c}$ as long as $U_c$ contains all vertices of $S$, we get that $|\pset_c|\leq |U_c|\leq |S| = \ceil{10|\pset|^{1/2}\log n}$.
\end{proof}

We now show in the following claims that, after the first step, all pairwise distances (in $G$) between vertices in $S$ are well-preserved in $H_D$. Therefore, as the set $S$ hits all balls in $\bset_0$ with high probability, the distance between any pair of vertices from $V(\balls_0)$ is also well-preserved.

\begin{claim}\label{C0}
After a vertex $s\in S$ is added to the set $U_c$ for some $\ball(c,r)\in \bset$, the following holds:
	$$\dist_{H_D}(s, c)\leq \dist_G(s, c) + 49\cdot 2^{(30k-10)/\epsilon}\cdot D^{1-1/k}.$$
\end{claim}
\begin{proof}
Assume that the shortest path $\pi_{s,t}$ was being processed when $s$ was added to $U_c$, and that $\ball(c, r)$ was $\ball(c_i, r_i)$ under the notation of the subroutine \pathpart. From the algorithm, for each $1\leq j\leq l$, the pair $(s_j, t_j)$ was added to the collection $\pset_{c_j}$. From the construction of $H_D$, for each $1\leq j\leq l$, there exists a path $\phi_j$ in $H_D$ from $s_j$ to $t_j$, such that: 
	$$\begin{aligned}
	|\phi_j| \leq |\alpha_j| + 2^{30(k-1)/\epsilon}\cdot |\alpha_j|^{1-1/(k-1)} &\leq |\alpha_j| + 2^{30(k-1)/\epsilon}\cdot \left(4\cdot 2^{10/\epsilon}\cdot D^{1-1/k}\right)^{1-1/(k-1)}\\
	&\leq |\alpha_j| + 4\cdot 2^{(30k-20)/\epsilon}\cdot D^{1-2/k}
	\end{aligned}$$
	(we have used the fact that $|\alpha_j| = \dist_G(s_j,t_j)\le \dist_G(s_j,c_j)+ \dist_G(t_j,c_j)\le  4r_{j} \leq 4\cdot 2^{10/\epsilon}\cdot D^{1-1/k}$). Taking a summation over all indices $1\leq j\leq i$ and using  \Cref{obs}, we have:
	$$\begin{aligned}
	\dist_{H_D}(s, c_i)&\leq  \dist_{H_D}(c_i, t_i)+ \dist_{H_D}(s, t_i)\\
	&\leq 2\cdot 2^{10/\epsilon}\cdot D^{1-1/k} + \dist_G(s, t_i) + 48\cdot 2^{(30k-10)/\epsilon}\cdot D^{1-1/k}\\
	&\leq 2\cdot 2^{10/\epsilon}\cdot D^{1-1/k} + (\dist_G(s, c_i) + \dist_G(c_i, t_i)) + 48\cdot 2^{(30k-10)/\epsilon}\cdot D^{1-1/k}\\
	&\leq \dist_G(s, c_i) + 49\cdot 2^{(30k-10)/\epsilon} \cdot D^{1-1/k}.
	\end{aligned}$$
\end{proof}

\begin{claim}\label{C0-ineq}
	For any $s, t\in S$ such that $\dist_G(s, t) < 2D + 4\cdot 2^{10/\epsilon}\cdot D^{1-1/k}$, we have:
	$$\dist_{H_D}(s, t)\leq \dist_G(s, t) + 100 \cdot 2^{(30k-10)/\epsilon}\cdot D^{1-1/k}.$$
\end{claim}
\begin{proof}
For any such pair of vertices $s, t\in S$, consider the moment when the shortest path $\pi_{s, t}\in \Pi$ was processed and partitioned into subpaths $\alpha_1,\ldots,\alpha_l$. We distinguish between the following two cases.
	\begin{itemize}[leftmargin=*]
		\item Case 1. There existed an index $1\le i\le l$ such that $s, t\in U_{c_i}$ at the moment.
		
		In this case, from \Cref{C0}:
		$$\dist_{H_D}(s, c_i)\leq \dist_G(s, c_i) + 49\cdot 2^{(30k-10)/\epsilon} \cdot D^{1-1/k}\leq \dist_G(s, s_i) + 50\cdot 2^{(30k-10)/\epsilon} \cdot D^{1-1/k},$$
		$$\dist_{H_D}(c_i, t)\leq \dist_G(c_i, t) +49\cdot 2^{(30k-10)/\epsilon}\cdot D^{1-1/k}\leq \dist_G(s_i, t) + 50\cdot 2^{(30k-10)/\epsilon}\cdot D^{1-1/k}.$$
		Summing up these two inequalities finishes the proof.
		
		\item Case 2. There did not exist any $i$ such that $s, t\in U_{c_i}$ at the moment.
		
		In this case, for each $1\leq j\leq l$, the pair $(s_j, t_j)$ was added to the set $\pset_{c_j}$. Similar to the proof of \Cref{C0}, for each $1\leq j\leq l$, there exists a path $\phi_j$ in $H_D$ from $s_j$ to $t_j$, such that: 
		$$\begin{aligned}
		|\phi_j| \leq |\alpha_j| + 2^{30(k-1)/\epsilon}\cdot |\alpha_j|^{1-1/(k-1)} &\leq |\alpha_j| + 2^{30(k-1)/\epsilon}\cdot \left(4\cdot 2^{10/\epsilon}\cdot D^{1-1/k}\right)^{1-1/(k-1)}\\
		&\leq |\alpha_j| + 4\cdot 2^{(30k-20)/\epsilon}\cdot D^{1-2/k}.
		\end{aligned}$$
		Taking a summation for all indices $1\leq j\leq l$ and using \Cref{obs}, we get that:
		$$\begin{aligned}
			\dist_{H_D}(s, t) &\leq \dist_G(s, t) + 4\cdot 2^{(30k-20)/\epsilon}\cdot D^{1-2/k}\cdot \brac{2D^{1/2} + 4\cdot 2^{10/\epsilon} + 1}\\
			&\leq \dist_G(s, t) + 28\cdot 2^{(30k-10)/\epsilon}\cdot D^{1-1/k}.
		\end{aligned}$$
	\end{itemize}
\end{proof}

\subsection{Step 2. Handling small balls} 

We now process all pairs $(s,t)\in \pset$ with $D\le \dist_G(s, t) < 2D$ in an arbitrary order. Take any such pair $(s,t)$, and compute a shortest path $\pi_{s,t}$ between them in $G$. We then apply the subroutine \pathpart to it and the collection $\bset$ of balls, and obtain a partitioning $\pi_{s,t} = \alpha_1\circ\alpha_2\circ\cdots \circ\alpha_l$ together with the balls $\set{\ball(c_i, r_i)}_{1\le i\le l}$ that host them. However, as the number of pairs $(s,t)\in \pset$ with $D\le \dist_G(s, t) < 2D$ can be very large, we cannot afford to add the pair $(s_i,t_i)$ to set $\pset_{c_i}$ for all $1\le i\le l$ as in Step 1, and will instead carefully pick a subset of indices in $[1,l]$ to do so.

\paragraph{Basic notations.} We will now describe how to choose these indices. This is again done via an iterative process. Throughout, we will gradually construct a binary tree $\tree$, which initially contains a single node, which is the root of $\tree$. 
Each node of the tree is labeled by a subinterval of $[1, l]$. Initially, the root of $\tree$ is labeled with $[1, l]$. 
Each index $i\in [1,l]$ is either marked \textbf{active} or \textbf{inactive}. Initially, all indices are inactive. The algorithm continues to be executed unless for each leaf of the tree $\tree$, all indices in its associated interval are active. Over the course of the algorithm, whenever we mark some index $i$ active, we will simultaneously add the pair $(s_i,t_i)$ to the set $\pset_{c_{i}}$.
We say that an index $i\in [1,l]$ is \textbf{at level $j$} if its corresponding ball $\ball(c_i,r_i)$ is at level $j$ (that is, if $n^{1-j\epsilon}/\sqrt{|\pset|}< |\ball(c_i, r_i)|\le n^{1-(j-1)\epsilon}/\sqrt{|\pset|}$).

\paragraph{Path-buying with tree structure $\tree$.} As a pre-processing step, we first focus on level-$0$ indices (if there are none of them, then we skip the pre-processing). Let $i_1, i_2$ be the smallest and the largest level-$0$ indices, respectively. We add two new nodes to the tree $\tree$, connecting each of them to the root by an edge. One of the new nodes is labeled by interval $[1, i_1-1]$, and the other node is labeled by interval $[i_2+1, l]$. So now the levels of indices on $[1, i_1-1]$ and $[i_2+1, l]$ are at least $1$. Here we do not modify the active/inactive status of the indices.




We now describe an iteration. We take an arbitrary leaf of the tree $\tree$ whose associated interval contains an inactive index, and denote by $I$ the interval associated with this leaf node. We can assume that all levels of $I$ are at least $1$.
For each $1\le L\le \ceil{1/\eps}$, let $i^L_1 < i^L_2 <\cdots < i^L_{p(L)}$ be all level-$L$ indices in the interval $I$ (so the number of level-$L$ indices in $I$ is $p(L)$); we will \textbf{ignore the superscript $L$} if it is clear from context. Consider the corresponding balls $\set{\ball(c_{i^L_a}, r_{i^L_a})}$ that host them. We say that a pair $\left(\ball(c_{i^L_a}, r_{i^L_a}), \ball(c_{i^L_b}, r_{i^L_b})\right)$ of balls is \textbf{tight}, if 
$$\dist_{H_D}(c_{i^L_a}, c_{i^L_b}) \le \dist_G(c_{i^L_a}, c_{i^L_b}) + \left(3\cdot 2^{(30k-20)/\epsilon} + 10\cdot 2^{10/\epsilon}\right)\cdot D^{1-1/k}.$$
and we denote by $q(L)$ the number of pairs of balls in $\set{\ball(c_{i^L_a}, r_{i^L_a})}$ which are not tight. In addition, we define $\beta(L)=\ceil{n^{(L+1)\epsilon}}$.

We will distinguish between the following two cases.

\paragraph{Case 1. For all $1\le L\le \ceil{1/\eps}$, $p(L) \leq \max\{4\beta(L), 2q(L)/\beta(L) \}$.}
In this case, we activate all indices $i\in I$ (and along with it add the corresponding pair $(s_{i}, t_{i})$ to the set $\pset_{c_{i}}$ and then update $H_D$ accordingly), and continue to the next iteration. We do not modify the tree $\tree$ in this iteration. Note that, after this iteration, the associated interval of the leaf that we processed in this iteration no longer contains any inactive indices, so it will remain as a leaf in $\tree$ forever.

\paragraph{Case 2. There exists some $1\le L\le \ceil{1/\eps}$, such that $p(L) > \max\{4\beta(L), 2q(L)/\beta(L) \}$.}
In this case, 
we activate the $\beta(L)$ smallest and $\beta(L)$ largest level-$L$ indices in $I$ (and along with it add the corresponding pair $(s_{i}, t_{i})$ to the set $\pset_{c_{i}}$ and then update $H_D$ accordingly).

In order to describe the modification of $\tree$ in this iteration, we need the following lemma.

\begin{lemma}
	There exists three level-$L$ indices $i_x, i_y, i_z\in I$ such that:
	\begin{enumerate}[(i)]
		\item $1\leq x\leq \beta(L)$;
		\item $p(L)-\beta(L) < y\leq p(L)$;
		\item $\beta(L) < z\leq p(L)-\beta(L)$;
		\item the pair $\left(\ball(c_{i_x}, r_{i_x}),\ball(c_{i_z}, r_{i_z})\right)$ and the pair $\left(\ball(c_{i_y}, r_{i_y}),\ball(c_{i_z}, r_{i_z})\right)$ of balls are both tight.
	\end{enumerate}
\end{lemma}
\begin{proof}
	Assume for contradiction that there do not exist such three indices. Consider now any pair $a,b$ of indices such that $1\le a\le \beta(L)$ and $\beta(L) < b \leq p(L) - \beta(L)$. From our assumption, for at least one pair of indices $(i^L_a,i^L_{b})$ and $(i^L_b,i^L_{p(L)+1-a})$, their corresponding balls are not tight.
	This implies that $q(L)\geq (p(L) - 2\beta(L))\cdot\beta(L)$. Then either $p(L)<4\beta(L)$ holds, or $q(L)\geq p(L)\cdot \beta(L) / 2$ holds, a contradiction of the initial assumption of Case 2.
\end{proof}

To complete this iteration, we pick indices $i_x, i_y, i_z$ as in the above lemma. We then add two new nodes to the tree $\tree$ (together with an edge connecting to the leaf) as the child nodes of this leaf. 
Denote $I=[i',i'']$, and we label the child nodes with intervals $[i', i_{x}]$ and $[i_{y}, i'']$, respectively. This completes the description of an iteration. See Figure \ref{bridge} for an illustration.

\begin{figure}[h]
	\begin{center}
		\begin{tikzpicture}[thick,scale=0.75]
	\draw [line width = .9mm] (0, 0.2) -- (0, -0.2);
	\draw [line width = .9mm] (1, 0.2) -- (1, -0.2);
	\draw [line width = 0.5mm, color=orange] (0.1, 0) -- (0.9, 0);
	\draw (0.5, -1.2) node[red, label={$i'$}]{};
	
	\draw [line width = 0.5mm] (1.1, 0) -- (2.2, 0);
	
	\draw [line width = .9mm] (2.3, 0.2) -- (2.3, -0.2);
	\draw [line width = .9mm] (3.3, 0.2) -- (3.3, -0.2);
	\draw [line width = 0.5mm, color=orange] (2.4, 0) -- (3.2, 0);
	
	\draw [line width = 0.5mm] (3.4, 0) -- (4.5, 0);
	
	\draw [line width = .9mm] (4.6, 0.2) -- (4.6, -0.2);
	\draw [line width = .9mm] (5.6, 0.2) -- (5.6, -0.2);
	\draw [line width = 0.5mm, color=orange] (4.7, 0) -- (5.5, 0);
	\draw (5.1, -1.2) node[red, label={$i_{x}$}]{};
	\draw [red] (5.1, 1.1) ellipse (1.5 and 1.5);
	\draw (5.1, 1.1) node[circle, draw, fill=black!50, inner sep=0pt, minimum width=6pt, label = $c_{i_{x}}$] {};
	
	\draw [decorate,
	decoration = {brace}] (5.4,-1.3) -- (0.2,-1.3);
	\draw (2.8, -2.5) node[label={activate the first $\beta(L)$ level-$L$ indices}]{};
	
	\draw [line width = .9mm] (20, 0.2) -- (20, -0.2);
	\draw [line width = .9mm] (19, 0.2) -- (19, -0.2);
	\draw [line width = 0.5mm, color=orange] (19.1, 0) -- (19.9, 0);
	\draw (19.5, -1.2) node[red, label={$i''$}]{};
	
	\draw [line width = 0.5mm] (17.8, 0) -- (18.9, 0);
	
	\draw [line width = .9mm] (17.7, 0.2) -- (17.7, -0.2);
	\draw [line width = .9mm] (16.7, 0.2) -- (16.7, -0.2);
	\draw [line width = 0.5mm, color=orange] (16.8, 0) -- (17.6, 0);
	
	\draw [line width = 0.5mm] (15.5, 0) -- (16.6, 0);
	
	\draw [line width = .9mm] (15.4, 0.2) -- (15.4, -0.2);
	\draw [line width = .9mm] (14.4, 0.2) -- (14.4, -0.2);
	\draw [line width = 0.5mm, color=orange] (14.5, 0) -- (15.3, 0);
	\draw (14.9, -1.3) node[red, label={$i_{y}$}]{};
	\draw [red] (14.9, 1.1) ellipse (1.5 and 1.5);
	\draw (14.9, 1.1) node[circle, draw, fill=black!50, inner sep=0pt, minimum width=6pt, label = $c_{i_{y}}$] {};
	
	\draw [decorate,
	decoration = {brace}] (19.8,-1.3) -- (14.6,-1.3);
	\draw (17.2, -2.5) node[label={activate the last $\beta(L)$ level-$L$ indices}]{};
	
	\draw [dashed] (5.7, 0) -- (8.9, 0);
	\draw [line width = .9mm] (9, 0.2) -- (9, -0.2);
	\draw [line width = .9mm] (11, 0.2) -- (11, -0.2);
	\draw [line width = 0.5mm, color=orange] (9.1, 0) -- (10.9, 0);
	\draw (10, -1.4) node[red, label={$i_{z}$}]{};
	\draw [dashed] (11.1, 0) -- (14.3, 0);

	\draw [red] (10, 1.3) ellipse (1.9 and 1.9);
	\draw (10, 1.3) node[circle, draw, fill=black!50, inner sep=0pt, minimum width=6pt, label = $c_{i_{z}}$] {};
	
	\draw [style={decorate, decoration=snake}] (5.3, 1.1) -- (9.8, 1.3);
	\draw [style={decorate, decoration=snake}] (14.7, 1.1) -- (10.2, 1.3);
\end{tikzpicture}
	\end{center}
	\caption{The distances between $c_{i_{x}}, c_{i_{z}}$ and $c_{i_{y}}, c_{i_{z}}$ are approximately preserved in $H$; the orange segments correspond to level-$L$ indices; a prefix and a suffix of level-$L$ indices are activated in this iteration.}\label{bridge}
\end{figure}
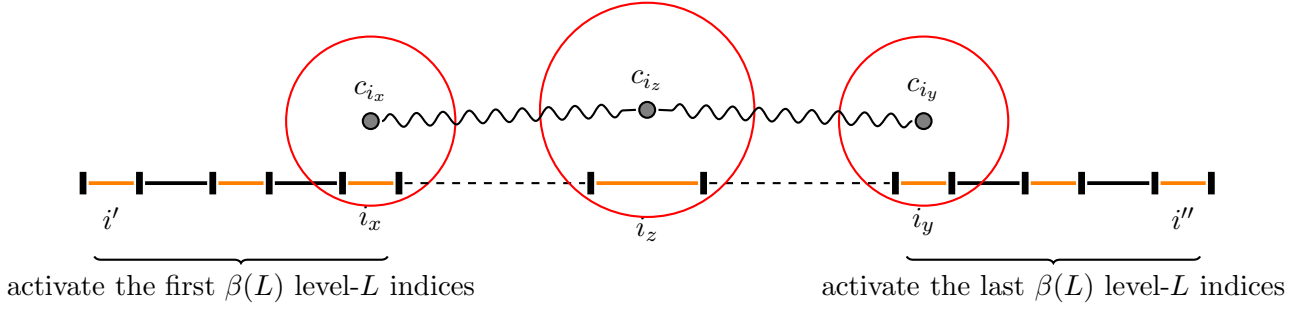

\paragraph{Stretch analysis of Step 2.}
Consider any pair $(s, t)\in \pset$ such that $D\le \dist_G(s, t) < 2D$. Let $\pi$ be the shortest path connecting $s$ to $t$ that was processed in Step 2. We start by showing that, in this iteration, the depth of the tree $\tree$ that we construct is small. 

\begin{observation}
At the end of the iteration of processing $\pi$, $\tree$ has depth at most $\ceil{1/\eps}+1$.
\end{observation}
\begin{proof}
The pre-processing step can increase the depth by at most one. From the algorithm description, it is easy to observe that, as we walk down in any root-to-leaf path in $\tree$, in each step, the number of levels with inactive indices decreases by one. Hence, the observation follows.
\end{proof}

We are now ready to analyze the stretch (in $H_D$) of pairs in $\pset$.

\begin{lemma}
\label{lem: step 2}
For every pair $(s, t)\in \pset$, $\dist_{H_D}(s, t)\leq \dist_G(s, t) + 2^{30k/\epsilon}\cdot D^{1-1/k}$.
\end{lemma}
\begin{proof}
	We will utilize the tree structure of $\tree$. Focus on the time when the construction of $\tree$ was completed. For any tree node $N$, its \emph{height} is defined to be the maximum depth of the subtree of $\tree$ rooted at $N$. We first prove the following claim.
	
	\begin{claim}
	\label{clm: stretch by depth}
		For each non-root tree node $N$ with height $h$ and associated interval $[i', i'']$,
		$$\dist_{H_D}(c_{i'}, c_{i''})\leq \dist_G(c_{i'}, c_{i''}) + 30\cdot (2^{h+1}-1)\cdot 2^{(30k-20)/\epsilon}\cdot D^{1-1/k}.$$
	\end{claim}
	\begin{proof}[Proof of \Cref{clm: stretch by depth}]		
	We prove the claim by induction on $h$. The base case is when $h = 0$ and $N$ is a leaf in $\tree$. From the algorithm, for each $i\in [i', i'']$, the pair $(s_i, t_i)$ is added into the set $\pset_{c_i}$. Similar to the proof of \Cref{C0}, for each $i'\leq i\leq i''$, there exists a path $\phi_i$ in $H_D$, such that 
		$$\begin{aligned}
			|\phi_j| \leq |\alpha_j| + 2^{30(k-1)/\epsilon}\cdot |\alpha_j|^{1-1/(k-1)} &\leq |\alpha_j| + 2^{30(k-1)/\epsilon}\cdot \left(4\cdot 2^{10/\epsilon}\cdot D^{1-1/k}\right)^{1-1/(k-1)}\\
			&\leq |\alpha_j| + 4\cdot 2^{(30k-20)/\epsilon}\cdot D^{1-2/k}.
		\end{aligned}$$
		Taking a summation and using \Cref{obs} (which implies that $l<3D^{1/k}$),
		$$\dist_{H_D}(s_{i'}, t_{i''})\leq \dist_G(s_{i'}, t_{i''}) + 12\cdot 2^{(30k-20)/\epsilon}\cdot D^{1-1/k}.$$
		By triangle inequality, we have:
		$$\dist_{H_D}(s_{i'}, t_{i''})\geq \dist_{H_D}(c_{i'}, c_{i''}) - \dist_{H_D}(c_{i'}, s_{i'}) - \dist_{H_D}(c_{i''}, t_{i''})\geq \dist_{H_D}(c_{i'}, c_{i''}) - 3\cdot 2^{10/\epsilon}D^{1-1/k},$$
		$$\dist_G(s_{i'}, t_{i''})\leq \dist_G(c_{i'}, c_{i''}) + \dist_G(c_{i'}, s_{i'}) + \dist_G(c_{i''}, t_{i''})\leq \dist_G(c_{i'}, c_{i''}) + 3\cdot 2^{10/\epsilon}D^{1-1/k}.$$
		Combining all three inequalities, we have:
		$$\begin{aligned}
			\dist_{H_D}(c_{i'}, c_{i''})&\leq \dist_G(c_{i'}, c_{i''}) + \left(12\cdot 2^{(30k-20)/\epsilon} + 6\cdot 2^{10/\epsilon}\right)\cdot D^{1-1/k}\\
			&\leq \dist_G(c_{i'}, c_{i''}) + 18\cdot 2^{(30k-20)/\epsilon}\cdot D^{1-1/k}.
		\end{aligned}$$
		
		Assume now that the claim is true for $1,\ldots,h-1$. Consider a node $N$ at height $h$ with two children $N_1, N_2$ associated with intervals $[i', i_{x}]$ and $[i_{y}, i'']$ respectively. By inductive hypothesis, we have:
		$$\dist_{H}(c_{i'}, c_{i_{x}})\leq \dist_G(c_{i'}, c_{i_{x}}) + 30\cdot (2^{h}-1)\cdot 2^{(30k-20)/\epsilon}\cdot D^{1-1/k},$$
		$$\dist_{H}(c_{i''}, c_{i_{y}})\leq \dist_G(c_{i''}, c_{i_{y}}) + 30\cdot (2^{h}-1)\cdot 2^{(30k-20)/\epsilon}\cdot D^{1-1/k}.$$
		From the algorithm, the node $N$ may only grow two children $N_1$ and $N_2$ in Case 2, and in this case there exists an index $i_z$ such that the pairs $(\ball(c_{i_x}, r_{i_x}),\ball(c_{i_z}, r_{i_z}))$ and $(\ball(c_{i_y}, r_{i_y}),\ball(c_{i_z}, r_{i_z}))$ of balls are both tight. In other words,
		\[\dist_{H_D}(c_{i_{x}}, c_{i_{z}})\leq \dist_G(c_{i_{x}}, c_{i_{z}}) + \left(3\cdot 2^{(30k-20)/\epsilon} + 10\cdot 2^{10/\epsilon}\right)D^{1-1/k},\]
		\[\dist_{H_D}(c_{i_{z}}, c_{i_{y}})\leq \dist_G(c_{i_{z}}, c_{i_{y}}) + \left(3\cdot 2^{(30k-20)/\epsilon} + 10\cdot 2^{10/\epsilon}\right)D^{1-1/k}.\]
		The sum of the left-hand sides of the above four inequalities is at least $\dist_{H_D}(c_{i'}, c_{i''})$, from triangle inequality. For their right-hand sides, notice that: 
		\[\begin{aligned}
			\dist_G(c_{i'}, c_{i_{x}}) + \dist_G(c_{i_{x}}, c_{i_{z}})&\leq (\dist_G(s_{i'}, s_{i_{x}}) + \dist_G(c_{i'}, s_{i'}) + \dist_G(c_{i_{x}}, s_{i_{x}}))\\
			& +(\dist_G(s_{i_{x}}, s_{i_{z}}) + \dist_G(c_{i_{x}}, s_{i_{x}}) + \dist_G(c_{i_{z}}, s_{i_{z}}))\\
			&\leq \dist_G(s_{i'}, s_{i_{z}}) + 4\cdot 2^{10/\epsilon}\cdot D^{1-1/k}\\
			&\leq \dist_G(c_{i'}, c_{i_{z}}) + 6\cdot 2^{10/\epsilon}\cdot D^{1-1/k}.
		\end{aligned}
		\]
		Symmetrically, we also have $\dist_G(c_{i''}, c_{i_{y}}) + \dist_G(c_{i_{y}}, c_{i_{z}})\leq \dist_G(c_{i''}, c_{i_{z}}) + 6\cdot 2^{10/\epsilon}\cdot D^{1-1/k}$.
		
		Note that
		$\dist_G(c_{i'}, c_{i_{z}}) + \dist_G(c_{i''}, c_{i_{z}})\leq \dist_G(c_{i'}, c_{i''}) + 2\cdot 2^{10/\epsilon}\cdot D^{1-1/k}$ (again by triangle inequality). Altogether, we get that:
		$$\dist_{H_D}(c_{i'}, c_{i''})\leq \dist_G(c_{i'}, c_{i''}) + 30\cdot (2^{h+1}-1)\cdot 2^{(30k-20)/\epsilon}\cdot D^{1-1/k}.$$
	\end{proof}

	Lastly, we consider the root $N_0$ of $\tree$. Recall that it is associated with the interval $[1,l]$. 
	If there are no level-$0$ indices in $[1, l]$, then \Cref{clm: stretch by depth} already implies \Cref{lem: step 2}. We assume from now on that there are level-$0$ indices in $[1, l]$. From the algorithm, in the pre-processing step, $N_0$ grows two children $N_1, N_2$ associated with intervals $[1, i_1-1]$ and $[i_2+1, l]$, respectively, such that $i_1, i_2$ are the smallest and the largest level-$0$ indices. By definition of $S$, there exists $v_1, v_2$ such that $v_1\in (\ball(c_{i_1}, r_{i_1})\cap S), v_2\in (\ball(c_{i_2}, r_{i_2})\cap S)$. Notice that $\dist_G(v_1, v_2)\leq \dist_G(s_{i_1}, s_{i_2}) + 4\cdot 2^{10/\epsilon}\cdot D^{1-1/k} < 2D + 4\cdot 2^{10/\epsilon}\cdot D^{1-1/k} $. Then, from \Cref{C0-ineq}, we have:
	$$\dist_{H_D}(v_1, v_2)\leq \dist_G(v_1, v_2) + 100\cdot 2^{(30k-10)/\epsilon} D^{1-1/k}.$$
	
	Suppose $N_1, N_2$ have height $h_1, h_2\leq \ceil{1/\eps}$ respectively. By \Cref{clm: stretch by depth},
	$$\dist_{H_D}(c_1, c_{i_1-1})\leq \dist_G(c_1, c_{i_1-1}) + 30\cdot (2^{h_1+1}-1)\cdot 2^{(30k-20)/\epsilon}\cdot D^{1-1/k},$$
	$$\dist_{H_D}(c_l, c_{i_2+1})\leq \dist_G(c_l, c_{i_2+1}) + 30\cdot (2^{h_2+1}-1)\cdot 2^{(30k-20)/\epsilon}\cdot D^{1-1/k}.$$
Taking a summation of the left-hand sides of the above three inequalities:
	\[\begin{aligned}
		&\dist_{H_D}(c_1, c_{i_1-1}) + \dist_{H_D}(c_l, c_{i_2+1}) + \dist_{H_D}(v_1, v_2)\\
		& \geq \dist_{H_D}(c_1, c_{i_1-1}) + \dist_{H_D}(c_l, c_{i_2+1}) + \dist_{H_D}(c_{i_1}, c_{i_2}) - 2\cdot 2^{10/\epsilon}\cdot D^{1-1/k}\\
		&\geq \dist_{H_D}(c_1, c_l) - 8\cdot 2^{10/\epsilon}\cdot D^{1-1/k}\\
		&\geq \dist_{H_D}(s, t) - 11\cdot 2^{10/\epsilon} D^{1-1/k}.
	\end{aligned}\]
Taking a summation of their right-hand sides (ignoring the tails for now):
	\[\begin{aligned}
		\dist_G(c_1, c_{i_1-1}) + \dist_G(c_l, c_{i_2+1}) + \dist_G(v_1, v_2) & \leq \dist_G(s, s_{i_1-1}) + (\dist_G(s, c_1) + \dist_G(c_{i_1-1}, s_{i_1}))\\
		&+\dist_G(t, s_{i_2+1}) + (\dist_G(t, c_l) + \dist_G(s_{i_2+1}, c_{i_2}))\\
		&+ \dist_G(s_{i_1}, s_{i_2}) + (\dist_G(v_1, s_{i_1}) + \dist_G(v_2, s_{i_2}))\\
		&\leq \dist_G(s, t) + 9\cdot 2^{10/\epsilon}\cdot D^{1-1/k}.
\end{aligned} \]
Altogether, $\dist_{H_D}(s, t)\leq \dist_G(s, t) + 200\cdot 2^{(30k-10)/\epsilon}D^{1-1/k} < \dist_G(s, t) + 2^{30k/\epsilon}D^{1-1/k}$; recall that we are assuming $\epsilon \in (0, 0.1)$ and so $2^{10/\epsilon} > 200$.
\end{proof}

\subsection{Size analysis}

In this subsection, we complete the proof of \Cref{pairwise-sublinear} by showing that the spanner $H$ constructed by the algorithm in this section satisfies the desired size bound. We start by proving the following claim.

\begin{claim}
\label{clm: pairs added}
In the end, for each level $1\leq i\leq \ceil{1/\eps}$, $\sum_{\ball(c, r)\in \clusters_i}|\pset_c|\leq O(\frac{1}{\epsilon^2}2^{\ceil{1/\epsilon}}n^{(i+1)\epsilon}|\pset|\log n)$.
\end{claim}
\begin{proof}

By \Cref{large-ball-size}, all sets $\pset_c$ have size $O(|\pset|^{1/2}\log n)$ before Step 2 begins. Hence, at the time when Step 1 finishes, we have $\sum_{\ball(c, r)\in \clusters_i}|\pset_c|\leq O(\frac{1}{\epsilon}n^{(i+1)\epsilon}|\pset|\log n)$. So, we only need to analyze how much $|\pset_c|$ has increased during Step 2. 

Define $\Phi_i$ to be the number of  pairs of level-$i$ balls that are not tight. So, at the beginning of Step 2, we have $\Phi_i\le |\clusters_i|^2\leq O(n^{2(i+1)\epsilon}|\pset|/\epsilon^2)$ using \Cref{obs:Bi-size}.

Consider the iteration in Step 2 when some pair $(s,t)\in \pset$ was being processed, and let $\tree$ be the binary tree constructed in that iteration. Let $N$ be a node of $\tree$. If the algorithm added more than $4\beta(i)$ pairs to $\bigcup_{\ball(c, r)\in\clusters_i}\pset_c$ in the round that $N$ was processed, then it must be through Case 1 and $N$ would remain a leaf in $\tree$ till the end. Put in other way, at most $4\beta(i)$ pairs were added to the collection $\bigcup_{\ball(c, r)\in\clusters_i}\pset_c$ in processing every non-leaf node in $\tree$.
Consider now a leaf node $N$ of $\tree$.
Using similar analysis in the proof of \Cref{clm: stretch by depth}, we can show that after the round of processing $N$,
every pair of level-$i$ balls that host some segment of the shortest path processed in that iteration became tight, and so $\Phi_i$ would decrease by at least $p(i)\cdot \beta(i)/2$; recall that $p(i)$ is the number of pairs added to the path collection $\bigcup_{\ball(c, r)\in\clusters_i}\pset_c$. Therefore, comparing with the decrease in $\Psi_i$ over the course of the iteration, the sum $\sum_{\ball(c, r)\in \clusters_i}|\pset_c|$ increased by at most: $$2^{\ceil{1/\eps}+2}\cdot 4 \beta(i) + 2\Delta\Phi_i / \beta(i)$$
where $\Delta\Phi_i$ was the total decrease of $\Phi_i$ in this iteration. Therefore, during Step 2, the total amount $\sum_{\ball(c, r)\in \clusters_i}|\pset_c|$ has increased by at most: 
$$|\pset|\cdot 2^{\ceil{1/\eps}+2}\cdot 4\beta(i) + 8|\clusters_i|^2 / \beta(i) = O(2^{\ceil{1/\epsilon}}n^{(i+1)\epsilon}|\pset|/\epsilon^2)$$
%
\end{proof}

\begin{lemma}
	\label{lem: size of dist D spanner}
In the end, the spanner $H_D$ contains at most $\tilde{O}(2^{2k/\epsilon}\cdot n^{1 + (10k-8)\epsilon}|\pset|^{1/2^{k+1}})$ edges.
\end{lemma}
\begin{proof}
	For level $0$, for each ball $\ball(c, r)\in \clusters_0$, since $\pset_c$ stays unchanged during Step 2, by \Cref{large-ball-size}, we have $|\pset_c|\leq O(|\pset|^{1/2}\log n)$ at the end of the algorithm. Therefore, by the inductive hypothesis on sublinear additive spanners with stretch function $f_{k-1, \epsilon}(\cdot)$, the number of edges added to $H$ by balls in $\clusters_0$ is asymptotically (up to an $O(\log n)$ factor) bounded by:
	$$\sum_{\ball(c, r)\in \clusters_0}2^{2(k-1)/\epsilon}\cdot|\ball(c, 4r)|^{1+10(k-1)\epsilon}\cdot |\pset_c|^{1/2^k}\leq \tilde{O}(2^{2k}n^{1+(10k-8)\epsilon}|\pset|^{1/2^{k+1}}).$$
	
	For each level $1\leq i\leq \ceil{1/\eps}$, by the inductive hypothesis on sublinear additive spanners with stretch function $f_{k-1, \epsilon}(\cdot)$, the number of edges added to $H$ by balls in $\clusters_i$ is asymptotically (up to an $O(\log n)$ factor) bounded by:
	\[\begin{aligned}
		&\sum_{\ball(c, r)\in \clusters_i}2^{2(k-1)/\epsilon}\cdot|\ball(c, 4r)|^{1+10(k-1)\epsilon}\cdot |\pset_c|^{1/2^k}\\
		&\leq 2^{2(k-1)/\epsilon}\cdot\left(\frac{n^{1-(i-1)\epsilon}}{|\pset|^{1/2}}\right)^{1+10(k-1)\epsilon}\cdot\sum_{\ball(c, r)\in \clusters_i}|\pset_c|^{1/2^{k}}\\
		&\leq 2^{2(k-1)/\epsilon}\cdot\frac{n^{1-(i-1)\epsilon + 10(k-1)\epsilon}}{|\pset|^{1/2}}\cdot |\clusters_i|\cdot \left(\frac{\sum_{\ball(c, r)\in \clusters_i}|\pset_c|}{|\clusters_i|}\right)^{1/2^{k}}\\
		&\leq 2^{2(k-1)/\epsilon}\cdot\frac{n^{1-(i-1)\epsilon + 10(k-1)\epsilon}}{|\pset|^{1/2}}\cdot O(n^{(i+1)\epsilon}|\pset|^{1/2}/\epsilon^2)\cdot \tilde{O}\left(2^{\ceil{1/\epsilon}}|\pset|^{1/2}\right)^{1/2^{k}}\\
		&\leq \tilde{O}(2^{2k/\epsilon}\cdot n^{1 + (10k-8)\epsilon}|\pset|^{1/2^{k+1}}).
	\end{aligned}\]
where the second inequality uses concavity of $x^{1/2^k}$ and Jensen's inequality, and the third inequality uses \Cref{clm: pairs added}, and the last inequality uses the fact that $2^{\frac{1}{\epsilon} - \ceil{\frac{1}{\epsilon}}^{1/2^k}} > 1 / \epsilon^2$ which holds as $k\geq 1$ and $\epsilon < 0.1$. Summing over all different levels $0\leq i\leq \ceil{1/\epsilon}$, we get that $|E(H_D)|=\tilde{O}(2^{2k/\epsilon}\cdot n^{1+(10k-8)\epsilon}|\pset|^{1/2^{k+1}})$.
\end{proof}

From \Cref{lem: size of dist D spanner} and by summing over all $D\in \set{1,2,\ldots,2^{\floor{\log n}}}$, we get that  $|E(H)|\le \tilde{O}(2^{2k/\epsilon}\cdot n^{1+(10k-8)\epsilon}\cdot |\pset|^{1/2^{k+1}}) = O(2^{2k/\epsilon}n^{1+10k\epsilon}|\pset|^{1/2^{k+1}})$, which completes the size analysis by induction.
\section{Almost Optimal Sublinear Additive Spanners}
\label{sec: sublinear spanner}

In this section, we provide the proof of \Cref{sublinear}.
In fact, we will prove the following theorem.

\begin{theorem}
\label{sublinear_real}
For any undirected unweighted graph $G = (V, E)$ on $n$ vertices, any parameter $\epsilon\in (0, 0.1)$ and any integer $k\ge 1$, there is a subgraph $H\subseteq G$ with $|E(H)|\le O\left(n^{1 + 10k\eps +\frac{1}{2^{k+1} - 1}}\right)$, such that for every pair $u,v\in V(G)$, $\dist_{H}(u,v)\le \dist_{G}(u,v)+2^{30k/\epsilon}\cdot \dist_{G}(u,v)^{1-1/k}$.
\end{theorem}

Note that, for any constant parameter $\delta>0$, if we let $\eps=\frac{\delta}{10k(2^{k+1}-1)}$, then \Cref{sublinear_real} gives a sublinear spanner with stretch function $f(d)=d+2^{O(k^2 2^k/\delta)}\cdot d^{1-1/k}$ on $O(n^{1+\frac{1+\delta}{2^{k+1}-1}})$ edges, which implies \Cref{sublinear}.

In the remainder of this section, we provide the proof of \Cref{sublinear_real} by induction on $k$, with the help of \Cref{pairwise-sublinear}. 
The base case is when $k=1$. We note that it was shown by \cite{baswana2010additive} that any undirected unweighted graph admits an $+6$-additive spanner of size $O(n^{4/3})$, and the base case follows from this result. Assume from now on that  \Cref{sublinear_real} is true for $1,\ldots,(k-1)$.
Similar to \Cref{sec: pairwise}, we denote $f_{k, \epsilon}(d) = d + 2^{30k/\epsilon}\cdot d^{1-1/k}$ for brevity.

 
Similar to the algorithm in \Cref{sec: pairwise}, for each $D\in \set{1,2,2^2,\ldots,2^{\floor{\log n}}}$, we will construct a subgraph $H_D\subseteq G$, such that for all pairs $s,t\in V(G)$ with $D\le \dist_G(s, t) <2D$, $\dist_H(s, t)\leq \dist_G(s, t) + O(D^{1-1/k})$ holds.
We will then let $H=\bigcup_{0\le i\le \floor{\log n}}H_{2^i}$ to finish the construction.
For convenience, we assume that $D^{1-1/k}$ is an integer\footnote{Our computation will in fact use $\floor{D^{1-1/k}}$. For notational convenience we omit it and just write $D^{1-1/k}$ instead.}.

\subsection{Algorithm description}

We now describe the construction of graph $H_D$. We first apply the algorithm of \Cref{clustering} to $G$ with parameters $R = D^{1-1/k}$ and $\epsilon$. Let $\clusters$ be the set of balls we obtain. Let $L_k$ be a threshold value to be determined later. We say that a ball $\ball(c, r)\in\clusters$ is \emph{small} if $|\ball(c, r)| \leq L_k$, otherwise we say it is \emph{large}.

Similar to the algorithm in \Cref{sec: pairwise}, the spanner $H_D$ is the union of the following graphs:
\begin{enumerate}[(i)]
	\item for each ball $\ball(c, r)\in\clusters$, a BFS tree that is rooted at $c$ and spans all vertices in $G[\ball(c, 4r)]$;
	
	\item for each small ball $\ball(c, r)\in \clusters$, a spanner of the induced subgraph $G[\ball(c, 4r)]$, with stretch function $f_{k-1, \epsilon}$ and size $O\left(|\ball(c, 4r)|^{1 + 10(k-1)\eps +\frac{1}{2^{k} - 1}}\right)$, whose existence is guaranteed by the inductive hypothesis;
	
	\item for each ball $\ball(c, r)\in\clusters$, a pairwise spanner with respect to a collection $\pset_c$ of pairs in $\ball(c, 2r)$, with stretch function $f_{k-1, \epsilon}$ and size $O(2^{2(k-1)/\epsilon}|\ball(c, 4r)|^{1+10(k-1)\epsilon}|\pset_c|^{1 / 2^{k}})$, whose existence is guaranteed by the inductive hypothesis; we will guarantee that for each demand pair $(s, t)\in \pset_c$, any $s$-$t$ shortest path in $G$ lies in the induced subgraph $G[\ball(c, 4r)]$. The construction of sets $\pset_c$ is iterative and described next.
\end{enumerate}


Let $S$ be a random subset of $V$ of size $\ceil{\frac{10n}{L_k}\log n}$, so with high probability, $S$ intersects all large balls in $\balls$.
We then compute $\Pi = \{\pi_{s, t}\mid s, t\in S, \dist_{G}(s, t) < 2D + 4\cdot 2^{10/\epsilon}D^{1-1/k}\}$, where $\pi_{s,t}$ is an $s$-$t$ shortest path in $G$.
We now proceed to iteratively construct the sets $\set{\pset_c}$ of pairs. Throughout, we maintain, for each large  ball $\ball(c,r)\in \bset$, a set $\pset_c$ of pairs of vertices in $\ball(c,2r)$, and another set $U_c$ of vertices in $V$ which intuitively contains all vertices that are ``settled with the ball $\ball(c, r)$''.

We then process all paths in $\Pi$ sequentially in an arbitrary order. For each path $\pi_{s, t}\in \Pi$, we first apply the subroutine \pathpart to it and the collection $\bset$ of balls, and obtain a partitioning $\pi_{s,t} = \alpha_1\circ\alpha_2\circ\cdots \circ\alpha_l$. 
For each $1\le i\le l$, we denote by $\ball(c_i, r_i)$ the ball that hosts the subpath $\alpha_i$.


If either all the balls $\ball(c_1, r_{1}), \ball(c_2, r_{2}), \ldots, \ball(c_l, r_{l})$ are small, or there exists a large ball $\ball(c_i, r_{i})$ such that $s, t\in U_{c_i}$, then we do nothing and move on to the next path in $\Pi$. Otherwise, for each large ball $\ball(c_i, r_{i})$, we add vertices $s,t$ to set $U_{c_i}$, and add the pair $(s_i, t_i)$ to set $\pset_c$.
This completes the description of the construction of sets $\set{\pset_c}$, and also finishes the description of the algorithm.

Before we proceed to the size and stretch analysis, we prove the following simple observation.

\begin{observation}\label{demand-size}
In the end, for each ball $\ball(c, r)\in \clusters$, $|\pset_c|\leq |S| = \ceil{\frac{10n}{L_k}\log n}$.
\end{observation}
\begin{proof}
	By the algorithm, each time a new pair is added to $\pset_c$, $|U_c|$ also increases by at least one. Therefore $|\pset_c|\leq |U_c|\leq |S| = \ceil{\frac{10n}{L_k}\log n}$.
\end{proof}


\subsection{Stretch analysis}

The stretch analysis of the algorithm in \Cref{sec: sublinear spanner} is quite similar to the Step 1 stretch analysis in \Cref{sec: pairwise} (\Cref{C0} and \Cref{C0-ineq}). We start with the following claims.

\begin{claim}\label{path-buy}
In the end, for each large ball $\ball(c, r)\in \balls$ and each vertex $s\in U_c$,
$$\dist_{H_D}(s, c)\leq \dist_{G}(s, c) + 50\cdot 2^{(30k-10)/\epsilon}D^{1-1/k}.$$
\end{claim}
\begin{proof}
Assume that the shortest path $\pi_{s,t}$ was being processed when $s$ was added to $U_c$, and that $\ball(c, r)$ was $\ball(c_i, r_i)$ under the notation of the subroutine \pathpart. 
Next, we will construct a short path from $s$ to $c$ in $H_D$. For each $1\leq j<i$, consider the shortest path $\alpha_j = \pi[s_j, t_j]$ and the ball $\ball(c_j, r_{j})$. By \Cref{obs}, $\alpha_j$ lies entirely within $G[\ball(c_j, 4r_{j})]$. Let $\phi_j$ be a shortest path from $s_j$ to $t_j$ in $H_D$. 
We distinguish between the following cases.
\begin{itemize}[leftmargin=*]
	\item $\ball(c_j, r_{j})$ is small.
	
	Recall that graph $H_D$ contains a sublinear spanner of the induced subgraph $G[\ball(c_j, 4r_{j})]$ with stretch function $f_{k-1,\eps}$. Therefore,  
	$$\begin{aligned}
		|\phi_j| \leq |\alpha_j| + 2^{30(k-1)/\epsilon}\cdot |\alpha_j|^{1-1/(k-1)} &\leq |\alpha_j| + 2^{30(k-1)/\epsilon}\cdot \left(4\cdot 2^{10/\epsilon}\cdot D^{1-1/k}\right)^{1-1/(k-1)}\\
		&\leq |\alpha_j| + 4\cdot 2^{(30k-20)/\epsilon}D^{1-2/k}.
	\end{aligned}$$

	\item $\ball(c_j, r_{j})$ is large.
	
	From the algorithm description, the pair $(s_j, t_j)$ was added to set $\pset_c$ in this iteration, and $\alpha_j$ is contained entirely within $G[\ball(c_j, 4r_{j})]$ (from \Cref{obs}). Therefore, 
	$$|\phi_j| \leq f_{k-1, \epsilon}(|\alpha_j|) \leq |\alpha_j| + 4\cdot 2^{(30k-20)/\epsilon}D^{1-2/k}.$$
\end{itemize}

As $t_{i-1} = s_i\in \ball(c_i, r_{i})$ holds for every $1\leq j<i$, we have $\dist_G(c, t_{i-1})\leq r \leq 2^{10/\epsilon}D^{1-1/k}$, and therefore
	$$\begin{aligned}
		\sum_{j=1}^{i-1}|\phi_j| &\leq \sum_{j=1}^{i-1}\left(|\alpha_j| + 4\cdot 2^{(30k-20)/\epsilon}D^{1-2/k}\right)\leq \sum_{j=1}^{i-1}|\alpha_j| + O_\epsilon(D^{1-1/k})\\
		&\leq \dist_G(s, c) + \dist_G(c, t_{i-1}) + i\cdot 4\cdot 2^{(30k-20)/\epsilon}D^{1-2/k}\\
		&< \dist_G(s, c) + 49\cdot 2^{(30k-10)/\epsilon}D^{1-1/k}.
	\end{aligned}$$
Finally, we let $\phi_i$ be an arbitrary shortest path connecting $t_{i-1}$ to $c$ in $H_D$. Since $H_D$ contains a breadth-first search tree rooted at $c$ that spans all vertices in $\ball(c, 4r)$, $|\phi_i|\leq 4r\leq 4\cdot 2^{10/\epsilon}D^{1-1/k}$. Therefore, $\rho = \phi_1\circ\phi_2\circ\cdots\circ\phi_i$ is a path in $H_D$ that connects $s$ and $c$, and moreover,
	$$\dist_{H_D}(s, c)\leq \sum_{j=1}^{i}|\phi_j| \leq \dist_G(s, c) + 50\cdot 2^{(30k-10)/\epsilon}D^{1-1/k}.$$
\end{proof}

\begin{claim}\label{hitset}
	For any pair of vertices $s, t\in S$ such that $\dist_G(s, t) < 2D +  4\cdot 2^{10/\epsilon}D^{1-1/k}$, 
	$$\dist_{H_D}(s, t)\leq \dist_G(s, t) + 101\cdot 2^{(30k-10)/\epsilon}D^{1-1/k}.$$
\end{claim}
\begin{proof}
For any such pair of vertices $s, t\in S$, consider the moment when the shortest path $\pi_{s, t}\in \Pi$ was processed and partitioned into $\alpha_1,\ldots,\alpha_l$. We distinguish between the following two cases.
	\begin{itemize}[leftmargin=*]
		\item There existed an index $1\le i\le l$ such that $s, t\in U_{c_i}$ at the moment.
		
		In this case, by \Cref{path-buy}, $\dist_{H_D}(s, c_i)\leq \dist_G(s, c_i) + 50\cdot 2^{(30k-10)/\epsilon}D^{1-1/k}$, and $\dist_{H_D}(c_i, t)\leq \dist_G(c_i, t) + 50\cdot 2^{(30k-10)/\epsilon}D^{1-1/k}$. By triangle inequality,
		$$\begin{aligned}
			\dist_{G}(s, c_i) + \dist_{G}(c_i, t)&\leq (\dist_G(s, s_i) + \dist_G(s_i, c_i)) + (\dist_G(s_i, t) + \dist_G(s_i, c_i))\\
			&\leq \dist_G(s, t) + 2\cdot 2^{10/\epsilon}D^{1-1/k}.
		\end{aligned}$$
		Therefore, $\dist_{H_D}(s, t)\leq \dist_{H_D}(s, c_i) + \dist_{H_D}(c_i, t)\leq \dist_G(s, t) +101\cdot 2^{(30k-10)/\epsilon}D^{1-1/k}$.
		
		\item There did not exist any $i$ such that $s, t\in U_{c_i}$ at the moment.
		
		In this case, for each $1\leq j\leq l$, the pair $(s_j, t_j)$ is added to the set $\pset_{c_j}$ after this iteration. According to the algorithm, in the resulting graph $H_D$, for each $1\leq i\leq l$, there is a path $\phi_i$ in ${H_D}$ between $s_i, t_i$ such that $|\phi_i|\leq |\alpha_i| + 4\cdot 2^{(30k-20)/\epsilon}D^{1-2/k}$. By  \Cref{obs},
		$$\begin{aligned}
			\dist_{H_D}(s, t)&\leq \sum_{i=1}^l|\phi_i|\leq \sum_{i=1}^l\left(|\alpha_i| + 4\cdot 2^{(30k-20)/\epsilon}D^{1-2/k}\right)\\
			&\leq \dist_G(s, t) + 48\cdot 2^{(30k-10)/\epsilon}D^{1-1/k}.
		\end{aligned}$$
	\end{itemize}
\end{proof}

In the following lemma, we complete the analysis of stretch of the graph $H_D$ on pairs of vertices in $G$ at distance at most $2D$.

\begin{lemma}
	For any pair of vertices $s, t\in V$ such that $\dist_G(s, t) < 2D$, we have: 
	$$\dist_{H_D}(s, t)\leq \dist_G(s, t) + 2^{30k/\epsilon}D^{1-1/k}.$$
\end{lemma}
\begin{proof}
	Let $\pi$ be an $s$-$t$ shortest path. We apply the subroutine \pathpart to path $\pi$ and the set $\balls$  of balls and obtain a partition $\pi = \alpha_1\circ\alpha_2\circ\cdots\circ\alpha_l$ and balls $\ball(c_1, r_{1}), \ball(c_2, r_{2}), \ldots, \ball(c_l, r_{l})$. If all these balls are small, note that for each  $1\leq i\leq l$, the subpath $\alpha_i$ lies entirely in $G[\ball(r_i, r_{i})]$, so there exists a path $\phi_i$ in $H_D$, such that
	$$\begin{aligned}
		|\phi_i| \leq |\alpha_i| + 2^{30(k-1)/\epsilon}\cdot |\alpha_i|^{1-1/(k-1)} &\leq |\alpha_i| + 2^{30(k-1)/\epsilon}\cdot \left(4\cdot 2^{10/\epsilon}\cdot D^{1-1/k}\right)^{1-1/(k-1)}\\
		&\leq |\alpha_i| + 4\cdot 2^{(30k-20)/\epsilon}D^{1-2/k}.
	\end{aligned}$$
	Since $l \leq \ceil{\frac{|\pi|}{D^{1-1/k}}} < 2D^{1/k} + 1 < 3D^{1/k}$, we get that $\dist_{H_D}(s, t)\leq \dist_G(s, t) + 12\cdot 2^{(30k-20)\epsilon}D^{1-1/k}$.
	
	Assume from now on that, some ball among $\ball(c_1, r_{1}), \ball(c_2, r_{2}), \ldots, \ball(c_l, r_{l})$ is large. Let $1\le x\le l$ be the smallest index of a large ball, and let $1\le y\le l$ be the largest index of a large ball. By the hitting set property, there exist $u, v\in S$ such that $u\in \ball(c_x, r_{x}), v\in \ball(c_y, r_{y})$. By triangle inequality,
	$$\dist_G(u, v)\leq \dist_G(c_x, c_y) + \dist_G(c_x, u) + \dist_G(c_y, v) < 2D + 4\cdot 2^{10/\epsilon}D^{1-1/k}.$$
	From \Cref{hitset}, $\dist_{H_D}(u, v)\leq \dist_G(u, v) + 101\cdot 2^{(30k-10)/\epsilon}D^{1-1/k}$.
	
	For each $j\in [1, x)\cup (y, l]$, since $\ball(c_j, r_{j})$ is small, there is a path $\phi_j$ in $H_D$ connecting $s_j$ to $t_j$, such that
	$$|\phi_j|\leq |\alpha_j| + 2^{30(k-1)/\epsilon}\cdot |\alpha_j|^{1-1/(k-1)} \leq |\alpha_j| + 4\cdot 2^{(30k-20)/\epsilon}D^{1-2/k}.$$
	Then,
	$$\dist_{H_D}(s, u) + \dist_{H_D}(v, t)\leq \dist_G(s, u) + \dist_G(v, t) + 12\cdot 2^{(30k-10)/\epsilon}D^{1-1/k}.$$
	Finally, by triangle inequality,
	$$\begin{aligned}
		\dist_{H_D}(s, t) &\leq \dist_G(s, u) + \dist_G(u, v) + \dist_G(v, t) + 113\cdot 2^{(30k-10)/\epsilon}D^{1-1/k}\\
		&\leq (\dist_G(s, s_x) + \dist_G(s_x, u))\\
		&+(\dist_G(s_x, s_y) + \dist_G(s_x, u) + \dist_G(s_y, v))\\
		&+(\dist_G(s_y, t) + \dist_G(v, s_y)) + 113\cdot 2^{(30k-10)/\epsilon}D^{1-1/k}\\
		&\leq \dist_G(s, t) + 8\cdot 2^{10/\epsilon}D^{1-1/k} + 113\cdot 2^{(30k-10)/\epsilon}D^{1-1/k}\\
		&\leq \dist_G(s, t) + 2^{30k/\epsilon}D^{1-1/k}.
	\end{aligned}$$
\end{proof}

\subsection{Size analysis}

From the algorithm, the edge set of $H_D$ contains three types of edges that are calculated below.
\begin{enumerate}[(1),leftmargin=*]
	\item For each ball $\ball(c, r)\in \clusters$, graph $H_D$ contains a BFS tree $T_c$ that is rooted at $c$ and spans all vertices in $\ball(c, 4r)$. From \Cref{clustering}, $\sum |E(T_c)|=\sum |\ball(c, 4r)|=O(n^{1+\epsilon}/\eps)$.
	\item For each small ball $\ball(c, r)\in \clusters$, graph $H_D$ contains a sublinear additive spanner of the subgraph $G[\ball(c, 4r)]$ with stretch function $f_{k-1, \epsilon}$, which contains $O(|\ball(c, 4r)|^{1+10(k-1)\epsilon + \frac{1}{2^k-1}})$ edges by inductive hypothesis. Summing over all small balls, the number of edges in these spanners is at most (ignoring constant factors):
	$$\sum_{\ball(c, r)\text{ is small}}|\ball(c, 4r)|^{1+10(k-1)\epsilon + \frac{1}{2^k-1}}\leq n^{1+(10k-5)\epsilon}\cdot L_k^{\frac{1}{2^k-1}}.$$
	\item For each large ball $\ball(c, r)\in \clusters$, graph $H_D$ contains a pairwise sublinear additive spanner of the induced subgraph $G[\ball(c, 4r)]$ with respect to the set $\pset_c$ of pairs, with stretch function $f_{k-1, \epsilon}$. By \Cref{pairwise-sublinear} and  \Cref{demand-size}, the number of edges in such a spanner is at most  $$\tilde{O}\left(2^{2k/\epsilon}|\ball(c, 4r)|^{1+10(k-1)\epsilon}\cdot \left(\frac{n}{L_k}\right)^{\frac{1}{2^k}}\right).$$
	Summing over all large balls, the number of edges in all these spanners is at most
	$$\tilde{O}\left(\sum_{\ball(c, r)\text{ is large}}2^{2k/\epsilon}|\ball(c, 4r)|^{1+10(k-1)\epsilon}\cdot \left(\frac{n}{L_k}\right)^{\frac{1}{2^k}}\right)\leq 2^{2k/\epsilon}n^{1+(10k-5)\epsilon}\cdot \left(\frac{n}{L_k}\right)^{\frac{1}{2^k}}.$$
\end{enumerate}

Setting $L_k = n^{\frac{2^k-1}{2^{k+1}-1}}$, the total number of edges over all the above types is at most  $O(2^{2k/\epsilon}\cdot n^{1+(10k-5)\epsilon + \frac{1}{2^{k+1}-1}})$. Summing over all $D\in \set{1,2,\ldots,2^{\ceil{\log n}}}$, we can conclude that $|E(H)|\le O(n^{1+10k\epsilon + \frac{1}{2^{k+1}-1}})$.

\section{Subset Additive Spanners}
\label{sec: subset}

In this section, we prove the following lemma, which will serve as a building block for \Cref{subquad}. 
For a graph $G = (V, E)$, a subgraph $H\subseteq G$ and a subset $U\subseteq V$, we say that $H$ is a \emph{subset spanner of $G$ on $U$ with additive error $f(n)$}, if for every pair $u,u'\in U$, $\dist_{H}(u,u')\le \dist_{G}(u,u')+f(n)$.

\begin{lemma}\label{subset}
For any $\epsilon> 0$, there is an algorithm that, given an unweighted undirected graph $G$ on $n$ vertices and $m$ edges, and a subset $U\subseteq V(G)$, in time $O\big(m\big(|U| + 2^{O(1/\epsilon)}\big)\big)$, computes a subset spanner $H$ of $G$ on $U$ with additive stretch $O(|U|^{3/2}\cdot n^\epsilon)$, such that $|E(H)|=O\brac{2^{O(1/\epsilon)}\cdot n\log n}$. 
\end{lemma}

\begin{remark}
	If we do not care about the runtime constraint, then the size of $E(H)$ can be made $2^{O(1/\epsilon)}n$ by replacing \Cref{clustering} with the original Lemma 13 from \cite{bodwin2021better}.
\end{remark}


\subsection{Algorithm description}

We now describe the algorithm for \Cref{subset}.
We first slightly perturb the (unit) weight of each edge, such that in the resulting graph, there is a \textbf{unique shortest path} between every pair of vertices. When calculating their distances, we ignore this perturbation. It is easy to verify that any subset spanner of this resulting graph immediately gives a subset spanner of the original graph with the same stretch function. We rename the resulting graph by $G$.

We now describe the algorithm for constructing $H$.
We first apply the algorithm from \Cref{clustering} to $G$ with parameter $R=\ceil{|U|^{3/2}}$ while replacing $\eps$ in \Cref{clustering} with $\frac{10}{\epsilon\log_2 n}$, and obtain a set $\clusters$ of balls in time $O\brac{2^{10/\epsilon}\cdot m\log  n}$. Classify balls in $\bset$ as two categories: (Small) $|\ball(c, r)| \leq |U|^2$, (Large) $|\ball(c, r)| > |U|^2$.

For each ball $\ball(c,r)\in \balls$, we compute a BFS tree $T_c$ that is rooted at $c$ and spans all vertices in $\ball(c, 4r)$, in time $O(\vol(\ball(c, 4r)))$. From \Cref{clustering}, $\sum_{c} |E(T_c)|= \sum_{c} |\ball(c, 4r)| = O\brac{2^{10/\epsilon}n\log n}$.

\paragraph{Handling small balls.} Consider a small ball $\ball(c, r)\in \clusters$. From Lemma \ref{bottleneck}, we can compute in time $O(\vol(\ball(c, 4r)))$ an integer $d\in [r, 2r]$ such that $|\ball^=(c, d)\cup\ball^=(c, d+1)|\leq 2|\ball(c, 4r)| / r \leq 2|\ball(c, 4r)| / R$. Further, denote: $$\clusters_1 = \{\ball(c, d)\mid \ball(c, r)\in\clusters\text{ is small} \}.$$

We then apply the algorithm from \Cref{cor: BFS consistent} to the induced subgraph $G[\ball(c, 4r)]$ with $S=\ball^=(c, d)\cup\ball^=(c, d+1)$, and obtain a collection $\Pi_c$ of shortest paths. We define $L_c=\bigcup_{\pi\in \Pi_c}\pi$.
From \Cref{cor: BFS consistent} and \Cref{clustering},
$$\begin{aligned}
	\sum_{c} |E(L_c)|&= \sum_{c} O\left(|\ball(c, 4r)| + \sqrt{|\ball(c, 4r)|}\cdot \left(\frac{|\ball(c, 4r)|}{R}\right)^2\right)\\
	& = \sum_c O(2^{15/\epsilon}|\ball(c, 4r)|)=O(2^{25/\epsilon}n\log n).
\end{aligned}$$

\paragraph{Handling large balls.} 
We first apply the algorithm from \Cref{cor: BFS consistent} to $G$ with $S=U$, and obtain a collection $\Pi$ of shortest paths connecting pairs of $U$.
%
We will then iteratively construct, for each large ball $\ball(c, r)$, a set $U_c\subseteq U$ of vertices. Initially, all sets $U_c$ are empty.

We then process all paths in $\Pi$ sequentially in an arbitrary order. Consider any path $\pi_{s, t}\in \Pi$ from $s$ to $t$. Then, the algorithm consists of the following steps.

\begin{enumerate}[(1),leftmargin=*]
	\item Recall that a vertex is covered by $\bset_1$ if it is contained in some ball in $\bset_1$. We first compute the set $\Sigma$ of all maximal subpaths of $\pi_{s, t}$ consisting of only vertices covered by $\bset_1$, and let $\Sigma'$ be the collection of subpaths of $\pi_{s, t}$ obtained by removing all edges of paths in $\Sigma$.
	
	\item Go over each vertex $x$ on $\pi_{s, t}$. Find an arbitrary ball $\ball(c_x, r_x)\ni x$. If $\ball(c_x, r_x)$ is large and $s, t\in U_{c_x}$, then we abort and continue on to the next path in $\Pi$.
	
	\item Otherwise, add all paths in $\Sigma^\prime$ to the final spanner $H$, and add $s, t$ to $U_{c_x}$ for each $x$ on $\pi_{s, t}$ where $\ball(c_x, r_x)$ is large.
\end{enumerate}

\subsection{Stretch analysis}

We now show that the graph $H$ obtained by the above algorithm satisfies the properties required in \Cref{subset}.
Recall that the balls in $\bset$ were obtained by applying 
the algorithm from \Cref{clustering} to graph $G$ with parameters $R=\ceil{|U|^{3/2}}$ and $\frac{10}{\epsilon\log_2 n}$.
Note that, from \Cref{clustering}, the radius of each ball in $\bset$ is at most $R\cdot 2^{10/(\frac{10}{\epsilon\log_2 n})}\le Rn^\eps$, so the radius (diameter, resp) of every ball $\ball(c,2r)$ is at most $2Rn^{\eps}$ ($4Rn^{\eps}$, resp).

\begin{definition}
	A vertex $v$ is \textbf{covered at boundary} by $\balls_1$, if $v$ belongs to the boundary of \textbf{every} ball in $\balls_1$ that contains $v$.
\end{definition}

\begin{claim}\label{exact}
Let $s, t$ be a pair of vertices such that all vertices on the shortest path between them are covered by $\clusters_1$. Then $\dist_H(s, t)\leq \dist_{G}(s, t) + 8n^\epsilon\cdot R$.
Furthermore, if both $s$ and $t$ are covered at boundary by $\clusters_1$, then $\dist_H(s, t) = \dist_{G}(s, t)$.
\end{claim}
\begin{proof}
Let $\pi$ be the shortest path $\pi$ from $s$ to $t$ such that all vertices in $\pi$ are covered by $\balls_1$. We first compute two sequences $(u_0, u_1, u_2, \ldots, u_{l-1})$, $(v_0, v_1, v_2, \ldots, v_{l-1})$ of vertices in $\pi$, such that vertices $v_0=u_0=s,v_1,u_1,\ldots,v_{l-1},u_{l-1}$ appear on path $\pi$ from in the direction from $s$ to $t$ in this order, and a sequence $(\ball(c_1, r_{1}), \ball(c_2, r_{2}), \ldots, \ball(c_l, r_{l}))$ of small balls in $\balls$, such that $v_0\in \ball(c_1, d_{1})$, and for all $1\leq i\leq l-1$, $v_i\in \ball^=(c_{i+1}, d_{{i+1}}), u_i\in \ball^=(c_i, d_{i}+1)$; recall that integers $d_i$ were computed by \Cref{bottleneck}. See Figure \ref{terminal} for an illustration.

\begin{figure}[h]
	\begin{center}
		\begin{tikzpicture}[thick,scale=0.7]
	\draw [red] (15, 0.5) ellipse (2.3 and 2);
	\draw [red] (0, 0.5) ellipse (2.3 and 2);
	\draw [red] (3, 0.5) ellipse (2.3 and 2);
	\draw [red] (6, 0.5) ellipse (2.3 and 2);
	\draw [red] (18, 0.5) ellipse (2.3 and 2);
	
	\draw (0, 0) node[circle, draw, fill=black!50, inner sep=0pt, minimum width=6pt, label = $s$] {};
	\draw (18, 0) node[circle, draw, fill=black!50, inner sep=0pt, minimum width=6pt,label = $t$] {};
	
	\draw (2, 0) node[circle, draw, fill=black!50, inner sep=0pt, minimum width=6pt,label = $u_1$] {};
	\draw (1, 0) node[circle, draw, fill=black!50, inner sep=0pt, minimum width=6pt,label = $v_1$] {};
	\draw [line width = 0.5mm] (0.2, 0) -- (0.8, 0);
	\draw [line width = 0.5mm] (1.2, 0) -- (1.8, 0);
	
	\draw (-1, 2.4) node[red, label={$\ball(c_1, d_{1}+1)$}]{};
	
	\draw (5, 0) node[circle, draw, fill=black!50, inner sep=0pt, minimum width=6pt,label = $u_2$] {};
	\draw [line width = 0.5mm] (2.3, 0) -- (3.8, 0);
	\draw [line width = 0.5mm] (4.2, 0) -- (4.8, 0);
	
	\draw (2.7, 2.4) node[red, label={$\ball(c_2, d_{2}+1)$}]{};
	
	\draw (4, 0) node[circle, draw, fill=black!50, inner sep=0pt, minimum width=6pt,label = $v_2$] {};
	\draw (8, 0) node[circle, draw, fill=black!50, inner sep=0pt, minimum width=6pt,label = $u_3$] {};
	\draw [line width = 0.5mm] (5.3, 0) -- (7.8, 0);
	
	\draw (6.5, 2.4) node[red, label={$\ball(c_3, d_{3}+1)$}]{};
	
	\draw[dotted] (3.7, 1.2) -- (2.5, -6);
	\draw[dotted] (8.3, 1.2) -- (9.5, -6);
	\draw (3.5, -7) node[circle, draw, fill=black!50, inner sep=0pt, minimum width=6pt,label = $v_2$] {};
	\draw (9, -7) node[circle, draw, fill=black!50, inner sep=0pt, minimum width=6pt,label = $u_3$] {};
	\draw (6, -6) node[circle, draw, fill=black!50, inner sep=0pt, minimum width=6pt,label = $c_3$] {};
	\draw [red] (6, -6) ellipse (3.4 and 3.4);
	\draw [line width = 0.5mm] (3.7, -7) -- (8.8, -7);
	\draw [dashed] (2, -7) -- (3.3, -7);
	\draw [dashed] (9.2, -7) -- (11.5, -7);
		\draw [decorate,
	decoration = {brace}] (3.5,-6) -- (5.6,-5.2);
	\draw (4.5, -5.6) node[label={[rotate=20]$d_3$}]{};
	\draw [dashed] (5.8, -6.1) -- (3.7, -6.9);
	\draw [decorate,
	decoration = {brace}] (6.4,-5.2) -- (9,-6.1);
	\draw (7.7, -5.7) node[label={[rotate=-16]$d_3+1$}]{};
	\draw [dashed] (6.2, -6.1) -- (8.8, -6.9);
	
	\draw (16, 0) node[circle, draw, fill=black!50, inner sep=0pt, minimum width=6pt,label = -90: $v_{l-1}$] {};
	\draw (17, 0) node[circle, draw, fill=black!50, inner sep=0pt, minimum width=6pt,label = $u_{l-1}$] {};
	\draw [line width = 0.5mm] (17.8, 0) -- (17.2, 0);
	\draw [line width = 0.5mm] (16.8, 0) -- (16.2, 0);

	\draw (18.5, 2.4) node[red, label={$\ball(c_l, d_{l}+1)$}]{};
	
	\draw (13, 0) node[circle, draw, fill=black!50, inner sep=0pt, minimum width=6pt,label = -90 : $v_{l-2}$] {};
	\draw (14, 0) node[circle, draw, fill=black!50, inner sep=0pt, minimum width=6pt,label = $u_{l-2}$] {};
	\draw [line width = 0.5mm] (13.8, 0) -- (13.2, 0);
	\draw [line width = 0.5mm] (15.8, 0) -- (14.2, 0);
	
	\draw (14.5, 2.4) node[red, label={$\ball(c_{l-1}, d_{{l-1}}+1)$}]{};
	
	\draw [dashed] (8.2, 0) -- (12.8, 0);
\end{tikzpicture}
	\end{center}
	\caption{The construction of sequences $u_1, u_2, \ldots, u_l$ and $v_1, v_2, \ldots, v_{l-1}$.}\label{terminal}
\end{figure}
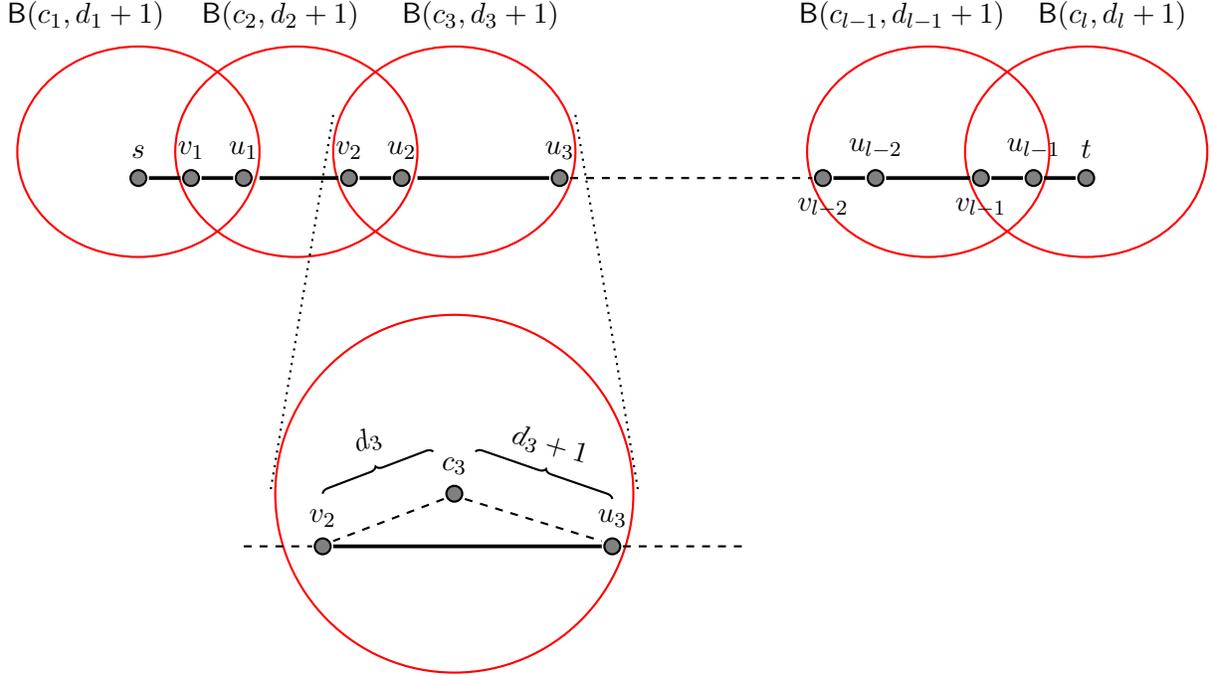
%

This is done via an iterative process described as the following steps. 

\begin{enumerate}[(1),leftmargin=*]
	\item Start with $i=0$ and set $u_0=v_0=s$. If $u_i = t$ then we terminate the process and set $l=i+1$.
	
	Otherwise, we find the ball in $\clusters_1$ that, among all balls in $\clusters_1$ that intersects $\pi[s, u_i]$, the one that contains a vertex on $\pi$ that is \textbf{closest} to $t$. In other words, if we denote $\pi$ as $(s=w_0,w_1,\ldots,w_{k-1},w_k=t)$, then we find the ball that intersects $\pi[s, u_i]$, and, subject to this, contains a vertex $w_j$ with the largest index.
	
	Note that such a ball always exists; for example, we can take an arbitrary ball in $\bset_1$ that contains $u_i$, and we can also notice that $u_{i+1}$ should belong to $\pi(u_i, t]$.
	
	We denote this ball by $\ball(c_{i+1}, d_{i+1})$ and set $u_{i+1}$ as the last vertex of $\pi$ that belongs to  $\ball(c_{i+1}, d_{i+1}+1)$. 
	
	\item If $i\ge 1$, we then let $v_i$ be any vertex that belongs to both $\pi[u_{i-1}, u_i]$ and $\ball^=(c_{i+1}, d_{i+1})$. We will prove shortly that such $v_i$ always exists.
	
	Then, increase $i\leftarrow i+1$ and go to Step (1).
\end{enumerate}

We now turn to argue the existence of $v_i$.

\begin{claim}
For each $i\ge 1$, $\pi[u_{i-1}, u_i]$ intersects with $\ball^=(c_{i+1}, d_{i+1})$, and if $u_{i+1}\neq t$, then $u_{i+1}\in \ball^=(c_{i+1}, d_{i+1}+1)$.
\end{claim}
\begin{proof}[Proof of claim]
First, if $\pi[s, u_i]$ does not intersect $\ball^=(c_{i+1}, d_{i+1})$, as $\pi[s, u_i]$ intersects $\ball(c_{i+1}, d_{i+1})$, it has to lie entirely within $\ball(c_{i+1}, d_{i+1})$, a contradiction to the choice of $\ball(c_1, d_{1})$ and $u_1$. Therefore, $\pi[s, u_i]$ intersect $\ball^=(c_{i+1}, d_{i+1})$. 

We now claim any such intersection $v_i$ must belong to $\pi[u_{i-1}, u_i]$. Otherwise, if $v_i$ belongs to $\pi[u_{j-1}, u_j]$ for some index $1\leq j<i$, then earlier when we were determining $u_{j+1}$, it should have been at least as close to $t$ as $u_{i+1}$, a contradiction to the property that $u_i\in (u_j,t]$.
%

Lastly, if $u_{i+1}\in \ball(c_{i+1}, d_{i+1})$, then the next vertex on path $\pi$ should also belong to $\ball(c_{i+1}, d_{i+1}+1)$, a contradiction to the choice of $u_{i+1}$.
\end{proof}
		
	
It is clear that when the iterative process terminates, $u_l = t$. We now show that for each $2\leq i\leq l-1$, the subpath $\pi[v_{i-1}, u_i]$ is contained in $H$. Note that $v_{i-1}\in \ball^=(c_i, d_{i})$ and $u_i\in \ball^=(c_i, d_{i}+1)$. Since $\ball(c_i, r_{i})$ is small, from the algorithm, we have constructed a subgraph $L_c$ in $G[\ball(c_i, 4r_{i})]$ that preserves all-pairs distances between vertices in $\ball^=(c_i, d_{i})\cup \ball^=(c_i, d_{i}+1)$. As the shortest path between every pair of vertices is unique, the subpath $\pi[v_{i-1}, u_i]$ should be entirely contained in $L_c$.

Lastly, we consider the distance in $H$ between the endpoints $s, t$ of $\pi$. From the construction, we know that vertices $s, u_1\in \ball(c_1, d_{1}+1)$, and $v_{l-1}, t\in \ball(c_l, d_{l}+1)$. Therefore, $\dist_H(s, u_1) \le  4Rn^\epsilon$, and $\dist_H(v_{l-1}, t) \le 4Rn^\epsilon$. It follows that $\dist_H(s, t)\leq \dist_{G}(s, t) + 8Rn^\epsilon$. Furthermore, if both $s, t$ are covered at boundary by $\clusters_1$, then $s\in \ball^=(c_1, d_{1})$ and $t\in \ball^=(c_l, d_{l})$. Therefore, the entire path $\pi$ belongs to $H$, and so $\dist_H(s, t) = \dist_{G}(s, t)$.
\end{proof}

We now focus on the algorithm for handling large balls.  Consider an iteration in which a shortest path $\pi_{s, t}$ is processed for some pair $s, t\in U$. We prove the following claims.

\begin{claim}\label{boundary}
Let $\pi_{s, t}[u, v]$ be a maximal subpath of $\pi_{s, t}$ consisting of only vertices covered by $\bset_1$. If $u\ne s$ and $v\ne t$, then both $u, v$ are covered at boundary by $\bset_1$.
\end{claim}
\begin{proof}
Assume for contradiction that $u$ belongs to a ball $\ball(c, d-1)$ for some small ball $\ball(c, r)\in\clusters$. Let $(w, u)$ be the last edge of subpath $\pi_{s, t}[s, u]$, then $w\in\ball(c, d)$, and so $w$ is also covered. Therefore, $\pi_{s, t}[w, v]$ should be a covered subpath, which contradicts the maximality of subpath $\pi_{s, t}[u, v]$.
\end{proof}

\begin{claim}\label{Uc}
For any large ball $\ball(c, r)$ and every vertex $s\in U_c$, $\dist_H(s, c)\leq \dist_{G}(s, c) + 10Rn^\epsilon$.
\end{claim}
\begin{proof}
Let $\pi_{s, t}$ be the shortest path that was processed in the iteration where $s$ was added to $U_c$. 
By the algorithm, $\ball(c, r)$ must be referring to some ball $\ball(c_x, r_x)$ for some $x\in \pi_{s, t}$. As $H$ contains a BFS tree that is rooted at $c$ and contains $x$, and the radius of each ball is at most $2Rn^{\eps}$, it suffices to show that $\dist_H(s, x)\leq \dist_{G}(s, x) + 8Rn^\epsilon$.

Recall that path $\pi_{s, t}$ was partitioned into subpaths in sets $\Sigma$ and $\Sigma^\prime$. From the algorithm description, all subpaths in $\Sigma'$ are entirely contained in $H$.
For every subpath $\pi_{s, t}[y, z]$ of $\pi_{s, t}$ in $\Sigma$, if $y\neq s$, then from \Cref{exact} and \Cref{boundary}, $\dist_H(y, z) = \dist_{G}(y, z)$; if $y = s$, then from \Cref{exact}, we get that $\dist_H(y, z) \leq \dist_{G}(y, z) + 8Rn^\epsilon$. Since there is at most one subpath $\pi_{s, t}[y, z]$ in $\Sigma$ with $y = s$, the overall additive error caused by subpaths in $\Sigma$ is at most $8Rn^\epsilon$.
\end{proof}

\begin{claim}[additive error]
For every pair  $s, t\in U$, $\dist_H(s, t)\leq \dist_{G}(s, t) + 24Rn^\epsilon$.
\end{claim}
\begin{proof}
Consider the iteration when $\pi_{s, t}$ was processed.
If $\pi_{s, t}$ is entirely covered by $\clusters_1$, 
then the claim follows from \Cref{exact}. 
If not, then from the algorithm, after this iteration some large ball $\ball(c, r)\in\clusters$ intersecting $\pi_{s, t}$ must have added $s, t$ into its set $U_c$. Then from \Cref{Uc}, 
$$\dist_{H}(s, c)\leq \dist_{G}(s, c) + 10Rn^\epsilon,$$
$$\dist_H(c, t)\leq \dist_{G}(c, t) + 10Rn^\epsilon.$$
Let $v$ be an arbitrary vertex in $\ball(c, r)\cap V(\pi_{s, t})$. By triangle inequality,
\[\dist_{G}(s, c) + \dist_{G}(c, t) \leq \dist_G(s, v) + \dist_G(v, t) + 2\cdot \dist(v, c) \leq \dist_G(s, t) + 4Rn^\epsilon.\]
Altogether, $\dist_H(s, t)\leq \dist_{G}(s, t) + 24Rn^\epsilon$.
\end{proof}

\subsection{Size and runtime analysis}

\begin{claim}[spanner size]
$|E(H)|=O\brac{2^{O(1/\epsilon)}\cdot n\log n}$.
\end{claim}
\begin{proof}
It suffices to bound the total number of edges added to $H$ when handling large balls. For each large ball $\ball(c, r)\in \bset$, let us conceptually construct a set $\Pi_c$ of paths within $G[\ball(c, 4r)]$ during handling large balls.

Initially, all sets $\Pi_c$ are empty. When processing a path $\pi_{s, t}\in \Pi$, suppose Step (3) is executed. Consider the set of large balls $\{\ball(c_x, r_x)\mid x\in \pi_{s, t} \}$. For any ball $\ball(c, r)\in \{\ball(c_x, r_x)\mid x\in \pi_{s, t} \}$, let $y, z\in \ball(c, r+1)\cap \pi_{s, t}$ be the first and the last vertex on $\pi_{s, t}$. Then, add the subpath $\pi_{s, t}[y, z]$ to $\Pi_c$. Note that this path $\pi_{s,t}$ lies entirely in $G[\ball(c, 4r)]$.

We first show that in the end, all edges added to $H$ for handling large balls is a subset of $E\left(\bigcup_{\ball(c, r)\in \bset\text{ is large}} \Pi_c\right)$. When processing a path $\pi_{s, t}\in \Pi$, consider any vertex $x$ on a sub-path $\rho\in\Sigma^\prime$ which is not an endpoint of $\rho$. Then, since $x$ is not covered by $\bset_1$, any ball that covers $x$ must be large, and in particular $\ball(c_x, r_x)$ should be large. Suppose $y, z\in \ball(c_x, r_x+1)$ are the first and the last vertex on $\pi_{s, t}$. Then, $\pi_{s, t}[y, z]$ must contain all edges incident on $x$ on $\rho$. Hence, the union of all $\pi_{s, t}[y, z]$ ranging over all $x$'s should contain all paths in $\Sigma^\prime$.

Next, we show that $|\Pi_c|\leq |U|$ for any $c$. In fact, each time we added a new path to $\Pi_c$, we must have added $s, t$ to $U_c$ which increased $|U_c|$. Since $U_c\subseteq U$, we know that $|\Pi_c|\leq |U|$ in the end.

Finally, it suffices to bound the number of edges in $E\left(\bigcup_{\ball(c, r)\in \bset\text{ is large}} \Pi_c\right)$. Using \Cref{consist}, we have: 
$$E(\Pi_c)\leq O\left(|\ball(c, 4r)| + \sqrt{\ball(c, 4r)}\cdot |\Pi_c|\right)\leq O\left(|\ball(c, 4r)|\right).$$
Hence, $$E\left(\bigcup_{\ball(c, r)\in \bset\text{ is large}} \Pi_c\right)\leq O\left(\sum_{\ball(c, r)\in \bset\text{ is large}}|\ball(c, 4r)|\right) = O\brac{2^{O(1/\epsilon)}\cdot n\log n}.$$
\end{proof}

\begin{claim}[runtime]
The runtime of the algorithm is $O\big(m\big(|U| + 2^{O(1/\epsilon)}\big)\big)$.
\end{claim}
\begin{proof}
From \Cref{clustering}, computing all the balls in $\clusters$ takes time at most $O(2^{10/\epsilon}m/\eps)$. From \Cref{cor: BFS consistent}, constructing graphs $L_c$ in all small balls takes time at most $O(|U|^{1/2}\sum_c \vol_{G}(\ball(c, 4r)))=O(|U|^{1/2}m \cdot 2^{10/\eps}/\eps)$, and computing the paths in $\Pi$ takes time $O(m|U|)$. As for the part with large balls, the runtime is dominated by scanning all shortest paths $\pi_{s, t}$ in $\Pi$, which takes time at most $O(n|U|)$. Overall, the runtime is $O\big(m\big(|U| + 2^{O(1/\epsilon)}\big)\big)$.
\end{proof}

\section{Linear-Size Additive Spanners in Sub-quadratic Time}
\label{sec: subq}

In this section we provide the proof of \Cref{subquad}. 
We will first show in \Cref{sec: subquad for 3/7} a simple subquadratic time algorithm for computing an $\tilde O(n^{3/7+\eps})$ additive spanner.
Then we will slightly modify it to achieve better additive error bound in \Cref{sec: proof of subquad}.

\subsection{A subquadratic time algorithm for $\tilde O(n^{3/7+\eps})$ additive spanner}
\label{sec: subquad for 3/7}

In this subsection, we prove the following theorem, which shows that an additive spanner with the current best error (as in \cite{bodwin2021better}) can be computed in subquadratic time.

\begin{theorem}\label{subquad for 3/7}
There is an algorithm, that, given any undirected unweighted graph on $n$ vertices and $m$ edges, and any parameter $\eps > 0$, computes a spanner on $O\brac{2^{O(1/\eps)}\cdot n\log n}$ edges with $\tilde O(n^{3/7+\eps})$ additive stretch, in time $\tilde{O}\big(m + 2^{O(1/\eps)}\cdot n^{13/7}\big)$.
\end{theorem}
\begin{remark}
	If we do not care about the runtime constraint, then the spanner size can be made $2^{O(1/\epsilon)}n$ by replacing \Cref{clustering} with the original Lemma 13 from \cite{bodwin2021better}.
\end{remark}

We start with the following lemma for a preliminary sparsification of the input graph.
\begin{lemma}\label{preproc}
There is an algorithm, that given any graph $G$ and parameter $1<d<n$, computes in $\tilde{O}(m)$ time an $+\tilde O(n/d)$ spanner $G'$ of $G$ with $|E(G')|\le O(nd)$ edges.
\end{lemma}
\begin{proof}
Graph $G^\prime$ is simply the union of (i) for each vertex $v\in V(G)$ with $\deg_G(v) \leq \ceil{d}$, all incident edges of $v$; and (ii) an $O(\log n)$-multiplicative spanner of $G$ with $O(n)$ edges; such a multiplicative spanner can be computed in $\tilde{O}(m)$ time as shown in \cite{baswana2007simple}. Clearly, graph $G'$ can be computed in $\tilde O(m)$ time.
	
Consider now any pair $s, t\in V(G)$ and let $\pi$ be an $s$-$t$ shortest path in $G$. 
It is easy to observe that at most $O(n/d)$ vertices in $\pi$ have degree more than $d$, so the number of edges in $E(\pi)$ that is incident to any degree $\le d$ vertex is at most $O(n/d)$. As we have included in $G'$ an $O(\log n)$-multiplicative spanner of $G$, the distance in $G'$ between the pair of endpoints of every such edge is at most $O(\log n)$.
Therefore, $\dist_{G'}(s,t)\le \dist_G(s,t)+\tilde{O}(n/d)$.
\end{proof}

We now proceed to describe our algorithm for \Cref{lem: reduction}. Similar to \Cref{sec: subset}, we assume that the (unit) edge weights are slightly perturbed so that for every pair $s,t$ there is a unique shortest path connecting them in $G$ (or any subgraph that contains both $s$ and $t$).

As a pre-processing step, if $|E(G)|\ge 10\cdot n^{2-3/7}$, then we first apply the algorithm from \Cref{preproc} to $G$ with $d=n^{1-3/7}$,
and get graph $G'$, so $|E(G')|=O(n^{2-3/7})$.
If $|E(G)|\le 10\cdot n^{2-3/7}$, then we simply set $G'=G$.

We then apply the algorithm from \Cref{clustering} to $G'$ with parameter $R$ (to be determined later) and $\frac{10}{\epsilon\log_2n}$, and let $\bset$ be the set of balls we obtain.
%
We say that a ball $\ball(c, r)\in\bset$ is \emph{small} if $|\ball(c, r)| \leq R^{5/3}$; otherwise we say it is \emph{large}.
For each ball $\ball(c,r)\in \balls$, we compute a BFS tree $T_c$ that is rooted at $c$ and spans all vertices in $\ball(c, 4r)$, in time $O(\vol(\ball(c, 4r)))$. From \Cref{clustering}, $\sum_{c} |E(T_c)|= \sum_{c} |\ball(c, 4r)| = O(2^{10/\epsilon}n\log n)$.


\paragraph{Handling small balls.} Consider a small ball $\ball(c, r)\in \bset$. From Lemma \ref{bottleneck}, we can compute in time $O(\vol(\ball(c, 4r)))$ an integer $d\in [r, 2r]$ such that $|\ball^=(c, d)\cup\ball^=(c, d+1)|\leq 2|\ball(c, 4r)| / r \leq 2|\ball(c, 4r)| / R$. We denote  
$\bset_1 = \{\ball(c, d)\mid \ball(c, r)\in\bset\text{ is small} \}$.
Then, for each small ball $\ball(c, r)$, apply Lemma \ref{subset} to compute a subset spanner $L_c$ of $G'[\ball(c, 4r)]$ on the set $\ball^=(c, d)\cup\ball^=(c, d+1)$.



\paragraph{Handling large balls.}
Choose a random subset $S\subseteq V(G)$ of size $\ceil{10R^{2/3}\log n}$.
We then apply Lemma \ref{subset} to compute a subset spanner $\hat H$ of $G'$ on $S$.

\paragraph{The spanner construction.}
The output graph $H$ is simply defined to be the union of 
\begin{itemize}
	\item for each ball $\ball(c,r)\in \balls$, the tree $T_c$;
	\item for each small ball $\ball(c,r)\in \balls$, graph $L_c$;
	\item graph $\hat H$.
\end{itemize}

\subsubsection{Stretch analysis}

According to \Cref{preproc}, in order to show that $H$ is a $+\tilde O(n^{3/7+\eps})$ spanner of $G$, it suffices to show that $H$ is a $+\tilde O(n^{3/7+\eps})$ spanner of $G'$. For notational convenience, in this subsection we rename $G'$ by $G$.
The following statement, which is a generalization of \Cref{exact}, is the key to the stretch analysis.
\begin{claim}\label{exact2}
Let $s,t$ be a pair of vertices in $G$ and let $\pi$ be a shortest path connecting them. Then there exists (i) a path $\phi$ in $H$ connecting $s$ to $t$; (ii) a sequence of balls $\ball(c_1, r_{1}), \ldots, \ball(c_l, r_{l})$; (iii) two sequences of vertices $(u_0 = s, u_1, \ldots, u_l = t)$ and $(v_0 = s, v_1, \ldots, v_{l-1})$, and two sequences of paths  $\alpha_1, \alpha_2, \ldots, \alpha_l$, $\beta_1, \beta_2, \ldots, \beta_l$, with the following properties.
	\begin{enumerate}[(a),leftmargin=*]
		\item Vertices $u_1, u_2, \ldots, u_l$ appear on path $\pi$ in this order in the direction from $s$ to $t$.
		\item $s\in \ball(c_1, d_{1}), t\in\ball(c_l, d_{l}+1)$, and $v_i\in \ball^=(c_{i+1}, d_{{i+1}}), u_i\in \ball^=(c_i, d_{i}+1), \forall 1\leq i\leq l-1$.
		\item For each $1\leq i\leq l$, $\alpha_i$ is a shortest path in $H$ connecting $v_{i-1}$ to $u_i$ that is entirely contained in $G[\ball(c_i, 4r_{i})]$, and moreover, vertex $v_i$ lies on $\alpha_i$.
		\item For each $1\leq i\leq l-1$, $\beta_i = \alpha_i[v_{i-1}, v_i]$; and $\beta_l = \alpha_l$.
		\item $\phi=\beta_1\circ \cdots \circ \beta_l$; and $|\phi|\leq \dist_{G}(s, t) + 15\cdot Rn^\epsilon + \tilde{O}\left(2^{15/\epsilon}\cdot n^\epsilon\cdot\sum_{i=2}^{l-1}\frac{|\ball(c_i, r_{i})|}{R^{2/3}}\right)$.
	\end{enumerate}
\end{claim}

\begin{figure}[h]
	\begin{center}
		\begin{tikzpicture}[thick,scale=0.7]
	\draw [red] (0, 0.5) ellipse (2.3 and 2);
	\draw [red] (3, 0.5) ellipse (2.3 and 2);
	\draw [red] (6, 0.5) ellipse (2.3 and 2);
	\draw [red] (18, 0.5) ellipse (2.3 and 2);
	
	\draw (0, 0) node[circle, draw, fill=black!50, inner sep=0pt, minimum width=6pt, label = 180 : $s$] {};
	\draw (18, 0) node[circle, draw, fill=black!50, inner sep=0pt, minimum width=6pt,label = 0 : $t$] {};
	
	\draw (2, 0) node[circle, draw, fill=black!50, inner sep=0pt, minimum width=6pt,label = $u_1$] {};
	\draw (1.2, 1.3) node[circle, draw, fill=black!50, inner sep=0pt, minimum width=6pt,label = $v_1$] {};
	\draw [line width = 0.5mm] (0.2, 0) -- (1.8, 0);
	\draw [line width = 0.5mm, color=orange] (0.1, 0.2) to[out=60, in=210] (1, 1.3);
	\draw [line width = 0.5mm, color=orange] (1.2, 1.1) to[out=-90, in=150] (1.8, 0.1);
	
	\draw (-1, 2.4) node[red, label={$\ball(c_1, d_{1}+1)$}]{};
	
	\draw (5, 0) node[circle, draw, fill=black!50, inner sep=0pt, minimum width=6pt,label = $u_2$] {};
	\draw [line width = 0.5mm] (2.3, 0) -- (4.8, 0);
	
	\draw (2.7, 2.4) node[red, label={$\ball(c_2, d_{2}+1)$}]{};
	\draw [line width = 0.5mm, color=orange] (1.4, 1.3) -- (4, 1.3);
	\draw [line width = 0.5mm, color=orange] (4.2, 1.1) to[out=-90, in=150] (4.8, 0.1);
	
	\draw (4.2, 1.3) node[circle, draw, fill=black!50, inner sep=0pt, minimum width=6pt,label = $v_2$] {};
	\draw (8, 0) node[circle, draw, fill=black!50, inner sep=0pt, minimum width=6pt,label = $u_3$] {};
	\draw [line width = 0.5mm] (5.3, 0) -- (7.8, 0);
	
	\draw (6.5, 2.4) node[red, label={$\ball(c_3, d_{3}+1)$}]{};
	\draw [line width = 0.5mm, color=orange] (4.4, 1.3) to[out=0, in=120] (7.9, 0.2);
	
	\draw[dotted] (3.7, 1.2) -- (0, -4);
	\draw[dotted] (8.3, 1.2) -- (9.5, -6);
	\draw (4.5, -4) node[circle, draw, fill=black!50, inner sep=0pt, minimum width=6pt,label = $v_2$] {};
	\draw (9, -7) node[circle, draw, fill=black!50, inner sep=0pt, minimum width=6pt,label = $u_3$] {};
	\draw (6, -6) node[circle, draw, fill=black!50, inner sep=0pt, minimum width=6pt,label = $c_3$] {};
	\draw [red] (6, -6) ellipse (3.4 and 3.4);
	\draw [dashed] (9.2, -7) -- (11, -7);
	\draw [line width = 0.5mm, color=orange] (4.7, -4) to[out=0, in=120] (8.9, -6.8);
	\draw (7.5, -5) node[red, label={$\alpha_3$}]{};
	\draw [line width = 0.5mm, color=orange] (1.4, -4) to (4.3, -4);
	\draw (1.2, -4) node[circle, draw, fill=black!50, inner sep=0pt, minimum width=6pt,label = $v_1$] {};
	\draw (2.5, -4.2) node[red, label={$\beta_2$}]{};
	
	\draw [decorate, decoration = {brace}] (5.7,-6) -- (4.5,-4.3);
	\draw [dashed] (4.7, -4.2) -- (5.8, -5.8);
	\draw (4.3, -6) node[label={[rotate=-40]$d_3$}]{};
	\draw [decorate, decoration = {brace}] (8.7, -7.3) -- (6, -6.3);
	\draw [dashed] (6.2, -6.1) -- (8.8, -7);
	\draw (7, -8) node[label={[rotate=-20]$d_3+1$}]{};
	
	\draw (16.2, 1.3) node[circle, draw, fill=black!50, inner sep=0pt, minimum width=6pt,label = $v_{l-1}$] {};
	\draw (17, 0) node[circle, draw, fill=black!50, inner sep=0pt, minimum width=6pt,label = $u_{l-1}$] {};
	\draw [line width = 0.5mm] (17.8, 0) -- (17.2, 0);
	
	\draw [line width = 0.5mm, color=orange] (16.4, 1.3) to[out=0, in=110] (17.9, 0.2);
	
	\draw (18.5, 2.4) node[red, label={$\ball(c_l, d_{l}+1)$}]{};
	
	\draw (13.2, 1.3) node[circle, draw, fill=black!50, inner sep=0pt, minimum width=6pt,label = $v_{l-2}$] {};
	\draw (14, 0) node[circle, draw, fill=black!50, inner sep=0pt, minimum width=6pt,label = $u_{l-2}$] {};
	\draw [line width = 0.5mm] (16.8, 0) -- (14.2, 0);
	\draw [red] (15, 0.5) ellipse (2.3 and 2);
	\draw (14.5, 2.4) node[red, label={$\ball(c_{l-1}, d_{l-1}+1)$}]{};
	\draw [line width = 0.5mm, color=orange] (13.2, 1.1) to[out=-90, in=150] (13.8, 0.1);
	\draw [line width = 0.5mm, color=orange] (13.4, 1.3) -- (16, 1.3);
	\draw [line width = 0.5mm, color=orange] (16.2, 1.1) to[out=-90, in=150] (16.8, 0.1);
	
	\draw [dashed, color=orange] (6.8, 1.3) -- (13.1, 1.3);
	\draw [dashed] (8.2, 0) -- (13.8, 0);
\end{tikzpicture}
		\caption{The construction of vertex sequences $u_1, u_2, \ldots, u_l$ and $v_1, v_2, \ldots, v_{l-1}$; the orange paths belong to the spanner $H$.}\label{add-err}
	\end{center}
\end{figure}
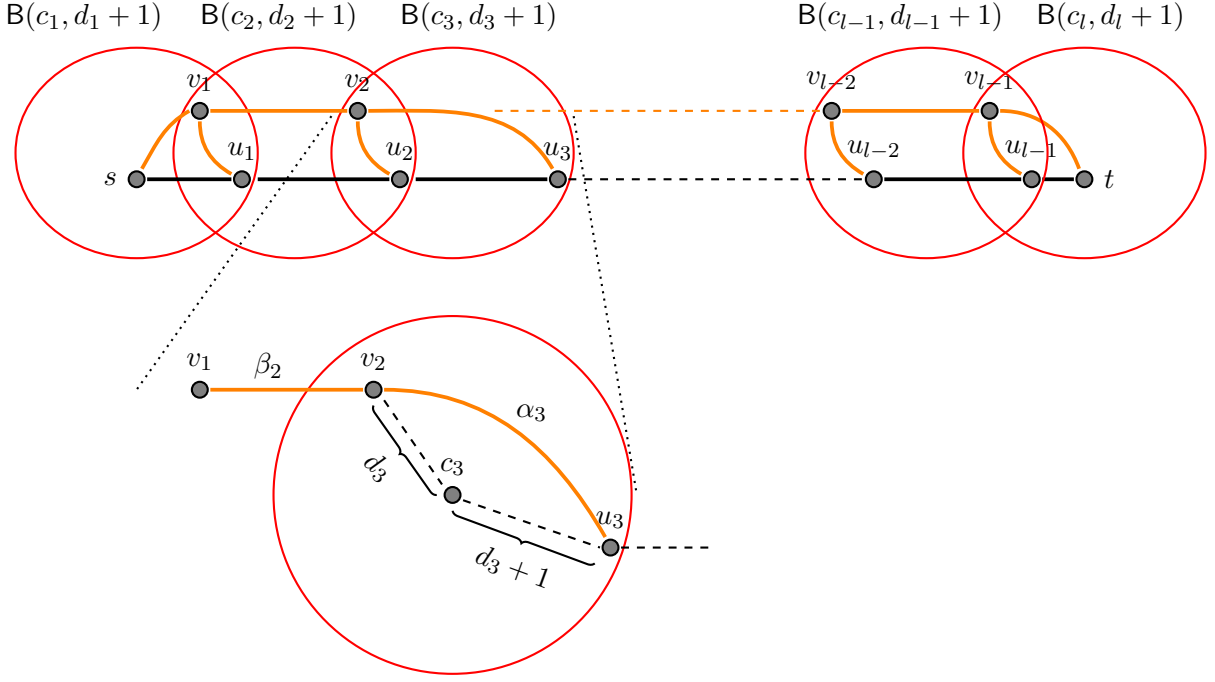

\begin{proof}
Start with $i=0$. Before iteration $i$, suppose we have already computed vertices $u_1, u_2, \ldots, u_i$, and $v_1, v_2, \ldots, v_{i-1}$, and paths $\alpha_1, \alpha_2, \ldots, \alpha_i$, and paths $\beta_1, \beta_2, \ldots, \beta_{i-1}$. During the algorithm, keep $\phi = \beta_1\circ \beta_2\cdots\circ \beta_{i-1}\circ\alpha_i$. So $\phi$ is a path in $H$ from $s$ to $u_i$.

\begin{enumerate}[(1),leftmargin=*]
	\item If $u_i = t$ then we terminate the process and let $l=i$.
	
	Otherwise, find the ball  in $\bset_1$ that, among all balls in $\bset_1$ that intersects with $\phi$, the one that contains a vertex on $\pi$ that is closest to $t$. Note that this ball always exists; for example we can pick an arbitrary one that contains $u_i$.
	
	We denote this ball by $\ball(c_{i+1}, d_{i+1})$ and set $u_{i+1}$ as the last vertex of $\pi$ that belongs to  $\ball(c_{i+1}, d_{i+1}+1)$.
	
	\item Next, if $i = 0$, then let $\alpha_1$ be the shortest path in $H$ from $s=v_0$ to $u_1$.
	
	If $i\ge 1$, we then let $v_i$ be any vertex that belongs to both $\alpha_{i}$ and $\ball^=(c_{i+1}, d_{i+1})$, and let $\alpha_{i+1}$ be the shortest path from $v_i$ to $u_{i+1}$ in $H$.  Then, define $\alpha_{i+1}$ to be the shortest path from $v_i$ to $u_{i+1}$, and $\beta_i = \alpha_i[v_{i-1}, v_i]$. We will prove shortly the existence of $v_i$. 
	
	Finally, increase $i\leftarrow i+1$ and go to Step (1).
\end{enumerate}


We now argue the existence of $v_i$.
\begin{claim}
	\label{clm: existence}
For each $i\ge 1$, $\phi$ intersects with $\ball^=(c_{i+1}, d_{i+1})$, and the intersection point $v_i$ must belong to $\alpha_i$. Also, if $u_{i+1}\neq t$, then $u_{i+1}\in \ball^=(c_{i+1}, d_{i+1}+1)$.
\end{claim} 
\begin{proof}[Proof of \Cref{clm: existence}]
First, if $\phi$ does not intersect $\ball^=(c_{i+1}, d_{i+1})$, then as $\phi$ already intersects with $\ball(c_{i+1}, d_{i+1})$, the path should lie entirely within $\ball(c_{i+1}, d_{i+1})$, a contradiction with the choice of $\ball(c_1, d_{1})$ and $u_1$. Therefore, there exists a vertex $v_i$ which is an intersection of $\phi$ and $\ball^=(c_{i+1}, d_{i+1})$. 

Second, if $v_i$ does not belong to $\alpha_i$ but instead belongs to sub-path $\beta_j$ for some index $j<i$, then in an earlier iteration when we were determining $u_{j+1}$, $u_{j+1}$ should be at least as close to $t$ as $u_{i+1}$, which contradicts the fact that $u_{i+1}\in \pi(u_{j+1}, t]$. 

Third, if $u_{i+1}\neq t$, while $u_{i+1}\in \ball(c_{i+1}, d_{i+1})$, then the next vertex of $\pi$ should also belong to $\ball(c_{i+1}, d_{i+1}+1)$, a contradiction to the fact that $u_{i+1}$ is the closest vertex to $t$.
\end{proof}

It is easy to verify that properties $(a)$-$(d)$ hold from the above iterative process. It remains to prove property $(e)$, which we do in the next claim.
	\begin{claim}
	\label{clm: beta_i}
For each $2\leq i\leq l-1$, $|\beta_i|\leq 5Rn^\epsilon$, and
		$$|\beta_i|\leq |\pi[u_{i-1}, u_i]| + |\alpha_{i-1}[v_{i-1}, u_{i-1}]| - |\alpha_{i}[v_{i}, u_{i}]| +\tilde{O}\left(2^{15/\epsilon}\cdot n^\epsilon\cdot\frac{|\ball(c_i, r_{i})|}{R^{2/3}}\right).$$
	\end{claim}
	\begin{proof}[Proof of \Cref{clm: beta_i}]
		As the radius of $G[\ball(c_i, r_{i})]$ is at most $Rn^\epsilon$, $\dist_H(v_{i-1}, u_i) \leq 5Rn^\epsilon$.
		
	From the iterative process, $v_{i-1}, u_i\in \ball^=(c_i, d_{i})\cup\ball^=(c_i, d_{i}+1)$. If the ball $\ball(c_i, r_{i})$ is small, then by the construction of subset spanners, $$\begin{aligned}
			\dist_H(v_{i-1}, u_i) &\leq \dist_{G}(v_{i-1}, u_i) + \tilde{O}\left(\left(\frac{|\ball(c_i, 4r_{i})|}{R}\right)^{3/2}\cdot n^\epsilon\right)\\
			&\leq \dist_{G}(v_{i-1}, u_i) + \tilde{O}\left( \frac{|\ball(c_i, r_{i})|}{R^{2/3}}\cdot 2^{15/\epsilon}\cdot n^\epsilon\right).
		\end{aligned}$$
		If $\ball(c_i, r_{i})$ is large, then by definition, $|\ball(c_i, r_{i})|\geq R^{5/3}$, and therefore,
		$$\begin{aligned}
			\dist_H(v_{i-1}, u_i) &\leq 5Rn^\epsilon\leq \tilde{O}\left(\frac{|\ball(c_i, r_{i})|}{R^{2/3}}\cdot n^\epsilon\right).
		\end{aligned}$$
	
		Finally, as $\dist_H(v_{i-1}, u_i) = |\beta_i| + |\alpha_i[v_i, u_i]|$, $\dist_{G}(v_{i-1}, u_i)\leq |\pi[u_{i-1}, u_i]| + |\alpha_{i-1}[v_{i-1}, u_{i-1}]|$. The claim now follows by rearranging the terms to yield the inequality.
	\end{proof}

	Summing all $2\leq i\leq l-1$ for the above claim, and note that $|\beta_0|, |\beta_l| \leq 5Rn^\epsilon, |\alpha_1|\leq 5\cdot 5Rn^\epsilon$. Property $(e)$ now follows.
\end{proof}

Eventually, we are ready to analyze the additive error of any pairs of vertices in $V$.
\begin{claim}[additive error]
\label{clm: additive error}
For every $s, t\in V$, $\dist_H(s, t)\leq \dist_{G}(s, t) + \tilde{O}(Rn^\epsilon) + \tilde{O}(n^{1+\epsilon} / R^{4/3})$.
\end{claim}
\begin{proof}
	Let $\pi$ be a shortest path between $s, t$ in $G$. We apply the algorithm from \Cref{exact2} to $\pi$, and obtain a path $\phi$ as well as all the other auxiliary sequences, such that:
	$$|\phi|\leq \dist_{G}(s, t) + 15\cdot Rn^\epsilon + \tilde{O}\left(2^{15/\epsilon}\cdot n^\epsilon\cdot\sum_{i=2}^{l-1}\frac{|\ball(c_i, r_{i})|}{R^{2/3}}\right).$$
	If $\sum_{i=1}^{l}|\ball(c_i, r_{i})|\leq n/R^{2/3}$, then  we are done. Otherwise, let $a$ be the smallest index such that $\sum_{i=1}^a |\ball(c_i, r_{i})| > n/R^{2/3}$, and let $b$ be the largest index such that $\sum_{i=b}^l |\ball(c_i, r_{i})| > n/R^{2/3}$. Then, by construction of $S$, with high probability, there exist indices $1\leq x\leq a, b\leq y\leq l$ such that $\ball(c_x, r_{x})\cap S\neq \emptyset, \ball(c_y, r_{y})\cap S\neq \emptyset$; we can assume $x\leq y$ by selecting the smallest choice of $x$ and the largest choice of $y$. Take two vertices $s_1\in \ball(c_x, r_{x})\cap S$ and $s_2\in \ball(c_y, r_{y})\cap S$. Since $H$ contains a subset spanner $\hat H$ on $S$,
	$$\begin{aligned}
		\dist_H(s_1, s_2)&\leq \dist_G(s_1, s_2) + \tilde{O}(|S|^{3/2}n^\epsilon)\leq \dist_G(s_1, s_2) + \tilde{O}(Rn^\epsilon)\\
		&\leq \dist_G(u_x, u_y) + \dist_G(u_x, s_1) + \dist_G(u_y, s_2) + \tilde{O}(Rn^\epsilon)\\
		&\leq \dist_G(u_x, u_y) + \tilde{O}(Rn^\epsilon).
	\end{aligned}$$
	
	Using similar arguments in the proof of \Cref{exact2}, we can show that:
	$$\dist_H(s, u_x)\leq \dist_G(s, u_x) + 15\cdot Rn^\epsilon + \left(2^{15/\epsilon}\cdot n^\epsilon\cdot\sum_{i=2}^{x-1}\frac{|\ball(c_i, r_{i})|}{R^{2/3}}\right),$$
	$$\dist_H(u_y, t)\leq \dist_G(u_y, t) + 15\cdot Rn^\epsilon + \left(2^{15/\epsilon}\cdot n^\epsilon\cdot\sum_{i=y+1}^{l-1}\frac{|\ball(c_i, r_{i})|}{R^{2/3}}\right).$$
	Summing over all three inequalities completes the proof.
\end{proof}

Setting $R=\ceil{n^{3/7}}$, the above claim implies that the additive error of $H$ is $\tilde O(n^{3/7+\eps})$.

\subsubsection{Size and runtime analysis}

\begin{claim}[spanner size]
\label{clm: spanner size}
$|E(H)|=2^{O(1/\epsilon)}\cdot n\log n$.
\end{claim}
\begin{proof}
From \Cref{subset}, each subset spanner $L_c$ within $G[\ball(c, 4r)]$ contains $2^{O(1/\epsilon)}|\ball(c, 4r)|$ edges, and so $\sum_{c}|E(L_c)|=2^{O(1/\epsilon)}\cdot 2^{O(1/\epsilon)}n\log n=2^{O(1/\epsilon)}n\log n$. Similarly, the global subset spanner $\hat H$ satisfies that $|E(\hat H)|=2^{O(1/\epsilon)}n\log n$. 
Additionally, $\sum_{c}|E(T_c)|=\sum_{c}|\ball(c, 4r)|=2^{O(1/\epsilon)}n\log n$ edges. Altogether, we get that $|E(H)|=2^{O(1/\epsilon)}\cdot n\log n$.
\end{proof}

\begin{claim}[runtime]
\label{clm: runtime}
The runtime of the algorithm is $\tilde O(m+|E(G')|\cdot 2^{O(1/\epsilon)}\cdot R^{2/3})$.
\end{claim}
\begin{proof}
The runtime for computing $G'$ is $\tilde O(m)$.
From \Cref{clustering}, the set $\balls$ of balls can be computed in time $2^{O(1/\epsilon)}\cdot |E(G')|$.
The runtime for computing BFS trees within balls in $\balls$ is $2^{O(1/\epsilon)}\cdot |E(G')|$.
From \Cref{subset}, the runtime for computing subset spanners within small balls is at most
\[
\begin{aligned}
\sum_{c}O\bigg(\vol_{G'}(\ball(c,4r))\cdot \bigg(2^{O(1/\eps)}+\frac{2|\ball(c,4r)|}{R}\bigg)\bigg)
& \le \sum_{c}O\bigg(\vol_{G'}(\ball(c,4r))\cdot 2^{O(1/\eps)}\cdot R^{2/3}\bigg)\\
& \le m \cdot 2^{O(1/\eps)}\cdot R^{2/3}.
\end{aligned}
\]
The claim now follows.
\end{proof}

Since we set $R=\ceil{n^{3/7}}$, the above claim implies that the runtime of the algorithm for computing $H$ is $\tilde O(m+2^{O(1/\eps)}\cdot n^{13/7})$, as $|E(G')| = O(n^{11/7})$.

\subsection{Completing the proof of \Cref{subquad}}
\label{sec: proof of subquad}

In this subsection, we complete the proof of \Cref{subquad} by slightly modifying the algorithm in \Cref{sec: subquad for 3/7} and applying the lemma and the algorithm recursively. Specifically, we first prove the following lemma.

\begin{lemma}
\label{lem: reduction}
Let $f(\rho)=\frac{2/3-\rho}{4-(19/6)\rho}$ be a function.
If there is an algorithm $\alg$, that given any graph $G$ on $n$ vertices and $m$ edges, in time $\tilde O(m+n^{\gamma})$ computes an $+O(n^{\rho})$ spanner of $G$ with at most $Cn$ edges, such that $\gamma\ge 1+\frac{(3/2)f(\rho)(1-\rho)}{3/2-\rho}$, then for any parameter $\epsilon>0$, there is an algorithm $\alg'$, that given any graph $G'$ on $n$ vertices and $m$ edges, in time $\tilde O(m+2^{O(\gamma/\eps)}n^{\gamma})$ computes an $+O(n^{\epsilon+f(\rho)})$ spanner of $G'$ with $\brac{2^{O(1/\eps)}\cdot Cn\cdot\log n}$ edges.
\end{lemma}
\begin{remark}
	If we do not care about the runtime constraint, then the spanner size of $G'$ can be made $O\brac{2^{O(1/\epsilon)}Cn}$ (that is, without the logarithmic factor) by replacing \Cref{clustering} with the original Lemma 13 from \cite{bodwin2021better}.
\end{remark}

We now use \Cref{lem: reduction} to prove
\Cref{subquad}. 
We set $\eps>0$ as a small enough constant.
Note that \Cref{subquad for 3/7} in fact gives an algorithm with parameter $(\gamma=13/7,\rho=3/7+0.1, C=O(1))$ and we denote it by $\alg_0$. We then apply \Cref{lem: reduction} with $\alg=\alg_0$, and denote by $\alg_1$ the algorithm that it produces, so $\alg_1$ has produce an $+O(n^{\epsilon+f(\rho)})$ spanner.
Note that the invariant point $\rho^*$ of the mapping $f$ (i.e., the value of $0<\rho^*<1$ such that $f(\rho^*)=\rho^*$) is $\rho^*=\frac{15-\sqrt{54}}{19}=0.4027...$.
We then iteratively apply \Cref{lem: reduction} for $K$ times (where $K$ is a large enough constant such that $f(f(\cdots f(3/7+0.1+\eps)\cdots)+\eps)+\eps<0.403$), and get algorithms $\alg_2,\ldots,\alg_K$. It is not hard to verify that the property $\gamma\ge 1+\frac{(3/2)f(\rho)(1-\rho)}{3/2-\rho}$ always holds.
Eventually, $\alg_K$ is the algorithm that we return. Note that the additive error of the spanner it produces is $+O(n^{0.403})$, the running time is $\tilde O(m+n^{13/7}\cdot 2^{O((13/7)\cdot (K/\eps))})$, which is $\tilde O(m+n^{13/7})$ as $1/\eps$ and $K$ are both constants, and the size of the spanner it produces is $2^{O((13/7)\cdot (K/\eps))}\cdot Cn\log^{O(K)}n=\tilde{O}(n)$.

We now sketch the proof of \Cref{lem: reduction}, highlighting the difference between the algorithm here and the algorithm in \Cref{sec: subquad for 3/7}.

\begin{proof}[Proof Sketch of \Cref{lem: reduction}]
The algorithm for \Cref{lem: reduction} is very similar to the algorithm for \Cref{subquad for 3/7}, except for (i) an extra Step 4 below, which is a recursive call of an additive spanner algorithm; and (ii) more fine-grained tuning of parameters. We define the function $g(\rho)=\frac{(3/2)\cdot f(\rho)}{(3/2-\rho)}$.

\paragraph{Step 1.} Sparsify $G$ to get $G'\subseteq G$ using \Cref{preproc} with $d=O(n^{1-f(\rho)})$ and $|E(G')|=O(n^{2-f(\rho)})$. 

\paragraph{Step 2.} Compute a set $\balls$ of balls using \Cref{clustering} with parameters $R=\ceil{n^{f(\rho)}}$ and $\frac{\eps}{10\log n}$. We say that a ball is \emph{small} if $|\ball(c,r)|\le n^{g(\rho)}$, otherwise we say it is \emph{large}.

\paragraph{Step 3.} For each small ball $\ball(c,r)$, we apply \Cref{bottleneck} to compute an integer $d\in [r, 2r]$ such that $|\ball^=(c, d)\cup\ball^=(c, d+1)|\leq 2|\ball(c, 4r)| / r \leq 2|\ball(c, 4r)| / R$, and then apply \Cref{subset} to compute a subset spanner $L_c$ of $G'[\ball(c, 4r)]$ on the set $\ball^=(c, d)\cup\ball^=(c, d+1)$.

\paragraph{Step 4.} For each large ball $\ball(c,r)$, we apply the algorithm \alg to compute a spanner of $G'[\ball(c, 4r)]$ with error $+O(|\ball(c,r)|^{\rho})$, that we denote by $L_c$.

\paragraph{Step 5.} Sample a random subset $S\subseteq V(G)$ of $\ceil{10R^{2/3}\log n}$ vertices, and apply \Cref{subset} to compute a subset spanner $\hat H$ of $G'$ on $S$ with additive error $O(|S|^{3/2}\cdot n^{\eps})$.

$\ $

The output graph $H$ is simply defined to be the union of 
\begin{itemize}
	\item for each ball $\ball(c,r)\in \balls$, a BFS tree $T_c$ that is rooted at $c$ and spans all vertices in $\ball(c, 4r)$;
	\item for each small or large ball $\ball(c,r)\in \balls$, graph $L_c$;
	\item graph $\hat H$.
\end{itemize}

The analysis is almost identical to the analysis in \Cref{sec: subquad for 3/7}, with the following changes.

\paragraph{Stretch analysis.} In \Cref{clm: beta_i}, the analysis would be changed to
$$\begin{aligned}
\dist_H(v_{i-1}, u_i) &\leq \dist_{G}(v_{i-1}, u_i) + \tilde{O}\left(\max\set{\left(\frac{|\ball(c_i, 4r_{i})|}{n^{f(\rho)}}\right)^{3/2}\cdot n^\epsilon, |\ball(c_i, 4r_{i})|^{\rho}}\right)\\
&\leq \dist_{G}(v_{i-1}, u_i) + \tilde{O}\left( \frac{|\ball(c_i, r_{i})|}{n^{(1-\rho)\cdot g(\rho)}}\cdot 2^{O(1/\epsilon)}\cdot n^\epsilon\right).
\end{aligned}$$
Consequently, the analysis in \Cref{clm: additive error} would be changed to
$$
\begin{aligned}
\dist_H(s, u_x) & \leq \dist_G(s, u_x) + O(n^{\epsilon+f(\rho)}) + O\left(2^{O(1/\epsilon)}\cdot n^\epsilon\cdot\sum_{i=2}^{x-1}\frac{|\ball(c_i, r_{i})|}{n^{(1-\rho)\cdot g(\rho)}}\right)\\
& \le \dist_G(s, u_x) + O(n^{\epsilon+f(\rho)}) + O\left(2^{O(1/\epsilon)}\cdot n^\epsilon\cdot\frac{n^{1-\frac{2}{3}f(\rho)}}{n^{(1-\rho)\cdot g(\rho)}}\right)\\
& \le \dist_G(s, u_x) + 2^{O(1/\epsilon)}\cdot n^{\epsilon+f(\rho)}.
\end{aligned}
$$
and similarly we get that $\dist_H(u_y, t)\leq \dist_G(u_y, t) + 2^{O(1/\epsilon)}\cdot n^{\epsilon+f(\rho)}$. Therefore, the additive error of the algorithm is $2^{O(1/\epsilon)}\cdot n^{\epsilon+f(\rho)}$.

\paragraph{Size and runtime analysis.} Via similar arguments as in the proof of \Cref{clm: spanner size}, we can show that $|E(H)|=O\brac{2^{O(1/\eps)}\cdot Cn\log n}$.
We now analyze the runtime of the algorithm. Ignoring the new Step 4, we can show via similar arguments as in the proof of \Cref{clm: runtime} that the runtime is $\tilde O(m+2^{O(1/\eps)}\cdot n^{1+\frac{(3/2)f(\rho)(1-\rho)}{3/2-\rho}})$. The runtime of Step 4 is
$$\sum_{c: \text{ }\ball(c,r) \text{ large}}\bigg(\tilde O(\vol_{G}(\ball(c,4r)))+|\ball(c,4r)|^{\gamma}\bigg)=\tilde O(m+2^{O(\gamma/\eps)}\cdot n^{\gamma}).$$
As $\gamma\ge 1+\frac{(3/2)f(\rho)(1-\rho)}{3/2-\rho}$, the runtime of the whole algorithm is 
$\tilde O(m+2^{O(\gamma/\eps)}\cdot n^{\gamma})$.
\end{proof}

\section*{Acknowledgment}
We would like to thank anonymous reviewers for the detailed reading and the helpful comments for improving the presentation of this paper, and in particular, for explaining in detail the comparison between the previous work \cite{thorup2006spanners,huang2019thorup,abboud2018hierarchy} and the connection between our results and previous results on $(1+\epsilon,\beta)$-spanners.

\vspace{5mm}
\bibliographystyle{alpha}
\bibliography{ref}

\vspace{5mm}
\appendix
\section{Proof of \Cref{clustering}}
\label{apd: Proof of clustering}

Throughout, we use the parameter $\beta = n^\eps$. 

The algorithm iteratively builds the collection $\balls$. Initially, $\balls = \emptyset$. The algorithm continues to be executed as long as there exists a vertex that is not covered by the collection $\balls$.
We now describe an iteration. First, we pick an arbitrary vertex $c$ that is not covered by the current collection $\balls$,
and add a new ball $\ball(c,r)$ to $\bset$ centered at $c$. Its radius $r$ is determined by the following process: 
\begin{enumerate}[(1)]
	\item Start with $r=R$.
	\item Perform breadth-first search from $c$ in $G$ to compute the ball $\ball(c, 4r)$.
	\item If $|\ball(c, 4r)|\leq \beta\cdot |\ball(c, r/2)|$ and $\vol(\ball(c, 4r))\leq \beta\cdot \vol(\ball(c, r/2))$, then add the ball $\ball(c, r)$ to $\balls$, and terminate the iteration. Otherwise, update $r\leftarrow 4r$ and repeat Steps (2) and (3).
\end{enumerate}

We now proceed to analyze the algorithm. 
First, it is easy to see that the radius of every ball in $\bset$ at the end of the algorithm is at least $R$, as the process of determining the radius of each new ball start with $r=R$ and only increases $r$ afterwards.
Second, from the algorithm, when a new ball is added to $\bset$, its radius $r$ is determined by an iterative process, where in each round, $r$ is increased by a factor of $4$ whenever $|\ball(c, 4r)|> \beta\cdot |\ball(c, r/2)|$ or $\vol(\ball(c, 4r)) > \beta\cdot |\ball(c, r/2)|$. As $|\ball(c, 4r)|$ and $\vol(\ball(c, 4r))$ are bounded by $n$ and $m$ respectively, the number of times that the radius $r$ is increased is at most $\ceil{\log_\beta n} + \ceil{\log_\beta m} < 5/\eps$. Therefore, in the end, $r\le  R\cdot  4^{5/\eps}= R\cdot  2^{10/\eps}$.

We next prove the following observation.

\begin{observation}\label{grow}
At the end of the algorithm, for every vertex $v\in V$, there are at most $5/\eps$ balls $\ball(c, r)$ in $\bset$, such that $\dist_{G}(c,v)\le r/2$.
\end{observation}
\begin{proof}
We first prove the following observation.
\begin{observation}\label{cover}
	When a new ball $\ball(c, r)$ is added to the collection $\bset$, for any other ball $\ball(c', r')$ in $\balls$ with $\ball(c', r'/2)\cap \ball(c, r/2)\neq\emptyset$, $r' \leq r/4$ must hold.
\end{observation}
\begin{proof}
	Suppose otherwise that $r'> r/4$. As all radius are integral powers of $4$, $r' \geq r$. Therefore, $\dist_{G}(c, c')\leq r'/2 + r/2 \leq r'$, and so $c\in\ball(c', r')$, which means that $c$ was covered by the collection $\bset$ before the ball $\ball(c, r)$ is added, a contradiction.
\end{proof}

We say that a vertex $v$ is \emph{captured} by a ball $\ball(c, r)$ if $\dist_{G}(c,v)\le r/2$.
From \Cref{cover}, when $v$ is captured by a new ball $\ball(c, r)$, its radius $r$ is at least $4$ times the radius of any other ball in $\bset$ that captures $v$. As we have shown that the radius of every ball in $\bset$ is at least $R$ and at most $R\cdot 4^{5/\eps}$, the number of balls  in $\bset$ that captures $v$ is at most $5/\eps$.
\end{proof}

From \Cref{grow}, at the end of the algorithm, each vertex in $G$ is occupied by at most $O(1/\eps)$ balls in $\balls$. Therefore, $\sum|\ball(c, r/2)|\leq  O(n/\eps)$; and
$\sum|\ball(c, 4r)|\leq \beta\cdot\sum |\ball(c, r/2)|\leq O(n^{1+\eps}/\eps)$.
Similarly, $\sum\vol(\ball(c, 4r))\leq \beta\cdot\sum \vol(\ball(c, r/2))\leq \beta\cdot O(m/\eps) = O(m\cdot n^{\eps}/\eps)$.

Finally, note that when we add a ball $\ball(c, r)$ to $\bset$, the running time of the algorithm in that iteration is $O(\vol(\ball(c, 4r)))$. Therefore, algorithm terminates in time $O\big(\sum\vol(\ball(c, 4r))\big)\leq O(m\cdot n^{\eps}/\eps)$.

\end{document}